\documentclass[11pt]{article}

\usepackage{tikz}
\usepackage{ktmacros}
\usepackage[utf8]{inputenc}
\usepackage{amsmath}
\usepackage{amsthm}
\usepackage{amssymb}
\usepackage{color}
\usepackage{boxedminipage}
\usepackage{fullpage} 
\usepackage[utf8]{inputenc}
\DeclareMathOperator*{\E}{\mathbb{E}}

\usepackage{booktabs}  

\ifdefined\IsDraft
    \newcommand{\authnote}[2]{{\textbf [{\color{red} #1's Note:} {\color{blue} #2}]}}
		\newcommand{\changed}[1]{{\color{blue} #1}}
		\newcommand{\deleted}[1]{{\color{blue} ~Deleted:~{\color{red} #1}}}
		\newcommand{\TODO}[1]{{\color{red} TODO:} {\color{blue} #1}}
\else
    \newcommand{\authnote}[2]{}
		\newcommand{\changed}[1]{#1}
		\newcommand{\deleted}[1]{}
		\newcommand{\TODO}[1]{}
\fi

\newcommand{\edclose}{\thickapprox_{(\varepsilon,\delta)}}
\newcommand{\edcloseTwo}{\thickapprox_{(2\varepsilon,2\delta)}}
\newcommand{\myvec}[1]{{\bf #1}}

\newtheorem{theorem}{Theorem}[section] 
\newtheorem{lemma}[theorem]{Lemma}
\newtheorem{corollary}[theorem]{Corollary}
\newtheorem{proposition}[theorem]{Proposition}
\newtheorem{definition}[theorem]{Definition}
\newtheorem{claim}[theorem]{Claim}
\newtheorem{remark}[theorem]{Remark}
\newcommand{\abs}[1]{\left\lvert #1 \right\rvert}
\newcommand{\sta}{\stackrel{\mbox{\tiny S}}{{\equiv}}}

\newtheorem*{theorem*}{Theorem}
\newtheorem*{lemma*}{Lemma}
\newtheorem*{corollary*}{Corollary}
\newtheorem*{proposition*}{Proposition}
\newtheorem*{definition*}{Definition}
\newcommand{\N}{{\mathbb{N}}}

\newcommand{\Nt}{\mathbb{N}}

\newcommand{\GF}{\mathbb{GF}}
\newcommand{\bitset}{\{0,1\}}
\newcommand{\SD}{\operatorname{SD}}

\newcommand{\set}[1]{\left\{#1\right\}}
\newcommand{\Sim}{\mathsf{S}}

\newcommand{\Pc}{\mathbf{P}}
\newcommand{\Vc}{\mathbf{V}}

\usepackage[backend=biber, style=alphabetic, maxbibnames=99, maxalphanames=10]{biblatex}
\defbibenvironment{bibliography}
  {\list
     {\printtext[labelalphawidth]{%
        \printfield{prefixnumber}%
        \printfield{labelalpha}%
        \printfield{extraalpha}}}
     {\setlength{\labelsep}{\biblabelsep}%
      \setlength{\leftmargin}{\labelsep}%
      \setlength{\itemsep}{\bibitemsep}%
      \setlength{\parsep}{\bibparsep}}%
      }
  {\endlist}
  {\item}
  
\newif\ifconf
\conffalse

\title{PINE: 
Efficient Norm-Bound Verification for Secret-Shared Vectors}
\date{}
\author{
Guy N. Rothblum \\ Apple
\and
Eran Omri\thanks{Part of this work was performed while EO was at Apple.}\\
Ariel University\\ Ariel Cyber Innovation Center
\and
Junye Chen \\ Apple
\and
Kunal Talwar \\ Apple
}

\date{}
\begin{document}
\maketitle

\begin{abstract}
Secure aggregation of high-dimensional vectors is a fundamental primitive in federated statistics and learning. A two-server system such as PRIO allows for scalable aggregation of secret-shared vectors. Adversarial clients might try to manipulate the aggregate, so it is important to ensure that each (secret-shared) contribution is well-formed. In this work, we focus on the important and well-studied goal of ensuring that each contribution vector has bounded Euclidean norm. Existing protocols for ensuring bounded-norm contributions either incur a large communication overhead, or only allow for approximate verification of the norm bound. We propose {\bf P}rivate {\bf I}nexpensive {\bf N}orm {\bf E}nforcement (PINE): a new protocol that allows exact norm verification with little communication overhead. For high-dimensional vectors, our approach has a communication overhead of a few percent, compared to the 16-32x overhead of previous approaches.
\end{abstract}

\section{Introduction}
\label{sec:intro}

Data analyses and machine learning on user data have the potential to improve various applications, but 
might raise concerns about users' privacy. Fortunately, these analyses can be performed in a federated setting~\cite{kairouz2021advances, LiSTS20} while ensuring strong formal privacy guarantees. This has sparked significant interest in developing techniques and infrastructures to support private federated data analyses. In this work, we focus on a fundamental primitive: estimating the mean (or equivalently, the sum) of a set of high-dimensional vectors. This primitive is natural and important in its own right, and is also a basic building block for several tasks (e.g. \cite{Honeycrisp,Orchard}). Most notably, it is fundamental for federated optimization of machine learning models.
Furthermore, several other ML tasks such as PCA, $k$-means and EM can be reduced to (multiple applications of) aggregation.

Secure Aggregation techniques enable strong privacy guarantees without requiring trust in a single server. Various approaches have been studied for designing such systems. One approach is to have client devices run a secure multiparty computation to compute the aggregate, which has been explored in~\cite{ODOpaper,Bonawitz17,SoGA20,turboaggregate,BellBGLR20}. Another approach relies on two or more servers to perform the aggregation, using protocols that ensure that no single server can learn anything beyond the final aggregate (security holds so long as the two servers do not collaborate).  A common approach is for each client to secret-share its contribution between the two servers, and for the servers to use a secure protocol that allows them to learn the aggregate, but nothing more about any individual client's contribution. We refer to this approach as {\em distributed aggregation}. The PRIO system~\cite{Corrigan-GibbsB17} proposed a protocol for computing arbitrary functions over secret-shared data in the presence of two or more servers. 
This approach has been further developed in subsequent research~\cite{BonehBCGI19, BonehBCGI21, BellGGKMRS22, CastroP22, BonehBCGI23, TalwarEtAl23}, applied in practice~\cite{firefox, enpa} and forms the basis of an IETF~\cite{ppm} draft standard for aggregation. In this work, we focus on distributed aggregation of high-dimensional integer\footnote{While one is often interested in vectors of floating point values, standard techniques can translate the problem to fixed-precision vectors. Scaling then turns those into integer vectors from a bounded range. See the discussion in Sections \ref{subsec:our-work} and \ref{sec:performance}.} vectors.

Real-world implementations of secure aggregation must deal with adversarial clients attempting to manipulate the results. An important
advantage of the PRIO approach is that it allows the servers to validate certain properties of the client contributions, while preserving a zero-knowledge (ZK) property, meaning that the servers learn nothing beyond the aggregate and the fact that the client shares were valid. A rich line of research has focused on developing efficient methods to verify various properties of client contributions. In the context of summing up vectors in $\Re^d$, a natural objective is verifying that each client contribution has bounded Euclidean norm.  The Euclidean norm constraint arises naturally in many machine learning settings~\cite{dosovitskiy2021an, Dehgani+23, ZhaiLLBRZGS23}. 
Many model poisoning attacks (e.g. ~\cite{BaruchBG19,FangCJG20,ShejwalkarH21}) rely on clients submitting vectors of large norm (so called ``boosted gradients''). Recent works~\cite{ShejwalkarHKR22,SunKSM19} have shown that ensuring a bounded $\ell_2$ norm is effective against a large family of realistic poisoning attacks in the federated learning setting.\footnote{No method can completely prevent model poisoning~\cite{ShejwalkarHKR22}  (for a theoretical perspective, see e.g. \cite{MahloujifarM17}); bounding the $\ell_2$ norm limits the impact of adversarial clients.}

In this work, we focus on zero-knowledge protocols for verifying the bounded Euclidean norm property. Such protocols require the client and the servers to perform additional computations and communication, where different works obtain different overheads.  One important measure is the communication cost, especially in settings where the vectors involved can be of high dimensions (e.g. ~\cite{ramaswamy2019federated, paulik2021federated, Xu-icassp23, xu-etal-2023-federated, pelikan2023federated}). We review the works most related to our contributions.
 
One approach taken by prior work \cite{Corrigan-GibbsB17, BonehBCGI19} and in proposed implementations in the IETF standard~\cite{libprio,RohrigU23} has the client secret-share each bit of each coordinate of its vector as an element in a finite field. For $d$-dimensional vectors with $b$ bits of precision, this requires the client to send at least $bd$ {\em field elements} to each server. The field size itself must be rather large to keep the soundness error small. In high-dimensional settings, this translates to a substantial communication overhead. For example, in the implementation in \cite{RohrigU23}, the communication overhead is at least 16x (compared to secret sharing without validity proofs). We note that the servers also need to verify that the secret shares are indeed of bits, but the communication needed for this latter task can be smaller \cite{BonehBCGI19}.

In PRIO+ \cite{AddankiGJOP22}, the communication from the client to each server is only $bd$ {\em bits} (rather than field elements), but this requires an offline setup phase that involves expensive cryptographic operations (it also gives a weaker norm bound guarantee). In ELSA \cite{RatheeSWP22}, the setup avoids expensive cryptographic operations, but requires more communication from the clients. We elaborate on the comparison with these works below, but the main distinction is that our work focuses on a setting where there is no expensive setup, and we aim to minimize the communication from the clients to the servers.

With the goal of minimizing the client-to-server communication, the work of ~\cite{Talwar22} relaxed the zero-knowledge guarantee to differential ZK and showed a more communication-efficient protocol. However, that protocol only ensured approximate norm verification, limiting its usefulness in some settings.  The protocol allows for a trade-off between the soundness and completeness errors and the approximation guarantee, but even for relatively permissive error  thresholds (say, $0.01$ soundness and completeness errors), the protocol might accept vectors of norm 50x the target bound. This reduces its effectiveness in dealing with malicious clients.

\subsection{Our Work}
\label{subsec:our-work}

Our main contribution is PINE (for {\bf P}rivate {\bf I}nexpensive {\bf N}orm {\bf E}nforcement): a communication-efficient protocol for Euclidean norm verification that can verify the exact norm bound with a strong (statistical) zero knowledge guarantee and with no offline setup. Zero knowledge holds even against malicious behavior by a server (we always assume the servers do not collaborate, i.e. at least one of them is honest). Theoretically, our verification protocol requires the client to communicate  only $\tilde{O}(\sqrt{d})$ additional field elements. In practice, for typical parameter settings, this overhead is a small fraction of the communication needed to secret share $d$ field elements.

\begin{theorem}
[Informal version of~\cref{thm:norm-bound}] Let $X \in \mathbb{Z}^d$ be a secret-shared vector and $B \geq 0$. Fix $\rho >0$, and set $r = \ceil{32\ln \frac 1 \rho}$. For any field of size $q \geq \Omega(\max \{ B, 3r \})$ there is a distributed verification protocol with the following properties:


    \begin{enumerate}
        \item \textbf{Completeness:} If  $\sum_i X_i^2 \leq B$, the verifiers accept with probability at least $1-\rho$.

        \item \textbf{Soundness:} If $\sum X_i^2 > B$ (over the integers), the probability that the verifiers accept is at most $\rho + O(\sqrt{d}/q)$.

        \item \textbf{Zero-Knowledge}: The protocol  satisfies distributed statistical zero-knowledge: the view of each verifier can be simulated up to statistical distance $\rho$. This guarantee holds even for a single malicious verifier (so long as the other verifier follows the protocol).
    \end{enumerate}
    The proof system is in the common reference string model, and consists of a single message of length $O\left(\sqrt{d} \log q + r \log^2 q \right)$ sent by the prover to each verifier.
\end{theorem}

The protocol is derived from the interactive 4-message protocol of~\cref{thm:norm-bound} by using a variant of the Fiat-Shamir heuristic suited for the distributed verification setting ( \ifconf
\else
see \cref{sec:fiat-shamir}
\fi
).\footnote{The soundness error is proved for the interactive protocol, and should be set to be sufficiently small to account for the Fiat-Shamir transform.} While we state our results for the case of two verifiers, our approach is modular and can be used with multiple verifiers. Concrete performance evaluation is below and in Section \ref{sec:performance}.

\medskip\noindent{\bf Statistical PINE: Overview.}
A central difficulty in efficient verification of Euclidean norms is that PRIO relies on arithmetic over a finite field, whereas our property of interest deals with arithmetic over integers/reals. One can use known techniques~\cite{BonehBCGI19} to efficiently verify that the sum of of squared entries {\em modulo the field size} satisfies a certain bound. If the coordinates are all small enough, then the squared norm modulo the field sizes equals the squared norm over the integers, so the preceding check suffices. However, it is challenging to validate smallness of coordinates in their natural representation as field elements. Existing approaches achieve this by encoding the coordinates in their binary representation, which can enforce the requisite smallness.

We take a different approach. Instead of encoding the coordinates of a vector in their binary representation, we devise a randomized test that can detect whether there is a ``wraparound'' when computing the sum of squares over a finite field. We first verify that the sum of squared entries modulo the field size $q$ is in the range $[0,B]$. For large enough $q$, this gives us a promise problem: either the sum of squared entries (over integers) is at most $B$, or is at least $q$. We test for this {\em wraparound} by taking a random dot product with a ${-1,0,1}$ vector, and we demonstrate that if there is wraparound, the dot product is likely to be large. On the other hand, if the sum of the squared entries was small to begin with, then the dot product will be small with high probability. These bounds on the dot product, which are our main technical contribution, are proved using a delicate case analysis on the vector's infinity norm. The challenging case (small infinity norm) is analyzed using the Berry-Esseen theorem. 

Thus, we reduce the problem of proving a bound on the squared norm to proving boundedness of a few scalars (the outcomes of independent dot products), which can be done at a small overhead. Our communication overhead is dominated by the communication needed to verify the sum of squared entries over the finite field; the $\sqrt{d}$ term can be further reduced to $d^{1/c}$ by using an $c+O(1)$-round protocol (the additional interaction can be eliminated using the Fiat-Shamir transform).

\medskip\noindent{\bf Differential ZK.} We also show a simpler scheme that relaxes ZK to Differential ZK: intuitively, the secrecy of the client's contribution is protected via a {\em differential privacy} \cite{DworkMNS06} guarantee (rather than the perfect or statistical guarantees that are more common in the literature). Beyond its simplicity, the protocol also achieves smaller communication in some parameter regimes (especially when the number of dimensions is not too large). We note that the relaxation to differential ZK can be quite reasonable, since differential privacy is often all that is guaranteed given that the (approximate) results of the entire aggregation are to be made public to the servers.

\begin{theorem}
[Informal Version of~\cref{thm:dzk-norm-bound}] Let $X \in \mathbb{Z}^d$ be a secret-shared vector. Let $\eps,\delta \in (0,1)$ and $B \geq 1$. For a field of size $q > 4\left(\sqrt{B} + \sqrt{d} + \sqrt{dB\tfrac{2\ln 2.5/\delta}{\eps}\cdot(1+\frac{2\sqrt{\log 8e/\delta}}{\sqrt{d}} + \frac{2\log 8e/\delta}{d})}\right)^2$, there is a distributed verification protocol with:

  \begin{enumerate}
   \item \textbf{Completeness:} If  $\sum_{i=1}^d X_i^2 \leq B$, then the verifiers accept with probability $1$.

    \item \textbf{Soundness:} If $\sum X_i^2 > B$ (over the integers), the probability that the verifiers accept is at most $O(\sqrt{d}/q)$.

    \item \textbf{Zero-Knowledge}: The protocol  satisfies $(\eps,\delta)$-differential zero knowledge: the view of each verifier can be efficiently simulated up to $(\eps,\delta)$-closeness. This guarantee holds even for a single malicious verifier (so long as the other verifier follows the protocol).
\end{enumerate}
In addition to the secret shares of $x$, the client sends a proof of length $(O(\sqrt{d}) + O(\log_2 q))\cdot \lceil\log_2 q\rceil$.
\end{theorem}
\ifconf
\else
The protocol is derived from the interactive 3-message protocol of~\cref{thm:dzk-norm-bound} by using a distributed-verification variant of the Fiat-Shamir heuristic (see \cref{sec:fiat-shamir}).
\fi

\medskip\noindent{\bf Performance analysis.} We analyze the performance of our protocols in terms of the communication overhead, beyond the communication that is needed to simply send secret shares for distributed aggregation (without any robustness to poisoning attacks). In Section \ref{sec:performance}, we provide analyses for several choices of parameters. Here, in Table \ref{table:intro-performance}, we highlight the performance for a typical choice of parameters, where we aggregate $d$-dimensional integer vectors of $\ell_2$ norm at most $2^{15}$ with $d \in \{10^4, 10^5, 10^6, 10^7\}$. We work over a field of size $\approx 2^{64}$ and use soundness and zero-knowledge error $2^{-50}$. 

\ifconf
\begin{table*}[!t]
\else
\begin{table}[!t]
\fi
\centering

\begin{tabular}{|c || c | c |c | c|} 
 \hline
  & $d=10^4$ & $d=10^5$ & $d=10^6$ & $d=10^7$ \\ [0.5ex] 
 \hline\hline
 no robustness, $\#$ bits sent & $64 \cdot 10^4$ &  $64 \cdot 10^5$ & $64 \cdot 10^6$ & $64 \cdot 10^7$ \\ 
 \hline
 prior work, overhead \cite{BonehBCGI19,libprio} & $> 1500\%$ & $> 1500\%$ & $> 1500\%$ & $> 1500\%$ \\
 \hline
 PINE, Statistical ZK, overhead & $22\%$ & $3.18\%$ & $0.49\%$ & $0.13\%$ \\
 \hline
 PINE, Differential ZK, overhead & $4.77\%$  & $1.46\%$ & $0.32\%$ & $12.63\%$ \\ [1ex] 
 \hline
\end{tabular}

\caption{Communication analysis: our protocols and prior work. Parameters: field size $q \approx 2^{64}$ for aggregation, $d$-dimensional data, soundness error $2^{-50}$, zero-knowledge error $\delta=2^{-50}$. For differential ZK $\varepsilon=0.1$.}
\label{table:intro-performance}

\ifconf 
\end{table*} 
\else 
\end{table}
\fi

As seen in Table \ref{table:intro-performance}, our statistical zero-knowledge protocol achieves small overhead even when the number of dimensions is as small as $10^4$, and the communication overhead becomes negligible as the number of dimensions grows. Comparing with prior work \cite{BonehBCGI19,libprio}, the overhead is reduced by a multiplicative factor of between 70x (when the number of dimensions is only 10,000) to $10^4$x (when the number of dimensions is as high as $10^7$). 

Our differential ZK protocol achieves even smaller overheads so long as the number of dimensions is not huge (albeit, the secrecy guarantee is more relaxed). Once the number of dimensions grows to $10^7$, PINE with Differential ZK requires a larger field size (9 bytes instead of 8 bytes), this incurs an initial overhead of $~\frac{1}{8}$ for secret-sharing the data over a larger field. Previous work with differential ZK~\cite{Talwar22}, would have similar increased field size requirements for these parameters. While their protocol would not have any additional overhead, it only gives a weaker robustness guarantee that vectors that are 50x the norm bound are rejected with probability $0.99$.

We briefly elaborate on our choice of the norm bound: a typical setting when aggregating gradients is that we aggregate floating point vectors of Euclidean norm at most 1. We can convert these floating point vectors to integer vectors by multiplying by $2^b$ and rounding. With $b=15$ bits of precision per co-ordinate, this translates the problem to verifying $\ell_2$ norm bound $\sqrt{B} = 2^{15}$. This is typically sufficient for high dimensional vectors. Working over a field of size about $\approx 2^{64}$ suffices to allow aggregating millions of such vectors (whereas field size $\approx 2^{32}$ would not be sufficient for 100K vectors). We refer the reader to \cref{sec:performance} for further elaboration and for performance evaluation in other parameter regimes.

\medskip\noindent{\bf Further comparison to PRIO+ and ELSA.} In Table \ref{table:intro-qualitative-comparison}, we provide a qualitative comparison to the most closely related works to ours. In particular, we consider works in the distributed aggregation setting, with two non-colluding servers and privacy against (one out of two) malicious servers. The PRIO+ system \cite{AddankiGJOP22} reduces the communication cost of sharing the client's vector to $bd$ bits (where $2^b$ bounds the magnitude of each entry). This is significantly smaller than our protocol, which requires sending $d$ field elements (e.g. in Table \ref{table:intro-performance} each field element is 64 bits, whereas $b=15$). However, PRIO+ requires an expensive offline setup between the servers, which perform cryptographic oblivious transfer operations and exchange more communication than $d$ field elements per client. Moreover, verifying a bound on the Euclidean norm would also require a subsequent online cryptographic protocol (and there is further overhead for malicious-server security, see \cite{RatheeSWP22}). Other works, such as \cite{HaoLXC021,HeKJ20}, also make use of more advanced cryptographic operations. PINE avoids an expensive setup and cryptographic operations of this type. The ELSA system \cite{RatheeSWP22} also avoids the use of expensive cryptographic operations, but replaces them with communication from the client, which sends more than $d$ field elements to the servers. The expensive communication from the client to the servers can be performed offline, but it is a large communication compared to our work. Further, in many PRIO-like settings, anonymous clients engage in a one-shot interaction with the servers, so an offline setup is not appropriate. 

\ifconf
\begin{table}[!t]
\else
\begin{table}[!t]
\fi
\centering

\begin{tabular}{|c || c | c |c |} 
 \hline
  Protocol & 
  \begin{tabular}{@{}c@{}} Linear \\ online \end{tabular} & 
  \begin{tabular}{@{}c@{}} Linear \\ offline \end{tabular} & 
  \begin{tabular}{@{}c@{}} Expensive \\ crypto \end{tabular} 
  \\ [0.5ex]
  \hline\hline
  Prio3 \cite{BonehBCGI19,libprio} & N &  Y & N \\
  \hline
  Prio+ \cite{AddankiGJOP22} & Y & N & Y  \\
  \hline
  ELSA \cite{RatheeSWP22} & Y & N & N  \\
  \hline
  PINE (our work) & Y & Y & N  \\
  [1ex] 
 \hline
\end{tabular}

\caption{Qualitative comparison to representative prior works. All works are in the distributed two-server trust model. {\em Linear offline and online} mean that the communication is dominated by the cost of communicating at most $d$ field elements for the field where aggregation should occur, either in an offline stage (before the client's contribution is specified), or in the online stage (respectively). {\em Expensive crypto} means that the protocol uses operations such as oblivious transfer. 
}
\label{table:intro-qualitative-comparison}

\ifconf 
\end{table} 
\else 
\end{table}
\fi

\medskip\noindent{\bf Further related work.} In a very recent work Boneh {\em et al.} \cite{BonehBCGI23} show how {\em arithmetic sketching schemes} can be used to design efficient protocols for verifying properties of secret-shared data. They show, however, that the linear sketches at the heart of their technique cannot be used to verify $L_2$-norm constraints (or any $L_p$ norm for $p > 1$).

\section{Model, Definitions  and Preliminaries}
\label{sec:prelims}
\begin{definition}\label{def:statDist}
	The statistical distance between two finite random variables $X$ and $Y$ is
	$$\SD(X,Y)=\frac12\sum_{a}\abs{\Pr[X=a]-\Pr[Y=a]}.$$

	Two ensembles $X=\set{X_{n}}_{n\in\N}$ and $Y=\set{Y_{n}}_{n\in\N}$ are
 said to be statistically close, denoted $X\sta Y$,
 if there exists a negligible function $\mu(\cdot)$, such that for all $n\in\N$, it holds that
 \ifconf
\else
	$$\SD\left({X_{n},Y_{n}}\right)\leq\mu(n).$$
 \fi
\end{definition}

\medskip
\noindent{\bf Runtimes and field operations over $\GF[q]$.} We measure the prover's and the verifiers' runtimes by the number of field operations we perform: we usually count addition and multiplication as a single field operation, and also allow other basic atomic operations such as translating field elements to their natural representation as bit vectors and vice versa. In performance evaluations, we also measure the number of multiplications as a primary complexity measure (as done in prior work, since these are significantly more expensive than other field operations).

\medskip\noindent{\bf Secret sharing over $\GF[q]$.} In the distributed verification setting we study, a client (or prover) secret-shares data between two or more servers (or verifiers). Each secret-shared value is an element $\alpha \in \GF[q]$ from a field, and the secret shares are random field elements whose sum is $\alpha$ (``arithmetic shares''). In particular, each server's share is, on its own, a uniformly random field element.

\newcommand{\pubParam}{\bf\bar{T}}
\newcommand{\privInp}{\bf\bar{I}}
\newcommand{\privOut}{\bf\bar{O}}

\subsection{Distributed Verification Protocols}
\label{subsec:distributed-verification}

As discussed above, we study a distributed model, in which a single prover (or client) $\Pc$ interacts with two verifiers (or servers) $\Vc_0, \Vc_1$ over a complete network with secure point-to-point channels. We design several protocols in this setting, and then compose these protocols to obtain the PINE proof systems. Our distributed verification protocols have inputs and outputs of the following type:

\begin{description}
    \item[Common (Public) Inputs:]
     All three parties hold the same common and publicly known  parameter vector $\pubParam$. Such parameters may include, field size, scalars, protocol parameters.
     
     \item[Private Inputs:]
     The prover holds some vector of input values $\privInp$. Each such value $I$ is secret-shared by the two verifiers such that $\Vc_j$ holds the share $[I^{(j)}]$.

     \item[Common (Public) Output:]
     The common output of the two verifiers includes a single bit, indicating whether they accepted or rejected the proof. The verifiers can also output additional common outputs.
     
     \item[Private Outputs:]
     At the end of the interaction, the prover outputs new private values.  Each such value $\privOut$ is secret-shared by the two verifiers, and these are called the shared outputs of the verifiers.
\end{description}

We define distributed interactive zero-knowledge proof protocols (dZKIPs) for this setting. These protocols guarantee  completeness and soundness based on a condition on the common and private inputs (this ``input condition'' is formalized by requiring that the combined input is in a (pair) language). The zero knowledge property is also guaranteed for inputs that satisfy this input condition (are in the language). Some of the protocols also have an ``output condition'' on the common and private outputs (also formalized as membership in an output pair language, this condition can be empty). Completeness means that if the inputs satisfy the input condition and all parties follow the protocol, then  w.h.p. the verifiers should accept and the common and private outputs (if any) should be in the output language. Soundness means that if the input condition is violated, then w.h.p. either the verifiers reject (in their common output), or the output condition is violated. Note that the verifiers might not ``know'' that the output condition is violated (since they only see secret shares of the private outputs): in our work this will be detected by a subsequent protocol. Allowing for private inputs and outputs, and for general input and output conditions is helpful for designing modular sub-protocols that can later be composed.

Our protocols also guarantee a strong zero-knowledge property: so long as the input is in the input language, and the verifiers are honest, both verifiers learn nothing from executing the protocol. Formally, each verifier's view in a protocol execution with the honest prover and a second verifier who follows the protocol can be simulated efficiently (the view includes the common inputs, its secret shares of the private inputs, random coins, messages received, and its secret shares of private outputs). We remark that in the classical setting for (non-distributed) interactive zero-knowledge proofs, a single prover $\Pc$ interacts with a single verifier $\Vc$ over a common input $X$. In such an interaction, the aim of the prover is to convince the verifier that $X \in L$ for some language $L$, where the verifier knows what $X$ is, and zero-knowledge means that the verifier should not learn anything beyond $X$ and the fact of its membership in $L$.  In contrast, in our (distributed) setting, the statement $X$ is not known to any single verifier (since it is secret-shared), and the zero-knowledge requirement means that any single verifier should learn nothing beyond $X$'s membership in $L$.

\begin{definition}[Distributed ZK Interactive Proof (dZKIP)]
\label{def:dZKIP}

We say that a $2$-verifier interactive proof protocol $\Pi = (\Pc; \Vc_0; \Vc_1)$ is a distributed (strong) zero-knowledge proof for an input (pair) language $\Linp$ and an output language $\Lout$ 
if $\Pi$ satisfies the following:

\begin{itemize}

\item{\textbf{$\alpha$-Completeness.}} If the common and private inputs are in the input language, i.e. $(\pubParam,\privInp) \in \Linp$ and the prover and the verifiers follow the protocol, then with all but $\alpha$ probability the verifiers accept and the private outputs satisfy the output condition, i.e. $\privOut \in \Lout$ (the probability is over all coins tossed by all parties in the protocol).

\item{\textbf{$\beta$-Soundness.}} If the common and private inputs are {\em not}  in the language, i.e. $(\pubParam,\privInp) \notin \Linp$, then for any adversarial cheating prover strategy, with all but $\beta$ probability, either the verifiers reject, or the private outputs violate the output condition, i.e. $\privOut \notin \Lout$ (the probability is over the verifiers' coin tosses. The cheating prover is deterministic w.l.o.g).

\item{\textbf{$\gamma$-Strong Distributed Honest-Verifier Zero-Knowledge (dZK)} (see \cite{BonehBCGI19}).} There exists an efficient simulator $\Sim$, such that for every input pair $(\pubParam,\privInp) \in \Linp$, for every $j \in \bitset$
and every 2-out-of-2 sharing   $[\privInp]=\left(\privInp^{(0)},\privInp^{(1)}\right)$, the view of $\Vc_j$ in an execution with the honest prover and $\Vc_{1-j}$ is $\gamma$-statistically close to the output of the simulator $\Sim$ on input $\left(j, \pubParam, \privInp^{(j)}\right)$.

\end{itemize}

By default we take $\alpha=0$ (perfect completeness) unless we explicitly note otherwise. Similarly, by default $\gamma=0$ (perfect distributed zero-knowledge).
We say the protocol is {\em public-coins} if all the messages sent from the verifiers to the prover are random coin tosses. We also require that in a public-coins protocol, the communication between the two verifiers consists solely of a single simultaneous message exchange: after the interaction with the prover is complete, $\Vc_0$ sends a single message to $\Vc_1$, and (at the same time) $\Vc_1$ sends a single message to $\Vc_0$. These messages should not depend on each other (this property is important for zero-knowledge of the Fiat-Shamir transform, see 
\ifconf
\else
\cref{sec:fiat-shamir}).
\fi
\end{definition}

We sometimes consider dZKIPs for {\em promise problems}, where the completeness condition and the soundness condition apply to disjoint sets (rather than to the language $\Linp$ and its complement). See, for example, the ``wraparound protocol'' of Section \ref{subsec:wraparound-protocol}.




\begin{remark}[indistinguishability under varying private inputs]
\label{remark:varying-private-inputs}

The dZK property implies that, for a fixed public input $\pubParam$, we can consider different private inputs that are in the pair language, and each verifier's views will be $\beta$-statistically close under these varying inputs, so long as its share remains unchanged. For example, fixing $\Vc_0$'s share to be $\privInp^{(0)}$, we can consider an execution where $V_1$'s share is $\privInp^{(1)}$ and another execution where $V_1$'s share is ${\privInp}^{' (1)}$, and $\Vc_0$'s views in these executions will be statistically close, so long as the underlying private inputs $\privInp,\privInp'$ are both in the pair language (w.r.t. the fixed public input $\pubParam$).
\end{remark}

\begin{remark}[Non-interactive and malicious ZK via Fiat-Shamir]
\label{remark:fiat-shamir}

We construct and analyze our protocols as interactive proof systems with honest-verifier zero-knowledge guarantees, as formalized in Definition \ref{def:dZKIP}. These protocols can be transformed to non-interactive protocols that guarantee zero-knowledge against (one out of two) malicious verifiers using the Fiat-Shamir transform for the distributed setting.
\ifconf
\else
(See \cref{sec:fiat-shamir} for further details).
\fi
\end{remark}

\subsection{Composition of dZKIPs}
\label{subsec:composition}

Distributed ZKIP protocols maintain their zero-knowledge properties under (sequential) composition, so long as the protocols have the property that for every common input, for every input share there exists a completion of the share to an input on which the verifiers accept. 

\begin{lemma}[ZK composition]
\label{lemma:composition}

Let $\Pi,\Pi'$ be dZKIP protocols (see Definition \ref{def:dZKIP}) that are run sequentially,
where the secret shares of the private input to $\Pi'$ can be efficiently computed from the shares of the private inputs and private outputs of $\Pi$ (each verifier can compute its share of the private input to $\Pi'$ on its own), and where the public inputs to $\Pi'$ can be efficiently computed from the public inputs and outputs of $\Pi$.  Suppose that:
\begin{enumerate}
    \item $\Pi$ is $\gamma$-dZK and $\alpha$-complete. $\Pi'$ is $\gamma'$-dZK.
    
    \item If $\Pi$'s inputs and outputs are in that protocol's input and output languages ($\Linp$ and $\Lout$), respectively, then they specify inputs for $\Pi'$ that are in $\Pi$'s input language $\Linp'$.
\end{enumerate}

Consider the composed protocol $(\Pi \circ \Pi')$, which runs $\Pi$, and then uses the resulting private output to specify inputs to an execution of $\Pi'$. Then the composed protocol $(\Pi \circ \Pi')$ satisfies $(\alpha + \gamma+\gamma')$-dZK.

\end{lemma}

\ifconf
\else
\begin{proof}[Proof sketch.]
We show the simulator for $\Vc_0$ (the simulator for $\Vc_1$ is symmetric). We begin by running the simulator $\Sim$ for $\Pi$, obtaining a view that is statistically close to the view in an actual execution. The simulated view from $\Pi$'s execution specifies the public input for $\Pi'$, and also  $\privInp^{(0)}$: $\Vc_0$'s share of the private input. We then run the simulator $\Sim'$ for $\Pi'$ on these inputs, to produce a simulation for $\Vc_0$'s view in $\Pi'$, and output the composed simulated views.

To analyze this simulator, consider first running $\Sim'$ on outputs of a real execution of $\Pi$. By $\Pi$'s completeness, with all but $\alpha$ probability the inputs to $\Pi'$ are in its input language. When this is the case, if we run the simulator for $\Pi'$ on the outcomes of the real execution, the resulting simulated view is $\gamma'$-close to the real view in $\Pi'$ (conditioned on the preceding run of $\Pi$). Thus the experiment of running $\Sim'$ on a real view generated by $\Pi$ gives a view that is $(\alpha + \gamma')$-close to the view in the real execution. Now we can replace the input to the $\Sim'$ by a simulated execution of $\Pi$. By $\Pi$'s dZK guarantee, this ``input'' simulated view is $\gamma$-close to the view in a real $\Pi$-execution. By a hybrid argument, running the simulator $\Sim'$ on the view that was output by $\Sim$ will give a complete view of the composed protocol that is $(\gamma+\alpha+\gamma')$ close to the real view.
\end{proof}
\fi

\subsection{An Example: Verifying Linear Equalities}
\label{subsec:dZKIP-example}

We next provide an example for a distributed zero-knowledge protocol, which allows a prover to prove to two verifiers that some linear equality holds with respect to the shared inputs that the verifiers hold. The protocol is very simple and requires no interaction with the prover. Nevertheless, this protocol will be useful as a sub-protocol in our construction. 

\begin{example}[A protocol for linear equality]\label{example:linearEquality}
\textbf{Common inputs:} Field size $q \in \Nt$, Dimension  $d \in \Nt$, coefficients $\alpha_1,\ldots,\alpha_d, 
    \in \GF[q]$, a field element $z\in \GF[q]$.\\

   \noindent\textbf{Secret-shared inputs:} A vector $X \in \GF[q]^d$, where each $X_i$ is secret-shared as $[X_i] = (X_i^{(0)}, X_{i}^{(1)})$.
    The prover knows the secret-shared values $X$ and the shares $[X]$.
    The verifiers each know their own shares (respectively $X^{(0)}$ and $X^{(1)}$). \\

     \noindent\textbf{Claim (to be verified):} $\sum_{i=1}^d \alpha_i \cdot X_i = z$. I.e. the input language is $\Linp = \{((q,d,\{\alpha_i\},z),X) : \sum_i \alpha_i \cdot X_i = z\}$, there are no private outputs (or output language). \\


    \noindent\textbf{The Protocol:}
    For $j\in\{0,1\}$, verifier $\Vc_j$ (locally) computes $z_j = \sum_{i=1}^d \alpha_i \cdot X_i^{(j)}$
    and sends $z_j$ to $\Vc_{1-j}$. Both verifiers accept if $z=z_0+z_1 \pmod{q}$ and reject otherwise.

\noindent\textbf{Properties of the protocol:}
This protocol offers perfect completeness and perfect soundness in the sense that the verifiers accept if and only if $$z = \sum_{i=1}^d \alpha_i \cdot (X_i^{(0)}+X_i^{(1)}) \pmod{q}.$$
To see that the protocol is perfect zero-knowledge, observe that for every input share $X^{(j)}$ for verifier $\Vc_j$, the simulator can compute $z_j = \sum_{i=1}^d \alpha_i \cdot X_i^{(j)}  \pmod{q}$, and send to $\Vc_j$ the value $z-z_j \pmod{q}$ as its only message from $\Vc_{1-j}$.

\end{example}

\section{Norm Verification}
\label{sec:norm-protocol}

We build our protocol in multiple steps. Some of our sub-protocols output secret shares of values that are then checked in subsequent sub-protocols, e.g. checking that the secret-shared values are bits (i.e. 0 or 1).

In~\cref{subsec:inequality} we show a protocol that reduces checking that a secret-shared value is in some range to: $(i)$ checking a linear equality over secret-shared values, and $(ii)$ checking that shares generated in the protocol are secret shares of bits. In~\cref{subsec:wraparound-protocol}, we build on this to construct our main contribution: a protocol for verifying that the sum of secret-shared values is not larger than the field size (the sum is taken over the integers, not over the field). In particular, this protocol lets us reduce certifying a norm bound of secret-shared values to certifying quadratic constraints.
We recall the appropriate form for quadratic constraint validation in~\cref{subsec:quadratic} and combine these ingredients to derive our main result in~\cref{subsec:putting-together}.

\subsection{Range-Check Subprotocol}
\label{subsec:inequality}

We present a simple 1-message protocol that is helpful in verifying inequalities of the form $\sum_i \alpha_i Q_i \in [\beta_1,\beta_2] \pmod{q}$,
where the coefficients $\alpha_i \in \GF[q]$ and the lower and upper bounds $\beta_1,\beta_2 \in \GF[q]$ are known and public (we view elements of $\GF[q]$ as integers in the set $\{-\lfloor q/2 \rfloor,\ldots, \lfloor q/2 \rfloor \}$), but the $Q_i$'s are known only to the prover, not to the verifiers. In particular, the $Q_i$'s are either secret-shared between the verifiers, or they are functions of secret-shared values (e.g. when we are verifying the inequality $\sum_i X_i^2 \leq B \pmod{q}$ over the secret-shared values $X_i$). The protocol reduces this verification task to: $(i)$ the verification of an {\em equality} modulo $q$, an easier claim to deal with, and $(ii)$ verifying that new secret-shared values (generated in the course of the protocol) are secret shares of bits.

\paragraph{Protocol overview.} 
The prover computes $V = ((\sum_i \alpha_i Q_i) - \beta_1)$, and secret-shares the bits $\{v_j\}$ of $V$'s binary representation. It also computes $U = (\beta_2 - (\sum_i \alpha_i Q_i))$ and secret-shares the bits $\{u_j\}$ of $U$'s binary representation. 
Observe that $U,V \in [0,(\beta_2 - \beta_1)]$, and thus $U$ and $V$ can be represented using $b = \lceil \log (\beta_2 - \beta_1 +1) \rceil$ bits. The verifiers get secret shares of these bits and verify that the sum of the values they represent is correct, i.e. that $(\sum_j v_j \cdot 2^j)+ (\sum_j u_j \cdot 2^j)= (\beta_2 - \beta_1) \pmod{q}$ (recall that linear equalities over secret-shared values are easy for the verifiers to check). If this holds, and also: $(i)$ the $v_j$'s and the $u_j$'s are all secret shares of bits, $(ii)$  $\sum_j v_j \cdot 2^j = (\sum_i \alpha_i Q_i) - \beta_1$, and $(iii)$ $q > 3(\beta_2 - \beta_1)+2$, then indeed it must be the case that $\sum_i \alpha_i Q_i \in [\beta_1,\beta_2]$. Note that if $Q_i$'s are themselves secret-shared values, then the verifiers can verify that condition $(ii)$ holds, and all that remains is to verify that the secret-shared values are indeed of bits (but we will also use this protocol in situations where the verifiers do not have secret shares of the $Q_i$ values).

The protocol is in Figure \ref{fig:inequality-mod-q-protocol}, its complexity and guarantee are in Lemma \ref{lemma:range-protocol}.

\ifconf\begin{figure*}[t]
  \thisfloatpagestyle{empty}
\else \begin{figure*}[t]
\fi
  \begin{boxedminipage}{\textwidth}
    \small \medskip \noindent

    \underline{\textbf{Protocol: Range Check $\pmod{q}$}}\\

    \textbf{Common inputs:} Field size $q \in \Nt$, number of variables $n \in \Nt$, coefficients $\alpha_1,\ldots,\alpha_n \in \GF[q]$ and claimed lower and upper bounds $\beta_1,\beta_2 \in \GF[q]$, viewed as integers in $\{-\lfloor q/2 \rfloor,\ldots, \lfloor q/2 \rfloor \}$, s.t. $\beta_1 \leq \beta_2$ and $q > 3(\beta_2 - \beta_1)+2$.\\

    \textbf{Other inputs:} The prover knows $Q_1,\ldots,Q_n \in \GF[q]$. We do not assume the verifiers have access to these $Q_i$'s. \\

    \textbf{Secret-shared outputs:} Shares $\{[v_j],[u_j]\}_{j=0}^{b-1}$, where $b = \lceil \log (\beta_2-\beta_1+1) \rceil$ (for each $j$, each verifier outputs its respective shares, $(v_j^{(0)},u_j^{(0)})$ or $(v_j^{(1)},u_j^{(1)})$). \\



    \textbf{The Protocol:}

    The prover secret shares:
    \begin{enumerate}

    \item The $b$ bits $(v_j)_{j \in [0,\ldots,b-1]}$ of $V = (\sum_i \alpha_i Q_i) - \beta_1 \pmod{q}$,

    \item The $b$ bits $(u_j)_{j \in [0,\ldots,b-1]}$ of $U = \beta_2 - (\sum_i \alpha_i Q_i) \pmod{q}$.

    \end{enumerate}

    The verifiers verify the linear equality $(\sum_{j=0}^{b-1} v_j \cdot 2^j)+ ( \sum_{j=0}^{b-1} u_j \cdot 2^j) = \beta_2 - \beta_1 \pmod{q}$  (rejecting otherwise).

  \end{boxedminipage}

  \caption{Range Check $\pmod{q}$ Protocol}
  \label{fig:inequality-mod-q-protocol}
\end{figure*}


\begin{lemma}
\label{lemma:range-protocol}

    The protocol of Figure \ref{fig:inequality-mod-q-protocol} satisfies:

    \begin{enumerate}
        \item \textbf{Completeness:} If $\sum \alpha_i Q_i \in [\beta_1, \beta_2] \pmod{q}$ (this is the input condition), then the verifiers accept and it holds that: $(i)$ $(\sum_{i=1}^n \alpha_i Q_i) - \beta_1 = \sum_{j=0}^{b-1} v_j \cdot 2^j \pmod{q}$, and $(ii)$ the secret-shared values are bits: $\forall j: \, v_j,u_j \in \bitset$ ($(i)$ and $(ii)$ are the output conditions).

        \item \textbf{Soundness:} If $\sum_{i=1} \alpha_i Q_i \notin [\beta_1,\beta_2] \pmod{q}$, if the verifiers do not reject, then either: $(i)$ $(\sum_{i=1}^n \alpha_i Q_i) - \beta_1 \neq \sum_{j=0}^{b-1} v_j \cdot 2^j \pmod{q}$, or $(ii)$ for some $j$, either $v_j$ or $u_j$ is not in $\bitset$.

        \item \textbf{Zero-Knowledge}:The protocol is strong zero-knowledge as per Definition~\ref{def:dZKIP}.\footnote{ Technically, this protocol does not fall into our distributed verification model, as the verifiers do not hold secret shares of the $Q_i$'s. 
        Indeed, here the simulator can create its view without {\em any} access to shares of the $Q_i$'s.}
    \end{enumerate}

    The proof consists of a single message of length $(2\lceil \log(\beta_2 - \beta_1 +1) \rceil \cdot \log (q)$  from the prover to each verifier (this message contains the secret shares of the bits $\{v_j,u_j\}$). The prover performs $O(n + \log(\beta_2-\beta_1 + 1))$ field operations. The verifiers each perform $O(\log(\beta_2-\beta_1 + 1))$ field operations and communicate $\log(q)$ bits between themselves.

\end{lemma}

The proof of Lemma \ref{lemma:range-protocol} is omitted: it follows from the construction. The simulator generates dummy shares of the bits $v_j$ and $u_j$, and simulates receiving a message in the protocol that verifies the linear equality, as in Example \ref{example:linearEquality}.

\begin{remark} \label{remark:inequality-efficiency}
  We note that if $(\beta_2-\beta_1+1)$ is a power of two, we can simplify this protocol by skipping sending $U$ and skipping the equality check. Indeed the fact that $V = (\sum_i \alpha_i Q_i) - \beta_1$ can be represented as $\sum_{j=0}^{b-1}v_j 2^j$, for bits $v_j$, is equivalent in this case to $V \in [0,2^b-1] = [0, \beta_2-\beta_1]$.
\end{remark}

\subsection{Detecting Wraparound}
\label{subsec:wraparound-protocol}

Our main technical contribution is a protocol that allows the servers to verify that the sum of the squares of the secret-shared values is not larger than the field size (where the sum is taken over the integers), i.e. that there is no ``wraparound'' when we take the sum of squares modulo the field size. The guarantees and complexity of the protocol are in Lemma \ref{lemma:wraparound-protocol}. The protocol itself is in Figure \ref{fig:wraparound-protocol}.

In more detail: let $q$ be the size of the field, constraints on which will be determined below. We view each $X_i$ as an integer whose value is in the range $\{0, \ldots, q-1 \}$. In what follows, all equalities and inequalities are over the integers unless we explicitly note otherwise (using and abusing the $\pmod{q}$ notation). The protocol considers a {\em promise problem}: in the YES case we have $\|X\|_2^2 \leq B$, whereas in the NO case we have $\|X\|_2^2 \geq q$ (i.e., there is wraparound). The verifiers accept or reject, and they also output secret shares $\{[w_j]\}$, which should be shares of bits (i.e. it should be true that $w_j \in \bitset$, and this always holds in the YES case assuming the prover follows the protocol). Soundness guarantees that in the NO case, the probability that the verifiers accept {\em and} the output shares are all of bits is small. The verifiers will run a subsequent sub-protocol to verify that all $w_j$'s are shares of bits.

\paragraph{Protocol overview.} The verifiers pick a random vector $Z \in \{-1,0,1\}^d$, where each $Z_i$ is  $-1$ w.p. $1/4$, $1$ w.p. $1/4$, and $0$ w.p. $1/2$.
The verifiers send
$Z$ to the prover, and verify that for a constant $\alpha > 0$ (set below, where we assume $q > 2\alpha \sqrt{B}$):
\begin{align} \label{eq:wraparound-test}
    \sum_i Z_i X_i & \in \left[ -\alpha \sqrt{B} , \alpha \sqrt{B} \right] \pmod{q}.
\end{align}
This verification is performed using the protocol of Figure \ref{fig:inequality-mod-q-protocol}. We show that for an appropriate choice of $\alpha$ there is a small constant $\eta$ s.t.~in the YES case ($\|X\|^2_2 \leq B$) Equation \eqref{eq:wraparound-test} holds with probability at least $1-\eta$ over the choice of $Z$. In the NO case, the probability is at most $1/2$.  We can repeat the test in parallel to make the soundness error negligible. This effectively reduces our original problem to range-checks (and, via the protocol of Figure \ref{fig:inequality-mod-q-protocol}, to bit-checks). There is, however, a zero-knowledge issue: there's a non-negligible failure probability in the YES case, and the verifiers see whether the test failed or not (and the values $Z$ that led to failure), which leaks information about $X$.

We could resolve this issue by increasing the field size to the point where the failure probability in the YES case becomes negligible, but this would entail a significant cost in the communication complexity required for sending field elements (we remark that increasing the field size doesn't reduce the soundness error). Instead, we take advantage of the fact that the protocol will be repeated many times (for soundness). The verifiers will not learn whether any of the individual repetitions succeeded, but only whether ``many'' of them succeeded. The threshold for ``many'' is set by a parameter $\tau$ indicating the fraction of repetitions that should succeed (by default we set $\tau=3/4$, so the verifiers accept if and only if at least three quarters of the repetitions succeed). In the YES case this will happen with all but negligible probability (so we get statistical zero-knowledge). To accomplish this, we repeat the test $r$ times, where in the $k$-th repetition the prover secret-shares a bit $g_k \in \bitset$, indicating whether or not Equation \eqref{eq:wraparound-test} holds with respect to the $k$'th test. The verifiers multiply both sides of the equation by $g_k$, so that if $g_k=0$ the test passes even though the equation doesn't hold. Later, the verifiers also verify that $\sum_k g_k \geq \tau \cdot r$ (the inequality can be verified in zero-knowledge using the protocol of Figure \ref{fig:inequality-mod-q-protocol}). Checking that Equation \eqref{eq:wraparound-test} holds after multiplying by $g_k$ is done using the protocol of Section \ref{subsec:quadratic} for checking quadratic constraints over secret-shared values.

We make two minor modifications to improve the protocol's efficiency: First,  the cost of the quadratic constraint protocol grows with the (square root of the) number of summands. Thus, for each iteration, each verifier computes (on its own) its secret share for the sum $\sum_i Z_i X_i$ (a linear function of the secret-shared $X_i$'s). This gives the verifiers a single secret-shared field element that should be multiplied by $g_k$ (instead of $d$). Indeed, in the full protocol, we fold further summands into this secret-shared sum. The second minor efficiency improvement we use is having the honest prover set {\em exactly} a $(1-\tau)$-fraction of the bits $g_k$ to 0 (w.h.p. this means that some repetitions where Equation \eqref{eq:wraparound-test} holds will not be checked). This allows the verifiers to replace checking $\sum_k g_k \geq \tau \cdot r$ (an inequality) with a simpler equality check $\sum_k g_k = \tau \cdot r$. 

\begin{figure*}[!t]
  \begin{boxedminipage}{\textwidth}
    \ifconf
    \else
    \small 
    \fi
    \medskip \noindent

    \underline{\textbf{Wraparound Detection Protocol}}\\

    \textbf{Common inputs:} Dimension $d \in \Nt$, claimed bound $B \in \Nt$, field size $q \in \Nt$, number of repetitions $r \in \Nt$, completeness error per repetition $\eta \in [0,1]$, threshold $\tau \in (1/2,1]$ for successful repetitions, s.t. $\tau \cdot r$ is an integer. \\

    \textbf{Secret-shared inputs:} A vector $X \in \GF[q]^d$, where each $X_i$ is secret-shared as $[X_i] = (X_i^{(0)}, X_{i}^{(1)})$.

    The client (prover) knows the secret-shared values $X$ and the shares $[X]$.

    The servers (verifiers) each know their own shares (respectively $X^{(0)}$ and $X^{(1)}$). \\

    \textbf{Secret-shared outputs:} $\{[g_k],[S_k],[v_{k,j}], [u_{k,j}]\}_{k \in [r], j \in [b]}$ where $b= \left( \left \lceil \log(2\alpha \sqrt{B} + 1) \right \rceil \right)$. \\

    \textbf{The Protocol:}

    Fix $\alpha = \sqrt{ \ln(2/\eta)}$. Let $D_Z$ be the distribution that samples a $d$-dimensional vector where for each $i \in [1,\ldots, d]$, $Z_i$ is drawn independently to be $-1$ w.p. $1/4$, 0 w.p. $1/2$ and $1$ w.p. $1/4$

    \begin{enumerate}

    \item The following test is repeated in parallel $r$ times, where in the $k$-th repetition:

    \begin{enumerate}

    \item  The verifiers choose a random string $Z_k \sim D_Z$ and send it to the prover.

    The prover computes $Y_k = \sum_{i=1}^d Z_{k,i} X_i$.

    \item If  $Y_k \in [-\alpha \sqrt{B}, \alpha\sqrt{B}]$, the prover sets $g_k=1$
    and uses the protocol of Figure \ref{fig:inequality-mod-q-protocol} to prove that $(\sum_i Z_{k,i} X_i) \in [-\alpha \sqrt{B}, \alpha\sqrt{B}]$. Let $\{v_{k,j}, u_{k,j}\}_{j \in [b]}$ be the secret-shared outputs of that protocol.

    Otherwise ($Y_k \notin [-\alpha \sqrt{B}, \alpha\sqrt{B}]$), the prover sets $g_k = 0$, sets the bits $\{v_{k,j}\}$ to be the bit representation of the value $2\alpha \sqrt{B}$, and sets the bits $\{u_{k,j}\}_{j \in [b]}$ to all be 0.
    Note that this setting ensures that the linear check in the protocol of Figure \ref{fig:inequality-mod-q-protocol} (the last step of that protocol) succeeds.

    \end{enumerate}

    \item If there are more than $(1-\tau) \cdot r$ repetitions in which $g_k=0$ then the prover aborts (and the verifiers reject).

    Otherwise, the prover sets $g_k=0$ for as many of the repetitions as needed to ensure there are {\em exactly} $(1-\tau) \cdot r$ such repetitions, and sends secret shares of $\{g_k\}_{k \in [r]}$.

    \item For each $k \in [r]$, the verifiers will check that either $g_k=0$ or the $v_j$'s satisfy soundness condition $(i)$ of the protocol of Figure \ref{fig:inequality-mod-q-protocol} (see Lemma \ref{lemma:range-protocol}). Towards this, 
    each verifier computes its a share of:
    \begin{align*}
     S_k = \left( \sum_{i=1}^d Z_{k,i} X_i \right)  + \alpha\sqrt{B}  - \left( \sum_{j \in [b]} 2^j \cdot v_{k,j} \right) \pmod{q},
    \end{align*}
     
    Note that the sum $S_k$ is a linear function of secret-shared values, so each verifier can indeed compute its share on its own. The protocol of Section \ref{subsec:quadratic} will later be used to check the quadratic equality $g_k \cdot S_k = 0 \pmod{q}$ (see the protocol of Figure~\ref{fig:L2-protocol}).

    \item The verifiers verify the linear equality $\sum_{k \in [r]} g_k = \tau \cdot r \pmod{q}$ (otherwise they reject).

    \end{enumerate}
  \end{boxedminipage}

  \caption{Wraparound Detection Protocol}
  \label{fig:wraparound-protocol}
\end{figure*}

\begin{lemma}
\label{lemma:wraparound-protocol}

    Fix a bound $B$, a number of repetitions $r$, a desired completeness error for each repetition $\eta \in [0,1]$ and a threshold $\tau \in (1/2,1]$ s.t. $\tau \cdot r$ is an integer. Let the field size $q$ be at least $\max \{  \frac{B \ln(2/\eta) }{4000}, 2600 \sqrt{B \ln(2/\eta)}, 2r\}$.

    Let $\Bin(\ell;r,p)$ denote the probability that the Binomial distribution with parameters $r$ and $p$ has outcome (number of successes) at least $\ell$. The protocol of Figure \ref{fig:wraparound-protocol} has the following properties:

    \begin{enumerate}
        \item \textbf{Completeness:} If $\sum X_i^2 \leq B$, then the probability that the prover aborts is at most:
        \begin{align} \label{eq:wraparound-errc}
            \errc = 1 - \Bin((\tau\cdot r); r , 1-\eta).
        \end{align}
        Thus, for $\tau \in (1/2,1-\eta)$, we get that $\errc \leq \exp\left(-2\left(1-\eta - \tau \right)^2 \cdot r \right)$. For $\tau \in [1-\eta,1]$, it is still the case that $\errc \leq r \cdot \eta$.
        
        If the prover doesn't abort, then the verifiers accept and it holds that: $(i)$ for every $k \in [r]$,  $g_k \cdot S_k =0$, and $(ii)$ the output shares $\{[g_k],[v_{k,j}],[u_{k,j}]\}$ are shares of bits.

        \item \textbf{Soundness:} If $\sum X_i^2 \geq q$, the probability that the verifiers accept, and that $(i)$ for every $k$ it holds that $g_k \cdot S_k =0$, and $(ii)$ the shares $\{[g_k], [v_{k,j}, [u_{k,j}]\}$ are all shares of bits, is at most
        \begin{align}
        \label{eq:wraparound-errs}
        \errs = \Bin((\tau\cdot r); r, \frac{1}{2}) \leq
        \exp\left(-2\left(\tau - \frac{1}{2}\right)^2 \cdot r \right).
        \end{align}

        \item \textbf{Zero-Knowledge}: The protocol  satisfies statistical zero-knowledge: the view of each verifier can be simulated up to statistical distance  $\errc$ (see Equation \eqref{eq:wraparound-errc}) .
    \end{enumerate}

    The protocol is public-coins, with 2 messages. The verifier's message is $(2 \cdot d \cdot r)$ bits, the prover's response is of length $\left(  \left ( \left \lceil \log \left( \sqrt{2 B \cdot \ln(2 / \eta)} + 1 \right) \right\rceil + 2 \right ) \cdot r \cdot \log(q) \right)$. The prover performs $O\left((d + \log(B \cdot \ln(2/\eta))) \cdot r \right)$ field operations. The verifiers each perform $O\left((d + \log(B \cdot \ln(2/\eta))) \cdot r \right)$ field operations, and communicate $O(r \cdot \log(q))$ bits between themselves.
\end{lemma}

\begin{proof} The protocol is in Figure \ref{fig:wraparound-protocol}, and the claimed complexity bounds follow by construction. We proceed to prove completeness for a single repetition (Proposition \ref{prop:wraparound-completeness}) and soundness for a single repetition (Proposition \ref{prop:wraparound-soundness}). For clarity, we omit the subscript $k$ when we consider a single iteration. Let $Y$ be the vector where $Y_i = Z_i X_i$. Recall that we set the parameter $\alpha = \sqrt{\ln (2/\eta)}$. 

\ifconf
\begin{restatable}[Completeness]{proposition}{completeness}
\else
\begin{proposition}[Completeness]
\fi
\label{prop:wraparound-completeness}

Let $\sum_i X_i^2 \leq B$. For every $\alpha > 0$ and every field size $q \geq 2\alpha\sqrt{B}$:
\begin{align*}
    & \Pr_{Z_1,\ldots,Z_d} \left[ Y \in [-\alpha \sqrt{B}, \alpha \sqrt{B}] \right] \geq 1 - 2e^{-\alpha^2}.
\end{align*}

\ifconf \end{restatable}
\else \end{proposition}
\fi

\ifconf
\else
\begin{proof} Recall that
$Y_i = Z_i X_i$. Thus, $Y_i$ has expectation 0 and is subgaussian\footnote{For background on subgaussian random variables, see e.g. \cite{Rig15}} 
 with parameter $(X_i/\sqrt{2})$:
\begin{align*}
    \forall \lambda \in \mathbb{R}: \E [e^{\lambda Y_i}] & = \frac{1}{2} + \frac{1}{4} \cdot \left( e^{\lambda X_i} + e^{-\lambda X_i}\right) \\
    & \leq \frac{1}{2} + \frac{1}{2} \cdot \left( e^{\lambda^2 X_i^2 / 2}  \right) \\
    & \leq e^{\lambda^2 X_i^2 / 4}.
\end{align*}

Thus, $Y$, which is a sum of independent subgaussians, is itself subgaussian, with parameter $\sqrt{\sum_i X_i^2/2} \leq \sqrt{B/2}$. By tail bounds for subgaussian RVs:
\begin{align*}
    \Pr[|Y| > \alpha\sqrt{B}] & \leq 2\exp \left( - \frac{2\alpha^2 B}{2B}
    \right) \\
    & \leq 2 \exp(-\alpha^2).
\end{align*}
\end{proof}
\fi


\begin{proposition}[Soundness]
\label{prop:wraparound-soundness}

For every $B$, every $\alpha \geq 1$ and every field size $q \geq \max \{  \frac{\alpha^2 B}{4000}, 2600 \alpha \sqrt{B}\}$, if it is the case that $\sum_i X_i^2 > q$, then:
\begin{align*}
    & \Pr_{Z_1,\ldots,Z_d}[Y \in [-\alpha \sqrt{B}, \alpha \sqrt{B}] \pmod{q}] \leq \frac{1}{2}.
\end{align*}

\end{proposition}

\begin{proof}  If any of the secret-shared $w_j$ values are not bits, then the soundness condition is automatically satisfied. Thus, we only need to bound the probability that the verifiers accept conditioned on the event that all $w_j$'s are bits. Under this conditioning, $(\sum_j w_j \cdot 2^j) \in [0,2\alpha \sqrt{B}]$. The verifiers check that $(\sum_j w_j \cdot 2^j) = Y + \alpha \sqrt{B} \pmod{q}$, and thus they will reject unless $Y \in [-\alpha \sqrt{B}, \alpha \sqrt{B}] \pmod{q}$. We bound the probability that $Y \pmod{q}$ is in this range using a case analysis (soundness follows), where we view the value of each $X_i$ modulo $q$ as an integer between $-q/2$ and $q/2$.

\paragraph{Soundness, case $I$ (large max).}
If:
\begin{align*}
    \max_i \left| X_i \right| > 2\alpha\sqrt{B} \pmod{q},
\end{align*}
then let $i^*$ be the argmax, where $|X_{i^*}| > 2\alpha \sqrt{B} \pmod{q}$ (note that this can only happen if $q > 4\alpha\sqrt{B}$).
The probability, over the choice of $Z$, that $\sum_i Z_i X_i \pmod{q}$ lands in the interval $[-\alpha \sqrt{B},\alpha \sqrt{B}]$ is at most $1/2$. To see this, fix all of $Z$'s entries except the $i^*$-th entry. There are two cases:

\begin{itemize}
    \item Let $S$ be the sum of all $Y_i$'s except the $i^*$-th $\pmod{q}$. If $S \in [-\alpha \sqrt{B}, \alpha \sqrt{B}]$, then when we add or subtract $X_{i^*}$, we end up outside the interval. This is because if $Z_{i^*} \neq 0$, then
    \begin{align*}
    |Y_{i^*}| \in (  2\alpha \sqrt{B} ,  q/2] \pmod{q}.
    \end{align*}
   Suppose w.l.o.g that $Y_{i^*}$ is positive (modulo $q$), then:
    \begin{align*}
    Y = S + Y_{i^*} \in ( \alpha \sqrt{B}, q/2 + \alpha \sqrt{B}] \pmod{q},
    \end{align*}
    and thus $Y$ is outside the interval (recall that $q > 4\alpha\sqrt{B}$). A similar statement holds for the case that $Y_{i^*}$ is negative (modulo $q$).

    \item If the sum of all $Y_i$'s except the $i^*$-th is outside the interval, then w.p. $1/2$ we have that $Z_{i^*}=0$ and the total sum is also outside the interval.
\end{itemize}

\paragraph{Soundness, case $II$ (bounded max).} Suppose it is the case that:
\begin{align*}
\max_i \left| X_i \right| \leq 2\alpha\sqrt{B}.
\end{align*}

In this case, we use the Berry-Esseen theorem to show that the distribution of $Y$ is sufficiently close to a Normal distribution, and this allows us to bound the probability that its magnitude is small modulo $q$.

\begin{theorem}[Berry-Esseen for our setting, see \cite{FellerV2, Shevstova2010}.]
\label{thm:berry-esseen-Y}

Suppose $\max_i |X_i| \leq 2\alpha\sqrt{B}$, and take $\sigma^2 = \sum_i X_i^2/2$. Then for any $y \in \mathbb{R}$:
\begin{align*}
\left| \Pr [Y \leq y] - \Pr[\mathcal{N}(0,\sigma^2) \leq y] \right| \leq \frac{0.56 \cdot  \alpha\sqrt{B}}{\sigma}.
\end{align*}
\end{theorem}

A symmetric claim follows for bounding the probability that $Y$ is above a value $y'$. For an interval $[y',y]$, the probability of $Y$ lying in the interval can be bounded via a Union Bound on the probability that it is below $y'$ and the probability that it is above $y$.

We use the Berry-Essen theorem, together with the concentration (and anti-concentration) properties of the Normal distribution, to bound the probability that $Y$'s magnitude is small modulo $q$. In Claim \ref{claim:bound-for-close} we bound the probability that $Y$ lies in any single interval where its magnitude modulo $q$ is small. In Claim \ref{claim:bound-for-far} we bound the probability that $Y$ is far from its expectation. To prove the soundness lemma, we combine these claims: bounding the probability that $Y$ is not too far from its expectation and has small magnitude modulo $q$, and also that it is far from its expectation. 
\ifconf
\else The details follow. 
\fi

\ifconf \begin{restatable}{claim}{boundclose}
\else
\begin{claim}
\fi
\label{claim:bound-for-close}
   Suppose that $\alpha \geq 1$, the field size is $q \geq 2\alpha \sqrt{B}$, and that $\max_i |X_i| \leq 2\alpha\sqrt{B}$, and take $\sigma^2 = \sum_i X_i^2/2$. Then for every integer $u \geq 0$:
\ifconf
\else   \begin{align*}
        \sum_{t=-u}^u \Pr \left[ Y \in [-\alpha\sqrt{B} + t \cdot q, \alpha\sqrt{B} + t \cdot q]  \right] \leq \left( 2u + 1 \right) \cdot \left( 1.12 + \sqrt{\frac{2}{\pi}} \right) \cdot \frac{ \alpha \sqrt{B}}{\sigma}
    \end{align*}
\fi
\ifconf \end{restatable}
\else \end{claim}
\fi

\ifconf
\else
\begin{proof}
For every integer $t$, by Theorem \ref{thm:berry-esseen-Y}:
\begin{align*}
    \Pr \left[ (Y - t \cdot q) \in [-\alpha\sqrt{B}, \alpha\sqrt{B}]  \right]
    & \leq \Pr \left[ \left( \mathcal{N}(0,\sigma^2) - t \cdot q \right) \in  [-\alpha\sqrt{B}, \alpha\sqrt{B}]  \right] +  \frac{1.12 \cdot  \alpha\sqrt{B}}{\sigma} \\
    & \leq \Pr \left[ \mathcal{N}(0,\sigma^2) \in  [-\alpha\sqrt{B}, \alpha\sqrt{B}]  \right] +  \frac{1.12 \cdot  \alpha\sqrt{B}}{\sigma} \\
    &\leq \left( 1.12 + \sqrt{\frac{2}{\pi}} \right) \frac{  \alpha\sqrt{B}}{\sigma}.
\end{align*}
The claim follows by taking a union bound over the possible values of $t$.
\end{proof}
\fi

\ifconf \begin{restatable}{claim}{boundfar}
\else \begin{claim}
\fi
\label{claim:bound-for-far}
    Suppose that $\alpha \geq 1$, the field size is $q \geq 2\alpha \sqrt{B}$, and that $\max_i |X_i| \leq 2\alpha\sqrt{B}$. Take $\sigma^2 = \sum_i X_i^2/2$. Then for every integer $t \geq (\sigma + \alpha \sqrt{B}) / q$ it holds that:
    \ifconf
\else
          \begin{align*}
        \Pr \left[ \left| Y \right| \geq  -\alpha\sqrt{B} + t \cdot q \right] \leq \sqrt{\frac{2}{\pi}} \exp \left(-\frac{\left( -\alpha\sqrt{B} + t \cdot q \right)^2}{2\sigma^2} \right) + \frac{1.12 \alpha \sqrt{B}}{\sigma}
      \end{align*}
      \fi
\ifconf \end{restatable}
\else \end{claim}
\fi
\ifconf
\else
\begin{proof}

We use a union bound over the probability that $Y$ is too large and the probability that it is too small. To bound the former, Theorem \ref{thm:berry-esseen-Y}, implies:
\begin{align*}
    \Pr \left[ Y \geq -\alpha\sqrt{B} + t \cdot q \right] & \leq \Pr \left[  \mathcal{N}(0,\sigma^2)  \geq -\alpha\sqrt{B} + t \cdot q \right] + \frac{0.56 \alpha \sqrt{B}}{\sigma} \\
    & \leq \frac{\sigma}{ \sqrt{2 \pi} \cdot (-\alpha\sqrt{B} + t \cdot q) } \cdot \exp \left(-\frac{\left( -\alpha\sqrt{B} + t \cdot q \right)^2}{2\sigma^2} \right) + \frac{0.56 \alpha \sqrt{B}}{\sigma} \\
    & \leq \frac{1}{\sqrt{2\pi}} \exp \left(-\frac{\left( -\alpha\sqrt{B} + t \cdot q \right)^2}{2\sigma^2} \right) + \frac{0.56 \alpha \sqrt{B}}{\sigma},
\end{align*}
where the last inequality holds because by the conditions of the claim $t \geq (\sigma + \alpha \sqrt{B}) / q$. Thus we have that $-\alpha \sqrt{B} + t \cdot q \geq \sigma$. The symmetric case (bounding the probability that $Y$ is too small) follows similarly.
\end{proof}
\fi
As described above, putting these claims together, we can bound both the probability that $Y \pmod{q}$ has small magnitude but $|Y|$ isn't too large, and the remaining case, where $|Y|$ is large. Towards this, for a parameter $\delta > 0$ to be specified below, we set:
\begin{align*}
u = \left \lfloor \frac{(\sigma \sqrt{2\ln(1/\delta)}) + (\alpha \sqrt{B})}{q}  \right \rfloor.
\end{align*}
\ifconf
\else
By a union bound over Claims \ref{claim:bound-for-close} and \ref{claim:bound-for-far}  (whose conditions all hold):
\begin{align*}
    & \Pr[Y \in [-\alpha \sqrt{B}, \alpha \sqrt{B}] \pmod{q}] \\
    & \leq \sum_{t=-(u-1)}^{u-1} \Pr \left[ Y \in [-\alpha\sqrt{B} + t \cdot q, \alpha\sqrt{B} + t \cdot q] \right] + \Pr \left[ \left| Y \right| \geq  -\alpha\sqrt{B} + u \cdot q \right] \\
    & \leq \left( 2u \cdot 1.12 + (2u -1) \cdot \sqrt{\frac{2}{\pi}} \right) \cdot \frac{ \alpha \sqrt{B}}{\sigma} + \sqrt{\frac{2}{\pi}} \cdot \exp \left(-\frac{\left( -\alpha\sqrt{B} + u \cdot q \right)^2}{2\sigma^2} \right) \\
    & \leq \left( 2u \cdot \left(1.12 + \sqrt{\frac{2}{\pi}} \right) \right) \cdot \frac{ \alpha \sqrt{B}}{\sigma} + \sqrt{\frac{2}{\pi}} \cdot \exp \left(-\frac{\left( -\alpha\sqrt{B} + u \cdot q \right)^2}{2\sigma^2} \right) \\
    &\leq \left( \alpha \sqrt{B} \cdot \left(1.12 + \sqrt{\frac{2}{\pi}} \right) \cdot \left( \frac{2 \sqrt{2 \ln(1/\delta)}}{q} +  \frac{2 \alpha \sqrt{B}}{q \cdot \sigma} \right) \right) + \delta  \\
    &\leq \left( \alpha \sqrt{B} \cdot \left(1.12 + \sqrt{\frac{2}{\pi}} \right) \cdot  \left( \frac{2 \sqrt{2 \ln(1/\delta)}}{q} +  \frac{0.002}{\sigma} \right) + \delta \right) \\
    & \leq \left( \alpha \sqrt{B} \cdot \left(1.12 + \sqrt{\frac{2}{\pi}} \right) \cdot \left(  \frac{2 \sqrt{ 2\ln(1/\delta)}  + 0.004 \sqrt{q}}{q} \right) + \delta \right),
\end{align*}
where above we used the fact that $\sigma > \sqrt{q/2}$ (which holds if there is wraparound) and also assumed that $q \geq 1000 \alpha\sqrt{B}$ (this latter condition will be guaranteed by the parameters set below).
\fi

Taking $\delta=0.005$ and working through the aritmetic, 
we have that:
\ifconf
\else
\begin{align*}
    \Pr[Y \in [-\alpha \sqrt{B}, \alpha \sqrt{B}] \pmod{q}] \leq
    \frac{13\alpha\sqrt{B}}{q} + \frac{\left(1.12 + \sqrt{\frac{2}{\pi}} \right) \cdot 0.004 \alpha\sqrt{B}}{\sqrt{q}}  + 0.005.
\end{align*}
\fi
\ifconf This probability can be bounded by $1/2$ for $q \geq \max \{  81 \alpha^2 B, 100\}$.\end{proof}
\else
Bounding the first summand by $0.005$ and the second summand by $0.49$, the total probability is bounded by $1/2$ (as claimed in the lemma statement). These bounds hold so long as
\begin{align*}
    q \geq \max \{ 2600 \alpha \sqrt{B}, \frac{\alpha^2 B}{4000}\}.
\end{align*}
\end{proof}
\fi
In each repetition $k$, the prover can satisfy the conditions $(i)$ $g_k \cdot S_k=0$, and $(ii)$  the secret-shared values $\{g_k, v_{k,j}, u_{k,j}\}$ are all shares of bits, if and only if either it holds that $Y_k \in [-\alpha \sqrt{B}, \alpha \sqrt{B}] \pmod{q}$, or the prover sets $g_k=0$. Over all $r$ repetitions, the prover can get the verifiers to accept while satisfying conditions $(i)$ and $(ii)$ above, if and only if in at least $\tau \cdot r$ of the repetitions, $Y_k \in [-\alpha \sqrt{B}, \alpha \sqrt{B}] \pmod{q}$.

Completeness and soundness over the $r$ repetitions are thus bounded by claimed Binomial terms (see Equations \eqref{eq:wraparound-errc} and \eqref{eq:wraparound-errs}). The binomial tail terms can be further bounded by taking a Chernoff Bound. For completeness, in each repetition the probability of success is at least $1-\eta$. The probability that in $r$ repetitions we have fewer than $\tau \cdot r$ successes is bounded by:
\begin{align}
    \exp \left( -\KL{\tau}{1-\eta} \cdot r \right) \leq \exp\left(-2\left(1-\eta - \tau \right)^2 \cdot r \right),
\end{align}
where $\KL{p}{q}$ is the KL divergence between the Bernoulli distribution with mean $p$ and the distribution with mean $q$.
For soundness, the probability of success in each repetition is at most $1/2$. The probability that in $r$ repetitions we have at least $\tau \cdot r$ successes is bounded by:
\begin{align}
    \exp \left( -\KL{\tau}{1/2} \cdot r \right) \leq \exp\left(-2\left(\tau - \frac{1}{2}\right)^2 \cdot r \right).
\end{align}

\begin{remark}
We note that while the protocol is described as sending the vectors $Z_1,\ldots,Z_k$, standard techniques~\cite{SchmidtSS95,AlonS00} of using limited independence ($O(\log 1/\rho)$-wise independence suffices for the concentration bounds) can be used to reduce the communication to $O(r\log d \log 1/\rho)$ bits.
\end{remark}
\begin{remark}
The analytic Chernoff bound expressions here can be rather loose for particular values of parameters (especially when $\eta$ is close to zero). In practice, one can use better analytic estimates, or numerical estimates for concentration on binomial random variables 
We directly state the bounds in terms of the CDF of the Binomial distribution in Equations \eqref{eq:wraparound-errc} and \eqref{eq:wraparound-errs}.
\end{remark}

\begin{remark}
  \label{remark:wraparound-efficiency}
  Suppose that for some $L,H$ satisfying $-\alpha\sqrt{B} \leq L < 0 < H \leq \alpha \sqrt{B}$, we test for $Y \in [L,H]$ instead of $Y \in [-\alpha \sqrt{B}, \alpha \sqrt{B}]$. Then the soundness argument continues to hold as $Y \in [L,H] \Rightarrow Y \in [-\alpha \sqrt{B}, \alpha\sqrt{B}]$.
  Moreover, the completeness argument now holds with $\tau' = 2\exp(-\hat{\alpha}^2)$ where $\hat{\alpha} = \min(\frac{-L}{\sqrt{B}}, \frac{H}{\sqrt{B}})$. The protocol's practical efficiency can be improved by a careful choice of $L,H$, see
  ~\cref{remark:inequality-efficiency}.
\end{remark}

\medskip \noindent \textbf{Honest-verifier zero-knowledge.}
To prove zero-knowledge, we need to show that for
a single verifier $V_j$, there exists a simulator $\Sim$ that does not know the private input $X$, but receives some share $X^{(j)}$ and generates a view that is statistically closely distributed to the view of the verifier $V_j$ in a real interaction with the prover and the remaining verifier $V_{1-j}$. 
We sketch the construction of such a simulator. Let $\set{X_1^{(j)},\ldots X_d^{(j)}}$ be the shares, given to $V_j$ as input.


\medskip \noindent \textbf{The Simulation:}
The simulator $\Sim$ starts by setting $\hat{X} = 0^d$ to be the (fake) input for the simulation and setting (fictitious) shares for $V_{1-j}$ accordingly. Specifically, $\Sim$  sets the ``shares'' for the second verifier $V_{1-j}$ to be $\hat{X}_i^{(1-j)}= -{X}_i^{(j)}$ and emulates an interaction between the two verifiers as follows.
    \begin{enumerate}

    \item The simulator $\Sim$ repeats the following $r$ times in parallel, where in the $k$-th repetition:



    It simulates the verifier $V_{1-j}$ in choosing a random string $Z_k \sim D_Z$ and receives from $V_j$ its messages to the prover to determine $Z_k$.

    $\Sim$ simulates the prover, setting $Y_k = \sum_{i=1}^d Z_{k,i} \hat{X}_i$ to be $0$. If $1\le k\le \tau\cdot r$ the simulator sets  $g_k=1$, and otherwise sets $g_k = 0$. 
    
    The simulator then
    sets the bits $\{v_{k,j}\}$ to be the Binary representation of the value $2\alpha \sqrt{B}$, and sets the bits $\{u_{k,j}\}_{j \in [b]}$ to all be $0$.
    The simulator $\Sim$ simulates the view of $\Vc_j$ in the appropriate $k$ executions of the protocol in Figure~\ref{fig:inequality-mod-q-protocol}, with $\alpha_i = Z_{k,i}$.

    \item
    The simulator selects shares for $\{g_k\}_{k \in [r]}$ and sends $\Vc_j$ its shares.

    \item For each $k \in [r]$, the simulator computes $\Vc_j$'s share of $S_k$ by taking the appropriate linear combination of existing shares:.
    \begin{align*}
     (s_k)^{(j)} = \left(\sum_{i=1}^d Z_{k,i} (X_i)^{(j)} \right)  + \alpha\sqrt{B}  - \left( \sum_{j \in [b]} 2^j \cdot (v_{k,j})^{(j)} \right) \pmod{q}
    \end{align*}

    \item The simulator $\Sim$ emulates $\Vc_{1-j}$ in its interaction with $\Vc_{j}$ verifying the linear equality $\sum_{k \in [r]} g_k = \tau \cdot r \pmod{q}$ (see Example \ref{example:linearEquality}).

    \end{enumerate}
    We first note that all choices of $Z_k$ for all $r$ repetitions are identically distributed in the simulation as in a real execution of the protocol. Condition on the event that for these choices, in the real execution of the protocol, it holds that  there are at least $\tau r$ values of $k$, for which $Y_k \in [-\alpha \sqrt{B}, \alpha\sqrt{B}]$. This means that the prover does not abort. Under the aforementioned conditioning, the view of $\Vc_j$ in the real execution is distributed identically to the view generated by the simulator. This is true since all that $\Vc_j$ ever sees are random shares and the views it sees in the simulation of the sub-protocol, which are perfectly simulated (since they are on accepting inputs). Finally, since in the real world, the probability that the prover aborts is at most $\errc$, we conclude that even without the conditioning, the two views are of statistical distance $\errc$.
\end{proof}

\subsection{Verifying Quadratic Constraints}
\label{subsec:quadratic}

In this section we recall a protocol from the work of Boneh {\em et al.} \cite{BonehBCGI19} for verifying a conjunction of low-degree constraints over secret-shared values. We focus on quadratic constraints. Let $\{X'_i\}_{i \in [n]}$ be a collection of secret-shared values. A quadratic constraint $C_k$ is specified by public coefficients $\{c_{k,i,j} \in \GF[q] \}_{i,j \in \{0, \ldots, n]}$ and a target value $a_k \in \GF[q]$, and the claim (to be verified) is that:
\begin{align*}
\sum_{i,j \in [0,n]} c_{k,i,j} \cdot X'_{i} \cdot X'_{j} = a_k,
\end{align*}
where we use the convention that $X'_0=1$ (this allows us to include linear terms in the constraint). 

Given a collection of $m$ quadratic constraints as above, 
there is a 2-message protocol for verifying that all the constraints hold, where the communication complexity scales with the {\em square-root} of the number of variables. More generally, the communication complexity can be reduced to (roughly) $n^{1/r}$ by increasing the amount of interaction to $(2r-2)$ messages.

Verifying that a secret-shared value $v$ is a bit can be expressed as the quadratic constraint $v^2-v=0$. Thus, we use this protocol to verify that (multiple) secret shares produced in other protocols are composed of bits.

\begin{lemma}[Boneh {\em et al.} {\cite[Corollary 4.7 and Remark 4.8]{BonehBCGI19}}]
\label{lemma:quadratic}

   For a field of size $q$, an integer $n$ that is a perfect square, a collection of secret-shared values $\{X'_i\}_{i=1}^n$, and quadratic constraints $C = \{c_{k,i,j},a_k\}_{k \in \{1,\ldots,m\}, i,j \in \{0,\ldots,n\}}$, there is a protocol for verifying the constraints as follows: all parties get as input the field size $q$ and the constraints. The prover gets both shares of each value $X'_i$. Each verifier gets a single share of each secret-shared value, where:

    \begin{description}
        \item[Completeness:] If all constraints hold, the verifiers accept.

        \item[Soundness:] If at least one of the constraints is violated, then (no matter what strategy the prover follows) the probability that the verifiers accept is at most $\left( \frac{2\sqrt{n}}{q - \sqrt{n}} + \frac{m}{q} \right)$.

        \item[Zero-Knowledge]: The protocol  is perfect strong distributed honest-verifier zero-knowledge (see Definition \ref{def:dZKIP}). 
    \end{description}
        

    The protocol is public-coins, with 2 messages. The verifiers' message is $\log(q)$ bits and the  prover's message is $((4\sqrt{n}+1) \cdot \log(q))$ bits. Let $\nnz(C)$ be the total number of non-zero coordinates in the constraints $\{c_{k,i,j}\}$. The prover performs $O(\nnz(C) + n\log(n) + m\log(m))$ field operations. In the course of verification, the verifiers each perform $O(\nnz(C) + n\log(n) + m\log(m))$ field operations, and they exchange $(2\sqrt{n}+2) \cdot \log(q)$ bits with each other.

\end{lemma}

\ifconf
\else
The completeness, soundness and dZK properties follow the proof in \cite{BonehBCGI19}, as do the communication complexities. We briefly recall the protocol to explain the runtimes of the prover and the verifier. The $k$ quadratic constraints are combined into one as follows: the verifiers choose a random field element $r$, and the protocol proceeds to prove the single constraint:
\begin{align*}
\sum_{k \in [1,m]} r^{k-1} \cdot \left( \left( \sum_{i,j} c_{k,i,j}\cdot X'_i \cdot X'_j \right) - a_k \right) = 0.
\end{align*}
The (public) coefficients of the resulting constraint can be arranged in an $n \times n$ matrix $A$ with $\nnz(C)$ non-zero coefficients. Viewing $X'$ as an $n \times 1$ vector, we need to show that:
\begin{align*}
    {X'}^T \cdot A \cdot X' = b,
\end{align*}
where $b$ is also a public element in the field (known to the prover and the verifiers). The verifiers can use their secret shares of $X'$ to compute secret shares of $Z = A \cdot X' \in \GF{q}^n$ (a linear function). The prover then needs to show that the inner product $\langle X', Z \rangle = b$, and this latter task is accomplished using a one-message protocol where the main steps involve polynomial interpolation, requiring time $O(n \log n)$ from the prover and the verifiers. Other than this protocol, there is also the question of computing the secret shares of $Z = A \cdot X'$. This can be done in time that is linear in the number of non-zero coordinates in $A$.
\fi

\paragraph{Soundness amplification.} The soundness error can be reduced by parallel repetition \cite{BabaiM88,Goldreich2001}:

\begin{corollary}[Parallel repetition of Lemma \ref{lemma:quadratic}]
\label{corollary:quadratic-repeated}

   For the same inputs considered in the protocol of lemma \ref{lemma:quadratic}, repeating that protocol in parallel $t$ times, where the verifiers accept iff all repetitions accept, gives a 2-message public-coins protocol with perfect completeness and zero-knowledge, where:
   \begin{enumerate}
       \item The soundness error is reduced to $\left( \frac{2\sqrt{n}}{q - \sqrt{n}} + \frac{m}{q} \right)^t$.

       \item The communication and runtime complexities are $t$ times larger than those in Lemma \ref{lemma:quadratic}. 
   \end{enumerate}
\end{corollary}

\subsection{Putting it Together: The Norm-Bound Protocol}
\label{subsec:putting-together}

Our final norm-bound protocol is described in Figure \ref{fig:L2-protocol}.

\begin{figure*}[t]
  \begin{boxedminipage}{\textwidth}
    \small \medskip \noindent

    \underline{\textbf{$L_2$-Bound Protocol}}\\

    \textbf{Common inputs:} Dimension $d \in \Nt$, claimed bound $B \in \Nt$, field size $q \in \Nt$, errors $\errc,\errs \in [0,1]$.  \\

    \textbf{Secret-shared inputs:} A vector $X \in \GF[q]^d$, where each $X_i$ is secret-shared as $[X_i] = (X_i^{(0)}, X_{i}^{(1)})$.

    The client (prover) knows the secret-shared values $X$ and the shares $[X]$.

    The servers (verifiers) each know their own shares (respectively $X^{(0)}$ and $X^{(1)}$). \\

     \textbf{Claim (to be verified):} $\sum_{i=1}^d X_i^2 \leq B$ (where summation is over the integers). \\

    \textbf{The Protocol:}

    \begin{enumerate}

    \item Run the wraparound protocol of Section \ref{subsec:wraparound-protocol}, setting the parameters $r, \eta, \tau$ as in Theorem \ref{thm:norm-bound}.

    The two messages in this subprotocol (the verifier sends the first message) are sent in messages 1 and 2 of the $L_2$-bound protocol. This results in secret shares of values $\{s_k\}_{k \in [r]}$ and alleged shares of bits $\{g_k, u_{k,j}, v_{k,j}\}$.

    \item Run the protocol of Section \ref{subsec:inequality} to verify that $\sum_i X_i^2 \in [0,B] \pmod{q}$.

    The prover's message in this sub-protocol is sent in message 2 of the $L_2$-bound protocol. Results in alleged shares of bits $\{v'_{j'},u'_{j'}\}$.

    \item \label{step:quadratic-check} Use the quadratic constraints protocol of Corollary \ref{corollary:quadratic-repeated}, setting the number of repetitions $t$ as in Theorem \ref{thm:norm-bound}, to verify the following quadratic constraints:
    \begin{enumerate}
        \item $\sum_{i=1}^d X_i^2 = \sum_{j'} v'_{j'} \cdot 2^{j'} \pmod{q}$.

        \item $\forall k \in [r], g_k \cdot s_k = 0$.

        \item the secret-shared values $\{v'_{j'},u'_{j'}\}$ and $\{g_k, u_{k,j}, v_{k,j}\}$ are all bits (i.e. in $\bitset$).

    \end{enumerate}

    The two messages of this sub-protocol (the verifier sends the first message) are sent in messages 3 and 4 of the $L_2$-bound protocol.

    \end{enumerate}

    If the verifiers in any of the sub-protocols executed above reject, then the verifiers in the $L_2$-bound protocol reject immediately. Otherwise, they accept.

  \end{boxedminipage}

  \caption{$L_2$-Bound Protocol}
  \label{fig:L2-protocol}
\end{figure*}

\ifconf
\else

\begin{theorem}
\label{thm:norm-bound}
    Fix a bound $B$, parameters $r,t \in \Nt, \eta \in [0,1], \tau \in (1/2,1]$ s.t. $\tau \cdot r$ is an integer. Let the field size be $q \geq \max \{  \frac{B \ln(2/\eta) }{4000}, 2600 \sqrt{B \ln(2/\eta)}, 1000dv , 3r \}$. The protocol of Figure \ref{fig:L2-protocol} has the following properties:

    \begin{enumerate}
        \item \textbf{Completeness:} If the claim is true and the prover follows the protocol, the verifiers accept with prob. $\geq 1-\errc$, where
        \begin{align*} 
            \errc = 1 - \Bin((\tau\cdot r); r , 1-\eta),
        \end{align*}
        where $\Bin(\ell;r,p)$ denotes the probability that the Binomial distribution with parameters $r$ and $p$ has outcome (number of successes) at least $\ell$. 
        
        Thus, for $\tau \in (1/2,1-\eta)$, $\errc \leq \exp\left(-2\left(1-\eta - \tau \right)^2 \cdot r \right)$. For $\tau \in [1-\eta,1]$, $\errc \leq r \cdot \eta$.

        \item \textbf{Soundness:} If $\sum X_i^2 > B$ (over the integers), the verifiers accept with probability at most:
        \begin{align*}
            \exp\left(-2\left(\tau - \frac{1}{2}\right)^2 \cdot r \right) + \left( \frac{ 2\sqrt{d + (\log(q)\cdot (r+2)/2)}}{q-2\sqrt{d + (\log(q)\cdot (r+2)/2)}} + \frac{\log(q)\cdot (r+2)}{2q} \right)^t.
        \end{align*}

        \item \textbf{Zero-Knowledge}: The protocol  satisfies statistical zero-knowledge: the view of each verifier can be simulated up to statistical distance $\errc$ (see above).
    \end{enumerate}
    The protocol is public-coins, with 4 messages. The message lengths (in bits) are:
    \begin{enumerate}
        \item the first messsage, sent by the verifier, is of length $(2 \cdot d \cdot r)$.

        \item the second message, sent by the prover, is of length at most $(\frac{r}{2}+2) \cdot \log^2(q)$.

        \item The third message, sent by the verifier, is of length $(t \cdot \log(q))$.

        \item The fourth message, sent by the prover, is of length $\left(t \cdot \left( 4\sqrt{d + (\log(q)\cdot (r+2)/2)} + 1 \right) \cdot \log(q)\right)$.

    \end{enumerate}

    The prover performs $\tilde{O}((d + \log(B \cdot \ln(2/\eta))) \cdot r)$ field operations. The verifiers each perform $\tilde{O}((d + \log(B \cdot \ln(2/\eta))) \cdot r)$ field operations, and exchange $O(\sqrt{d + (\log(B \cdot \ln(2/\eta)) \cdot r )} \log(q))$ bits between themselves.

\end{theorem}

\begin{proof}

Completeness and soundness follow directly from the sub-protocols.  For the quadratic constraints protocol that is run Step (\ref{step:quadratic-check}) there are:
\begin{itemize}
    \item The number of constraints is:
    \begin{align*}
        & 1 + r + 2\lceil \log(B + 1) \rceil + r \cdot(2\lceil \log(2\sqrt{B \ln(2/\eta)} + 1)\rceil + 1) \\
        &\leq r \cdot 2 \lceil \log(4\sqrt{B \ln(2/\eta)} + 2)\rceil + 4\lceil \log(\sqrt{B} + 1) \rceil + 1\\
        &\leq 2 \lceil \log(4\sqrt{B \ln(2/\eta)} + 2)\rceil \cdot \left( r + 2 \right)
    \end{align*}

    \item Similarly, the number of variables is:
    \begin{align*}
        & d + 2r + 2\lceil \log(B + 1) \rceil + r \cdot(2\lceil \log(2\sqrt{B \ln(2/\eta)} + 1)\rceil + 1) \\
        &\leq d + r \cdot 2 \lceil \log(4\sqrt{B \ln(2/\eta)} + 2)\rceil + 4\lceil \log(\sqrt{B} + 1)  \rceil \\
        &\leq d + \left( 2 \lceil \log(4\sqrt{B \ln(2/\eta)} + 2)\rceil \cdot \left( r + 2 \right) \right).
    \end{align*}
\end{itemize}
Combining the soundness error of the wraparound and quadratic constraints protocol, the total soundness error is bounded by:
\begin{align*}
\errs+ \left( \frac{ 2\sqrt{d + \left( 2 \lceil \log(4\sqrt{B \ln(2/\eta)} + 2)\rceil \cdot \left( r + 2 \right) \right)}}{q-\sqrt{d + \left( 2 \lceil \log(4\sqrt{B \ln(2/\eta)} + 2)\rceil \cdot \left( r + 2 \right) \right)}} + \frac{2 \lceil \log(4\sqrt{B \ln(2/\eta)} + 2)\rceil \cdot \left( r + 2 \right)}{q} \right)^t.
\end{align*}
Since $q \geq \max \{81B\cdot \ln(2/\eta), 1000, 3r\}$, we have that $\lceil \log(4\sqrt{B \ln(2/\eta)} + 2) \rceil \leq \log(\sqrt{q}/2) \leq \log(q)/2$. The soundness error is thus bounded by:
\begin{align*}
\errs+ \left( \frac{ 2\sqrt{d + (\log(q)\cdot (r+2)/2)}}{q-2\sqrt{d + (\log(q)\cdot (r+2)/2)}} + \frac{\log(q)\cdot (r+2)}{2q} \right)^t.
\end{align*}

The claimed communication complexity follows from the complexity of the different sub-protocols, and the above bounds on the number of variables and constraints in the protocol run in Step (\ref{step:quadratic-check}). In particular:
\begin{itemize}

\item The second message, sent by the prover, is of length at most:
\begin{align*}
\ifconf    & \left(  \left ( \left \lceil \log \left( \sqrt{2 B \cdot \ln(2 / \eta)} + 1 \right) \right\rceil + 2 \right) \cdot r  \right.\\&\,\,\,\,\left.+ \left(  2\lceil \log(B +1) \rceil  \right) \right) \cdot \log (q) \\
\else
    & \left(  \left ( \left \lceil \log \left( \sqrt{2 B \cdot \ln(2 / \eta)} + 1 \right) \right\rceil + 2 \right) \cdot r  + \left(  2\lceil \log(B +1) \rceil  \right) \right) \cdot \log (q) \\
\fi
&\leq \left(  \left( \log(q) \cdot r / 2\right)  + \left(  2 \log(q)  \right) \right) \cdot \log (q) \\
    &= \log^2(q) \cdot(\frac{r}{2}+2).
\end{align*}

\item The fourth message, sent by the prover, is of length at most:
\begin{align*}
 & t \cdot (4\sqrt{d + \left( 2 \lceil \log(4\sqrt{B \ln(2/\eta)} + 2)\rceil \cdot \left( r + 2 \right) \right)}+ 1 ) \cdot \log(q) \\
 &\leq t \cdot \left( 4\sqrt{d + (\log(q)\cdot (r+2)/2)} + 1 \right) \cdot \log(q).
\end{align*}
\end{itemize}

The zero-knowledge property follows directly from the zero-knowledge of all sub-protocols, using dZK composition, see Section~\ref{subsec:composition}.

 The claimed complexities follow by construction. We remark that the quadratic constraints, expressed as coefficient matrices as in the statement of Lemma \ref{lemma:quadratic}, have $O(d + (\log(B \cdot \ln(2/\eta)) \cdot r))$ non-zero coefficients.
\end{proof}
\fi

\ifconf
\else
\section{A Distributed Fiat-Shamir Transform for Removing Interaction}
\label{sec:fiat-shamir}

We use a distributed variant of the Fiat Shamir transform to obtain a non-interactive proof system. This distributed variant was proposed in~\cite{BonehBCGI19}, adapting the Fiat-Shamir transform~\cite{FiatS86}, a technique for eliminating interaction in standalone interactive protocols, to the setting of interactive proofs on distributed data.

\paragraph{Fiat-Shamir in the standalone setting.} The original scheme was shown to be sound in the random oracle model (ROM), when applied to constant-round public-coin protocols \cite{PointchevalS96}. The basic idea behind the transform is to replace each of the (standalone) verifier's random challenges with a challenge obtained by applying  cryptographic hash function $H$ to the preceding transcript (in the random oracle model $H$ is taken to be a truly random function). This eliminates the need for interaction between the prover and the verifier. In terms of zero-knowledge, if the original protocol satisfies honest-verifier zero knowledge, then the transformed protocol is zero-knowledge in the random oracle model. It has also been shown that concrete instantiations of the hash function guarantee zero knowledge under mild ``programmability'' conditions \cite{CanettiCHLRR18,CanettiLW18}.

\paragraph{Distributed Fiat-Shamir.} As mentioned by~\cite{BonehBCGI19}, in the distributed proof setting, even for public-coin protocols where the challenges $r_i$ are indeed random (and public), applying the Fiat-Shamir transform is not straightforward. This is because both the input and the communication transcript are
distributed and cannot be revealed to any single verifier. To get around this difficulty, they propose to let the prover generate each random challenge $r_i$ based on the joint view of the verifiers in previous rounds, by adding blinding values to the view of each verifier. This way, the validity of a challenge can be verified jointly by the verifiers, where each verifier checks the correlation to its private view (being given also the appropriate blinding value).

For an  $m$-round zero-knowledge proof protocol with two verifiers and public challenges $r_1,\ldots, r_m$, the distributed variant of the Fiat-Shamir transform suggested by~\cite{BonehBCGI19} proceeds as follows. The prover derives each random challenge
$r_i$ as the hash of the random challenges $r_{i,0}$ and $r_{i,1}$, where $r_{i,j}$ is obtained by hashing the view of the verifier $V_j$ up to this point. Concretely, let
$H \colon \set{0,1}^*\mapsto \set{0,1}^\kappa$
be a random hash function, where
$\kappa$ is a security parameter. For each round $i\in[m]$ and $j \in\set{0,1}$, the prover $P$ selects the blinding value  $\nu_{i,j}\leftarrow \set{0,1}^{\kappa}$ and sets $r_{i,j}=H(i,j,x_j,\nu_{i,j},\pi_{i,j})$ (to be the part of the challenge that is verifiable by $V_j$). Finally, the prover sets the $i$'th challenge $r_i$ to be $r_{i}=H(i,\bot, r_{i,0},r_{i,1})$.

Note that the blinding values $\nu_{i,j}$ ensure that the hashes
$r_{i,j}$ do not leak any information about $\pi_{i,j}$. Furthermore, if the prover sends each verifier $V_j$ the blinding value $\nu_{i,j}$, then
the verifiers can verify the validity of the challenges by replaying the process in which the prover constructed them.

Finally, we emphasize that the transformed non-interactive proof system maintains the zero-knowledge guarantee of the original interactive protocol. If that protocol was public coins, and the only communication between the verifiers was exchanging a single message at the end (see Definition \ref{def:dZKIP}), then the non-interactive protocol is dZK even against malicious verifiers (each verifier's behavior does not effect the messages sent by the prover or by the other verifier). For concrete instantiations of the protocol, this will be true so long as the hash family is programmable, see \cite{CanettiCHLRR18,CanettiLW18} (this is a very mild condition). Thus, our transformed non-interactive proofs all satisfy zero-knowledge also against a single malicious verifier (we always assume that at least one verifier is honest).
\fi

\section{Differentially Private Secret Sharing}
\label{sec:DP}

We first recall a notion of near-indistinguishability used in
Differential Privacy, and the notion of differential zero knowledge from \cite{Talwar22}:
\begin{definition}[$(\varepsilon,\delta)$-closeness]
Two random variables $P$ and $Q$ are said to be $(\varepsilon,\delta)$-close, denoted by $P \edclose Q$ if for all events $S$, it holds that
$\Pr[P \in S] \leq e^{\varepsilon}\cdot \Pr[Q \in S]+\delta$, and similarly, $Pr[Q \in S] \leq e^{\varepsilon}\cdot \Pr[P \in S]+\delta$.
\end{definition}

\ifconf
\else

\begin{definition}[Differential Zero Knowledge]
We say a protocol $\pi$ is $(\varepsilon, \delta)$-Differentially Zero Knowledge w.r.t.~$L$ if there is an efficiently samplable
distribution $Q$ such that for all $x \in L$, the distribution $\pi(x)$ of the protocol’s transcript on input $x$
satisfies $\pi(x)\edclose Q$.
\end{definition}

We remark that even if the distribution $Q$ in the above definition is {\em not} efficiently samplable, we can get an efficient simulator by paying (at most) a factor of 2 in the privacy parameters: 

\begin{remark}[A computationally bounded simulator]
If $\pi$ is an efficient protocol which is $(\varepsilon, \delta)$-Differentially Zero Knowledge w.r.t.~$L$ but with an {\em inefficient} simulator distribution $Q$, one can get an efficient simulator
distribution $Q'$ such that for all $x \in L$, the distribution $\pi(x)$ of the protocol’s transcript on input $x$
satisfies $\pi(x)\edcloseTwo Q'$. A sample from $Q'$ is obtained by running $\pi$ with some $x_0\in L$ and outputting the transcript. We assume that it is computationally easy to find {\em some} input $x_0 \in L$ (or, alternatively, the simulator can be non-uniform). 
\end{remark}

\begin{figure*}[!t]
  \begin{boxedminipage}{\textwidth}
    \small \medskip \noindent

    \underline{\textbf{Differential Zero Knowledge $L_2$-Bound Protocol}}\\

    \textbf{Common inputs:} Dimension $d \in \Nt$, claimed bound $B \in \Nt$, field size $q \in \Nt$, parameters $\eps,\delta \in [0,1]$.  \\

    \textbf{Prover Input:} A vector $X \in \GF[q]^d$.\\

     \textbf{Claim (to be verified):} $\sum_{i=1}^d X_i^2 \leq B$ (where summation is over the integers). \\

    \textbf{The Protocol:}

    \begin{enumerate}

    \item Set $\sigma = c_{\eps, \frac{\delta}{2}}\cdot \sqrt{B}$, $\Delta = d\sigma^2(1+\frac{2\sqrt{\log 8e/\delta}}{\sqrt{d}} + \frac{2\log 8e/\delta}{d})$. Sample $R \sim \mathcal{N}^{\Delta}(0, \sigma^2 \mathbb{I})$, and define secret shares $X_j^{(0)} = -\ceil{R_j}$, $X_j^{(1)} = X_j + \ceil{R_j}$.
    Send to the verifiers these secret shares. Verifier 0 (resp. Verifier 1) will check that each received secret shares satisfies $|X^{(0)}|_2 \leq \sqrt{B}+\sqrt{\Delta}+\sqrt{d}$ (resp. $|X_j^{(1)}| \leq \sqrt{B}+\sqrt{\Delta}+\sqrt{d}$) and rejects otherwise.

    \item Run the protocol of Section \ref{subsec:inequality} to verify that $\sum_i X_i^2 \in [0,B] \pmod{q}$.

    The prover's message in this sub-protocol is sent in message 1 of the $L_2$-bound protocol. Results in alleged shares of bits $\{v'_{j'},u'_{j'}\}$.

    \item \label{step:quadratic-check-dzk} Use the quadratic constraints protocol of Corollary \ref{corollary:quadratic-repeated}, setting the number of repetitions $t$ as in Theorem \ref{thm:dzk-norm-bound}, to verify the following quadratic constraints:
    \begin{enumerate}
        \item $\sum_{i=1}^d X_i^2 = \sum_{j'} v'_{j'} \cdot 2^{j'} \pmod{q}$.

        \item the secret-shared values $\{v'_{j'},u'_{j'}\}$ are all bits (i.e. in $\bitset$).

    \end{enumerate}

    The two messages of this sub-protocol (the verifier sends the first message) are sent in messages 2 and 3 of the $L_2$-bound protocol.

    \end{enumerate}

    If the verifiers in any of the sub-protocols executed above reject, then the verifiers in the $L_2$-bound protocol reject immediately. Otherwise, they accept.

  \end{boxedminipage}

  \caption{Differential Zero Knowledge $L_2$-Bound Protocol}
  \label{fig:L2-protocol-dzk}
\end{figure*}


\paragraph{Proving the norm with Differential Zero Knowledge.}
We next describe a way for a
client to share a vector $X \in L = \{X \in \mathbb{Z}^d : \|X\|_2^2 \leq B\}$, while preserving differential secrecy. Here, the client will share each $X_i$ by adding (rounded) truncated Gaussian noise\footnote{Alternately, one can also use a truncation of the Discrete Gaussian Mechanism~\cite{canonne2020discrete}, which will give similar bounds.} of magnitude large enough to guarantee differential privacy.

For parameters $\sigma,\Delta$, let $\mathcal{N}^{\Delta}(\mu, \sigma^2 \mathbb{I})$ be the multi-dimensional Gaussian distribution with variance $\sigma^2 \mathbb{I}$, conditioned on the $\ell_2$ norm being at most $\sqrt{\Delta}$.
The following lemma is standard; 
\begin{lemma}
  \label{lem:gaussian-dzk}
    Let $\eps,\delta \in (0,1)$, and let $x \in \mathbb{Z}^d$ satisfy $\|x\|_2^2 \leq B$. Let $\sigma^2 = c_{\eps,\frac{\delta}{2}}^2 B$, where $c_{\eps, \delta} =  \sqrt{2\ln \frac{1.25}{\delta}} / \eps$ and let $\Delta \geq d\sigma^2(1+\frac{2\sqrt{\log 8e/\delta}}{\sqrt{d}} + \frac{2\log 8e/\delta}{d})$.
  Let $P \sim \mathcal{N}^{\Delta}(0, \sigma^2 \mathbb{I})$ and $Q \sim x + \mathcal{N}^{\Delta}(0, \sigma^2 \mathbb{I})$.
  Then $P \edclose Q$. It follows that $\lceil P \rceil \edclose \lceil Q \rceil$.
\end{lemma}

\begin{proof}
  Let $P'$ and $Q'$ denote samples from un-truncated Gaussians $\mathcal{N}(0, \sigma^2 \mathbb{I})$ and $ x + \mathcal{N}(0, \sigma^2 \mathbb{I})$ respectively.  With this $\sigma$, the privacy analysis of the Gaussian Mechanism~\cite{DworkR14} implies that without truncation, the two random variables $P'$ and $Q'$ would be $(\eps,\delta/2)$-close.
  Let $r$ the probability that a sample from $P'$ lies in the $\ell_2$ ball; We choose $\Delta$ so that the probability of truncation $(1-r)$ is at most $\delta/8e$. Indeed recall that the squared $\ell_2$ norm of $P'$ (scaled by $\sigma^2$) is a sample from a chi-square distribution $\chi_d^2$. We will use the following tail bounds for $\chi_k^2$ random variables from Laurent and Massart \cite[Lemma 1 rephrased]{LaurentM2000}:
\begin{theorem}
  \label{thm:chi2_tails}
  Let $Z$ be a $\chi_k^2$ random variable. Then for any $\beta >0$,
  \begin{align*}
    \Pr[\frac{1}{k}Z \geq 1 + 2\sqrt{\beta/k} + 2\beta/k] &\leq \exp(-\beta).
  \end{align*}
  \end{theorem}
  Setting $\beta = \delta/8e$, we conclude that $\Delta = d\sigma^2(1+\frac{2\sqrt{\log 8e/\delta}}{\sqrt{d}} + \frac{2\log 8e/\delta}{d})$ suffices.

  Now for any event $E$, we write
  \begin{align*}
  \Pr_P[E] &\leq r^{-1}\Pr_{P'}[E]\\
  &\leq r^{-1} (e^\eps \cdot \Pr_{Q'}[E] + \delta/2)\\
  &\leq r^{-1} (e^{\eps}\cdot (r\Pr_Q[E] + (1-r)) + \delta/2)\\
  &\leq e^{\eps}\cdot \Pr_Q[E] + r^{-1}((1-r)e^{\eps}+\delta/2)\\
  &\leq e^{\eps}\cdot \Pr_Q[E] + (8/7)((\delta/8e)(e)+\delta/2)\\
  &\leq e^{\eps}\cdot \Pr_Q[E] + \delta.
\end{align*}
Finally, the standard post-processing property of differential privacy implies that applying co-ordinate-wise ceiling to $P$ and $Q$ preserves $(\eps,\delta)$-closeness.
\end{proof}

We note that the parameter $c_{\eps,\frac{\delta}{2}}$ may be replaced by a tighter numerical bound using the Analytical Gaussian Mechanism~\cite{BalleW18} without impacting the result.

The client (prover) will sample truncated gaussian noise and take the ceiling to get $R$. It will then secret-share $X$ as $-R$ and $X+R$. Note that since $\|X\|_2 \leq \sqrt{B}$ and $\|R\|_{2} \leq \sqrt{\Delta}+\sqrt{d}$, it follows that $\max(\|-R\|_2, \|X+R\|_{2}) \leq \Lambda \eqdef \sqrt{B}+\sqrt{\Delta}+\sqrt{d}$.
For any shares $R^{(1)}, R^{(2)}$ with $\|R^{(i)}\|_2 \leq \Lambda$, it is immediate the $\|R^{(1)}+R^{(2)}\|_2^2 \leq  4\Lambda^2$.
Thus for $q > 4 \Lambda^2$, there can be no wraparound, and it suffices to validate that the squared norm modulo $q$ is bounded.

Note that for an honest prover, the secret sharing scheme always succeeds, and satisfies $(\eps,\delta)$-differential zero knowledge. Further, for $q > 4\Lambda^2$, there is no wraparound and thus the protocol inherits the soundness of the ``range check modulo $q$'' protocol. We estimate this lower bound on $q$ in terms of the parameters:
\begin{align*}
  4&\Lambda^2  \leq 4 (\sqrt{B} + \sqrt{\Delta} + \sqrt{d})^2\\
  &= 4\left(\sqrt{B} + \sqrt{d} + \sigma\sqrt{d\cdot(1+\frac{2\sqrt{\log 8e/\delta}}{\sqrt{d}} + \frac{2\log 8e/\delta}{d})}\right)^2\\
  &= 4\left(\sqrt{B} + \sqrt{d} + \sqrt{dBc_{\eps,\frac{\delta}{2}}\cdot(1+\frac{2\sqrt{\log 8e/\delta}}{\sqrt{d}} + \frac{2\log 8e/\delta}{d})}\right)^2\\
\end{align*}

\begin{theorem}
  \label{thm:dzk-norm-bound}
  Let $(\eps,\delta) \in (0,1)$ and $B \geq 1$. Set $q > 4\left(\sqrt{B} + \sqrt{d} + \sqrt{dBc_{\eps,\frac{\delta}{2}}\cdot(1+\frac{2\sqrt{\log 8e/\delta}}{\sqrt{d}} + \frac{2\log 8e/\delta}{d})}\right)^2$. Then the protocol of Figure \ref{fig:L2-protocol-dzk} has the following properties.
  \begin{enumerate}
   \item \textbf{Completeness:} If the claim $\sum_{i=1}^d X_i^2 \leq B$ is true and the prover follows the protocol, the verifiers accept with probability $1$.

    \item \textbf{Soundness:} If $\sum X_i^2 > B$ (over the integers), the probability that the verifiers accept is at most:
    \begin{align*}
        \left( \frac{ 2\sqrt{d + 2\log(q)}}{q-2\sqrt{d + 2\log(q)}} + \frac{2\log(q)+1}{q} \right)^t.
    \end{align*}

    \item \textbf{Zero-Knowledge}: The protocol  satisfies $(\eps,\delta)$-differential zero knowledge: the view of each verifier can be efficiently simulated up to $(\eps,\delta)$-closeness.
\end{enumerate}
The protocol is public-coins, with 3 messages. The message lengths (in bits) are:
\begin{enumerate}
    \item the first messsage, sent by the prover, is of length at most $(2 \cdot d \cdot \lceil\log_2 q\rceil + 4\lceil log_2 q\rceil^2)$.

    \item The second message, sent by the verifier, is of length $(t \cdot \lceil\log_2 q\rceil)$.

    \item The third message, sent by the prover, is of length $\left(t \cdot \left( 4\sqrt{d + 2\lceil \log(B + 1) \rceil} + 1 \right) \cdot \lceil\log_2 q\rceil\right)$.

\end{enumerate}

\end{theorem}
\begin{proof}

  Completeness and soundness follow directly from the sub-protocols. For the quadratic constraints protocol that is run Step (\ref{step:quadratic-check-dzk}) there are:
  \begin{itemize}
      \item The number of constraints is $1 + 2\lceil \log(B + 1) \rceil$.
      \item The number of variables is $d + 2\lceil \log(B + 1) \rceil$.
  \end{itemize}
  Thus from~\cref{lemma:quadratic}, the soundness error is:
  \begin{align*}
\left( \frac{ 2\sqrt{d + 2\lceil \log(B + 1) \rceil}}{q-\sqrt{d + 2\lceil \log(B + 1) \rceil}} + \frac{1 + 2\lceil \log(B + 1) \rceil}{q} \right)^t.
  \end{align*}
  The claimed communication complexity follows from the complexity of the different sub-protocols, and the above bounds on the number of variables and constraints in the protocol run in Step (\ref{step:quadratic-check-dzk}).

  We next argue the zero knowledge property. The simulator for verifier $1$ samples $R \sim \mathcal{N}^{\Delta}(0, \sigma^2 \mathbb{I})$, and outputs $-\ceil{R}$. This gives perfect zero knowledge for this step.
  The simulator for verifier $2$ samples $R \sim \mathcal{N}^{\Delta}(0, \sigma^2 \mathbb{I})$, and outputs $\ceil{R}$. \cref{lem:gaussian-dzk} implies that this step satisfies $(\eps,\delta)$-zero knowledge.
  The zero knowledge property of steps (2) and (3) follows from the ZK property of the respective subprotocols.
\end{proof}
\fi

\remove{

\paragraph{Proving the norm with differential secrecy.}
We next describe a way for a
client to share $\myvec{x}$: Share each $x_i$ by adding
Gaussian noise of magnitude large enough to guarantee differential privacy, but still much smaller than the field size. Specifically, the prover draws
$r_i \longleftarrow \mathcal{N}(0,\sigma^2)$, where
$\sigma = \frac{2}{\varepsilon}\cdot\sqrt{B\ln\left(\frac{1.25}{\delta}\right)}$, and shares $r_i$ and $x_i-r_i$ to the two servers.

Now,  the servers can verify that the following hold:
\begin{enumerate}
    \item The magnitude of each share is at most. This requires no actual verification, as the client would only send as many bits as required for this magnitude (all other bits in will be set to zero by the servers).
    \item The magnitude of the shared value of $\sum_i X_i = B \pmod{q}$. This is done using the above -- Range Check Modulo $q$ Protocol -- as black-box.
\end{enumerate}

We next prove three properties of the above scheme.
\begin{description}
\item[Completeness:]
We use the following tail bound on Normal distribution.
\begin{lemma}[Gaussian tail bound -- Stack Exchange]
Let $r\longleftarrow \mathcal{N}(0,\sigma^2)$ and let $\gamma>0$, then
\begin{align*}
     \Pr\left[\left|r\right| > \sigma\gamma\right]  \leq \frac{ 2e^{-\gamma^2/2}}{\gamma\cdot\sqrt{2\pi}}.
\end{align*}
\end{lemma}
Roughly, for $\beta$-completeness, we need $\gamma$ to be such that
$d\cdot e^{-\gamma^2/2} < \beta$. It suffices to take $\gamma> \sqrt{2(\ln d-\ln \beta)}$. In addition, for such a $\gamma$, it suffices to allow shares of length $\log(\sigma\gamma+\sqrt{B})$ (i.e., the maximum size of the noise $r_i$ added to the maximum value of $x_i$).

\item[Soundness:]
To argue about soundness, we can make sure there is no wraparound, and let the soundness parameter $\alpha$ be determined by the external proof of the size of $\sum_i x_i^2$. To guarantee this, we need to set the field size $q$ to be large enough. It suffices to limit the magnitude of each share to
$\alpha \sqrt{B}$, and let $q> t\cdot d \cdot \alpha \sqrt{B}$, where $t$ is the number of servers (in our case $t=2$).

By the arguments above, $\alpha B= \sigma\gamma +\sqrt{B}\geq  \frac{2}{\varepsilon}\cdot\sqrt{B\ln\left(\frac{1.25}{\delta}\right)}\cdot\sqrt{100+\log d}+\sqrt{B}$. That is, we can take $\alpha > 1+\frac{2}{\varepsilon}\cdot\sqrt{\ln\left(\frac{1.25}{\delta}\right)\cdot\left(100+\log d\right)}$
\item[Differentially private ZK:]

\end{description}

The following is a special case of the privacy of the Gaussian mechanism \cite[Thm A.1]{DworkR14}.
\begin{lemma}
Let $\epsilon, d > 0$ and let $x \in \mathbb {R}^d$ satisfy $||x||_2 \leq B$. Let $P \sim \mathcal{N}(0, \sigma^2 \mathbb{I}_d )$ and let $Q \sim
x +  \mathcal{N}(0, \sigma^2 \mathbb{I}_d)$. Then, $P \edclose Q$ if $\sigma > \frac{1}{\varepsilon}\cdot\sqrt{B\cdot\ln\left(\frac{1.25}{\delta}\right)}$.
\end{lemma}

} 

\ifconf
\begin{table*}[!t]
\else
\begin{table}[!t]
\fi
\centering

\begin{tabular}{|c || c | c |c | c|} 
	\hline
	& $d=10^4$ & $d=10^5$ & $d=10^6$ & $d=10^7$ \\ [0.5ex] 
	\hline\hline
	no robustness, $\#$ bits sent & $64 \cdot 10^4$ &  $64 \cdot 10^5$ & $64 \cdot 10^6$ & $64 \cdot 10^7$ \\ 
	\hline
	prior work, overhead \cite{BonehBCGI19,libprio} & $> 1500\%$ & $> 1500\%$ & $> 1500\%$ & $> 1500\%$ \\
	\hline
	PINE, Statistical ZK, overhead &$43.01\%$  & $6.27\%$ & $0.97\%$ & $0.26\%$ \\
	\hline
	PINE, Differential ZK, overhead &$8.92\%$  & $2.86\%$ & $0.63\%$ & $12.76\%$ \\ [1ex] 
	\hline
\end{tabular}

\caption{Communication analysis: our protocols and prior work. Parameters: {\textbf{\color{blue} field size $q \approx 2^{64}$}} for aggregation, $d$-dimensional data, {\textbf{\color{blue}soundness error $2^{-100}$}}, zero-knowledge error $\delta=2^{-50}$. For differential ZK $\epsilon=0.1$.}
\label{table:perf-eval-szk -64}

\begin{tabular}{|c || c | c |c | c|} 
 \hline
  & $d=10^4$ & $d=10^5$ & $d=10^6$ & $d=10^7$ \\ [0.5ex] 
 \hline\hline
 no robustness, $\#$ bits sent & $128 \cdot 10^4$ &  $128 \cdot 10^5$ & $128 \cdot 10^6$ & $128 \cdot 10^7$ \\ 
 \hline
 prior work, overhead \cite{BonehBCGI19,libprio} & $> 1500\%$ & $> 1500\%$ & $> 1500\%$ & $> 1500\%$ \\
 \hline
 PINE, Statistical ZK, $\rho=2^{-50}$, overhead & $22\%$ & $3.18\%$ & $0.49\%$ & $0.13\%$ \\
 \hline
 PINE, Statistical ZK, $\rho=2^{-100}$, overhead &$36\%$  & $4.58\%$ & $0.63\%$ & $0.15\%$ \\
 \hline
 PINE, Differential ZK, $\rho=2^{-50}$, overhead &$4.77\%$  & $1.46\%$ & $0.32\%$ & $0.11\%$ \\
 \hline
 PINE, Differential ZK, $\rho=2^{-100}$, overhead &$4.77\%$  & $1.46\%$ & $0.32\%$ & $0.11\%$ \\ [1ex] 
 \hline
\end{tabular}

\caption{Communication analysis: our protocols and prior work. Parameters: {\textbf{\color{blue} field size $q \approx 2^{128}$}}, $d$-dimensional data, {\textbf{\color{blue} soundness error $\rho=2^{-50}$ and $2^{-100}$}}, zero-knowledge error $\delta=2^{-50}$. For differential ZK $\epsilon=0.1$.}
\label{table:perf-eval-szk-128}

\begin{tabular}{|c || c | c |c | c|} 
 \hline
  & $d=10^4$ & $d=10^5$ & $d=10^6$ & $d=10^7$ \\ [0.5ex] 
 \hline\hline
 no robustness, $\#$ bits sent & $64 \cdot 10^4$ &  $64 \cdot 10^5$ & $64 \cdot 10^6$ & $64 \cdot 10^7$ \\ 
 \hline
 prior work, overhead \cite{BonehBCGI19,libprio} & $> 1500\%$ & $> 1500\%$ & $> 1500\%$ & $> 1500\%$ \\
 \hline
 PINE, Differential ZK, $\epsilon=0.1$, overhead & $4.77\%$ & $1.46\%$ & $0.32\%$ & $12.63\%$ \\
 \hline
 PINE, Differential ZK, $\epsilon=0.01$, overhead &$17.87\%$  & $14.14\%$ & $12.86\%$ & $25.14\%$ \\
 \hline
 PINE, Differential ZK, $\epsilon=0.001$, overhead &$17.87\%$  & $26.83\%$ & $25.40\%$ & $37.66\%$ \\ [1ex] 
 \hline
\end{tabular}

\caption{Communication analysis: {\textbf{\color{blue}Differential ZK}} protocol and prior work. Parameters: field size $q \approx 2^{64}$, $d$-dimensional data, soundness error $\rho=2^{-50}$, zero-knowledge error $\delta=2^{-50}$.}
\label{table:perf-eval-dzk-64}

\ifconf 
\end{table*} 
\else 
\end{table}
\fi

\section{Performance Evaluation}
\label{sec:performance}

We provide a further performance analysis for our protocols in different parameter regimes. As in Table \ref{table:intro-performance}, we analyze performance in terms of the communication overhead, beyond the communication that is needed to simply send secret shares for distributed aggregation (without any robustness to poisoning attacks). We consider aggregating  $d$-dimensional integer vectors of $\ell_2$ norm at most $2^{15}$ with $d \in \{10^4, 10^5, 10^6, 10^7\}$, where the aggregation is performed over secret-shared inputs in a field of size $q=2^{64}$ or $q=2^{128}$ (sending secret shares for the client's data requires $d \cdot \log q$ bits).  The (statistical) zero-knowledge error is set to  $\delta=2^{-50}$ throughout.

In Table \ref{table:perf-eval-szk -64} we analyze PINE's communication overhead for a smaller soundness error $2^{-100}$, fixing all other parameters to be the same as in the results of Table \ref{table:intro-performance} (where the soundness error was $2^{-50}$). The overhead for {\textbf{Statistical ZK}} is larger by a roughly 2x factor (compared with the results in Table \ref{table:intro-performance}): the larger overhead comes from additional repetitions of the sub-protocols. The overhead for {\textbf{Differential ZK}} is, similarly, larger by a roughly 2x factor (compared with the results in Table \ref{table:intro-performance}), except in the high-dimensional regime, where the main bottleneck is the increased field size, and the overhead for smaller soundness error is not much larger than the overhead for soundness error $2^{-50}$.

In Table \ref{table:perf-eval-szk-128}, we analyze PINE's communication overhead when the field is large ($q=2^{128}$) for soundness errors $2^{-50}$ and $2^{-100}$. For {\textbf{Statistical ZK}}, 
comparing with the performance for field size $2^{64}$, the larger field size does not change the overhead for soundness error $2^{-50}$ (vs. Table \ref{table:intro-performance}), but the overhead is slightly smaller for soundness error $2^{-100}$ (vs. Table \ref{table:perf-eval-szk -64}). For {\textbf{Differential ZK}}, the large field size reduces the overhead by a 2x or larger (for high dimensionality) factor (vs. Tables \ref{table:intro-performance} and \ref{table:perf-eval-szk -64}). This is because the main bottleneck in the Differential ZK protocol for high-dimensional data was having a large-enough field. For lower-dimensional data, the larger field size automatically gives a smaller soundness error, which reduces overheads.  Indeed, for a field size this large, soundness $2^{-100}$ comes ``for free'', at no additional overhead (compared with the overhead for soundness error $2^{-50}$).

In Table \ref{table:perf-eval-dzk-64}  we analyze {\textbf{Differential PINE}}'s performance as a function of the privacy parameter $\varepsilon \in  \{0.1,0.01,.0.001\}$. The overhead increases by (at most) a constant multiplicative factor for each order of magnitude improvement (reduction) in the privacy parameter. This is again due to the main ``bottleneck'' being the field size, which needs to grow linearly in $(1/\varepsilon)$ (thus the bit length of a field element grows with $\log(1/\varepsilon)$).

\medskip\noindent{\bf On the precision parameter.} As discussed after the performance analysis in Section \ref{subsec:our-work}, we consider aggregating floating point vectors of Euclidean norm at most 1, and support $b=15$  bits of precision.  Our main focus is on distributed aggregation, where noise will be added to the aggregate before it is revealed to the servers (to guarantee differential privacy). Since we expect noise to be added to the aggregate, there is limited value in increasing the precision for individual contributions.

\medskip\noindent{\bf Further performance evaluation.}We provide further evaluations in 
\ifconf
\else
\cref{sec:appendix-evaluations}.
\fi
First, we analyze PINE's {\em runtime} overhead for the prover and the verifier, and find that they are improved by 1-2 orders of magnitude compared to the prior work of \cite{BonehBCGI19,libprio}. We also provide more detailed evaluations of statistical PINE's communication overhead for many different choices of soundness and zero-knowledge errors.

\printbibliography

\ifconf
\else
\appendix
\ifconf
\section{Formal description of Range Check mod $q$ Protocol}
\label{app:inequality-mod-q-protocol}

The protocol is in Figure \ref{fig:inequality-mod-q-protocol}.

\begin{figure*}[t]
  \thisfloatpagestyle{empty}
  \begin{boxedminipage}{\textwidth}
    \small \medskip \noindent

    \underline{\textbf{Protocol: Range Check $\pmod{q}$}}\\

    \textbf{Common inputs:} Field size $q \in \Nt$, number of variables $n \in \Nt$, coefficients $\alpha_1,\ldots,\alpha_n \in \GF[q]$ and claimed lower and upper bounds $\beta_1,\beta_2 \in \GF[q]$, viewed as integers in $\{-\lfloor q/2 \rfloor,\ldots, \lfloor q/2 \rfloor \}$, s.t. $\beta_1 \leq \beta_2$ and $q > 3(\beta_2 - \beta_1)+2$.\\

    \textbf{Other inputs:} The prover knows $Q_1,\ldots,Q_n \in \GF[q]$. We do not assume the verifiers have access to these $Q_i$'s. \\

    \textbf{Secret-shared outputs:} Shares $\{[v_j],[u_j]\}_{j=0}^{b-1}$, where $b = \lceil \log (\beta_2-\beta_1+1) \rceil$ (for each $j$, each verifier outputs its respective shares, $(v_j^{(0)},u_j^{(0)})$ or $(v_j^{(1)},u_j^{(1)})$). \\

    \textbf{The Protocol:}

    The prover secret shares:
    \begin{enumerate}

    \item The $b$ bits $(v_j)_{j \in [0,\ldots,b-1]}$ of $V = (\sum_i \alpha_i Q_i) - \beta_1 \pmod{q}$,

    \item The $b$ bits $(u_j)_{j \in [0,\ldots,b-1]}$ of $U = \beta_2 - (\sum_i \alpha_i Q_i) \pmod{q}$.

    \end{enumerate}

    The verifiers verify the linear equality $(\sum_{j=0}^{b-1} v_j \cdot 2^j)+ ( \sum_{j=0}^{b-1} u_j \cdot 2^j) = \beta_2 - \beta_1 \pmod{q}$  (rejecting otherwise).

  \end{boxedminipage}

  \caption{From Range Check $\pmod{q}$ Protocol}
  \label{fig:inequality-mod-q-protocol}
\end{figure*}

\section{Deferred Proofs from \cref{subsec:wraparound-protocol}}
\label{app:wraparound}
\completeness*
\begin{proof} Recall that
$Y_i = Z_i X_i$. Thus, $Y_i$ has expectation 0 and is subgaussian\footnote{For background on subgaussian random variables, see e.g. \cite{Rig15}} with parameter $(X_i/\sqrt{2})$:
\begin{align*}
    \forall \lambda \in \mathbb{R}: \E [e^{\lambda Y_i}] & = \frac{1}{2} + \frac{1}{4} \cdot \left( e^{\lambda X_i} + e^{-\lambda X_i}\right) \\
    & \leq \frac{1}{2} + \frac{1}{2} \cdot \left( e^{\lambda^2 X_i^2 / 2}  \right) \\
    & \leq e^{\lambda^2 X_i^2 / 4}.
\end{align*}

Thus, $Y$, which is a sum of independent subgaussians, is itself subgaussian, with parameter $\sqrt{\sum_i X_i^2/2} \leq \sqrt{B/2}$. By tail bounds for subgaussian RVs:
\begin{align*}
    \Pr[|Y| > \alpha\sqrt{B}] & \leq 2\exp \left( - \frac{2\alpha^2 B}{2B}
    \right) \\
    & \leq 2 \exp(-\alpha^2).
\end{align*}
\end{proof}

\boundclose*
\begin{proof}
For every integer $t$, by Theorem \ref{thm:berry-esseen-Y}:
\begin{align*}
    \Pr &\left[ (Y - t \cdot q) \in [-\alpha\sqrt{B}, \alpha\sqrt{B}]  \right]\\
    & \leq \Pr \left[ \left( \mathcal{N}(0,\sigma^2) - t \cdot q \right) \in  [-\alpha\sqrt{B}, \alpha\sqrt{B}]  \right] +  \frac{1.12 \cdot  \alpha\sqrt{B}}{\sigma} \\
    & \leq \Pr \left[ \mathcal{N}(0,\sigma^2) \in  [-\alpha\sqrt{B}, \alpha\sqrt{B}]  \right] +  \frac{1.12 \cdot  \alpha\sqrt{B}}{\sigma} \\
    &\leq \left( 1.12 + \sqrt{\frac{2}{\pi}} \right) \frac{  \alpha\sqrt{B}}{\sigma}.
\end{align*}
The claim follows by taking a union bound over the possible values of $t$.
\end{proof}
\boundfar*
\begin{proof}

We use a union bound over the probability that $Y$ is too large and the probability that it is too small. To bound the former, Theorem \ref{thm:berry-esseen-Y}, implies:
\begin{align*}
    \Pr &\left[ Y \geq -\alpha\sqrt{B} + t \cdot q \right] \\ & \leq \Pr \left[  \mathcal{N}(0,\sigma^2)  \geq -\alpha\sqrt{B} + t \cdot q \right] + \frac{0.56 \alpha \sqrt{B}}{\sigma} \\
    & \leq \frac{\sigma}{ \sqrt{2 \pi} \cdot (-\alpha\sqrt{B} + t \cdot q) } \cdot \exp \left(-\frac{\left( -\alpha\sqrt{B} + t \cdot q \right)^2}{2\sigma^2} \right) \\&\;\;+ \frac{0.56 \alpha \sqrt{B}}{\sigma} \\
    & \leq \frac{1}{\sqrt{2\pi}} \exp \left(-\frac{\left( -\alpha\sqrt{B} + t \cdot q \right)^2}{2\sigma^2} \right) + \frac{0.56 \alpha \sqrt{B}}{\sigma},
\end{align*}
where the last inequality holds because by the conditions of the claim $\alpha, t \geq 1$ and $q \geq 2\alpha \sqrt{B}$. Thus we have that $-\alpha \sqrt{B} + t \cdot q \geq \alpha \sqrt{B} \geq \sigma$. The symmetric case (bounding the probability that $Y$ is too small) follows similarly.
\end{proof}

\claimunionbound*
\begin{proof}
We write
\begin{align*}
    & \Pr[Y \in [-\alpha \sqrt{B}, \alpha \sqrt{B}] \pmod{q}] \\
    & \leq \sum_{t=-(u-1)}^{u-1} \Pr \left[ Y \in [-\alpha\sqrt{B} + t \cdot q, \alpha\sqrt{B} + t \cdot q] \right] \\ & \,\,\,+ \Pr \left[ \left| Y \right| \geq  -\alpha\sqrt{B} + u \cdot q \right] \\
    & \leq \left( 2u \cdot 1.12 + (2u -1) \cdot \sqrt{\frac{2}{\pi}} \right) \cdot \frac{ \alpha \sqrt{B}}{\sigma} \\ & \,\,\, + \sqrt{\frac{2}{\pi}} \cdot \exp \left(-\frac{\left( -\alpha\sqrt{B} + u \cdot q \right)^2}{2\sigma^2} \right) \\
    & \leq \left( 2u \cdot \left(1.12 + \sqrt{\frac{2}{\pi}} \right) \right) \cdot \frac{ \alpha \sqrt{B}}{\sigma} \\ & \,\,\,+ \sqrt{\frac{2}{\pi}} \cdot \exp \left(-\frac{\left( -\alpha\sqrt{B} + u \cdot q \right)^2}{2\sigma^2} \right) \\
    &\leq \left( \alpha \sqrt{B} \cdot \left(1.12 + \sqrt{\frac{2}{\pi}} \right) \cdot \left( \frac{2 \sqrt{2 \ln(1/\delta)}}{q} +  \frac{2 \alpha \sqrt{B}}{q \cdot \sigma} \right) \right) + \delta  \\
    &\leq \left( \alpha \sqrt{B} \cdot \left(1.12 + \sqrt{\frac{2}{\pi}} \right) \cdot  \left( \frac{2 \sqrt{2 \ln(1/\delta)}}{q} +  \frac{0.02}{\sigma} \right) + \delta \right) \\
    & \leq \left( \alpha \sqrt{B} \cdot \left(1.12 + \sqrt{\frac{2}{\pi}} \right) \cdot \left(  \frac{2 \sqrt{ 2\ln(1/\delta)}  + 0.04 \sqrt{q}}{q} \right) + \delta \right),
\end{align*}
where above we used the fact that $\sigma > \sqrt{q/2}$ and also assumed that $q \geq 100 \alpha\sqrt{B}$ (this will be guaranteed by the parameters set below).
\end{proof}

\section{Deferred Proofs from \cref{subsec:putting-together}}
\label{app:putting-together}
\thmputtogether*

\begin{proof}

Completeness and soundness follow directly from the sub-protocols.  For the quadratic constraints protocol that is run Step (\ref{step:quadratic-check}) there are:
\begin{itemize}
    \item The number of constraints is:
    \begin{align*}
        & 1 + r + 2\lceil \log(B + 1) \rceil + r \cdot(2\lceil \log(2\sqrt{B \ln(2/\eta)} + 1)\rceil + 1) \\
        &\leq r \cdot 2 \lceil \log(4\sqrt{B \ln(2/\eta)} + 2)\rceil + 4\lceil \log(\sqrt{B} + 1) \rceil + 1\\
        &\leq 2 \lceil \log(4\sqrt{B \ln(2/\eta)} + 2)\rceil \cdot \left( r + 2 \right)
    \end{align*}

    \item Similarly, the number of variables is:
    \begin{align*}
        & d + 2r + 2\lceil \log(B + 1) \rceil + r \cdot(2\lceil \log(2\sqrt{B \ln(2/\eta)} + 1)\rceil + 1) \\
        &\leq d + r \cdot 2 \lceil \log(4\sqrt{B \ln(2/\eta)} + 2)\rceil + 4\lceil \log(\sqrt{B} + 1)  \rceil \\
        &\leq d + \left( 2 \lceil \log(4\sqrt{B \ln(2/\eta)} + 2)\rceil \cdot \left( r + 2 \right) \right).
    \end{align*}
\end{itemize}
Combining the soundness error of the wraparound and quadratic constraints protocol, the total soundness error is bounded by:
\begin{align*}
\errs&+ \left( \frac{ 2\sqrt{d + \left( 2 \lceil \log(4\sqrt{B \ln(2/\eta)} + 2)\rceil \cdot \left( r + 2 \right) \right)}}{q-\sqrt{d + \left( 2 \lceil \log(4\sqrt{B \ln(2/\eta)} + 2)\rceil \cdot \left( r + 2 \right) \right)}} \right.\\ & \,\,\,\left.+ \frac{2 \lceil \log(4\sqrt{B \ln(2/\eta)} + 2)\rceil \cdot \left( r + 2 \right)}{q} \right)^t.
\end{align*}
Since $q \geq \max \{81B\cdot \ln(2/\eta), 1000, 3r\}$, we have that $\lceil \log(4\sqrt{B \ln(2/\eta)} + 2) \rceil \leq \log(\sqrt{q}/2) \leq \log(q)/2$. The soundness error is thus bounded by:
\begin{align*}
\errs+ \left( \frac{ 2\sqrt{d + (\log(q)\cdot (r+2)/2)}}{q-2\sqrt{d + (\log(q)\cdot (r+2)/2)}} + \frac{\log(q)\cdot (r+2)}{2q} \right)^t.
\end{align*}

The claimed communication complexity follows from the complexity of the different sub-protocols, and the above bounds on the number of variables and constraints in the protocol run in Step (\ref{step:quadratic-check}). In particular:
\begin{itemize}

\item The second message, sent by the prover, is of length at most:
\begin{align*}
& \left(  \left ( \left \lceil \log \left( \sqrt{2 B \cdot \ln(2 / \eta)} + 1 \right) \right\rceil + 2 \right) \cdot r  \right.\\&\,\,\,\,\left.+ \left(  2\lceil \log(B +1) \rceil  \right) \right) \cdot \log (q) \\
&\leq \left(  \left( \log(q) \cdot r / 2\right)  + \left(  2 \log(q)  \right) \right) \cdot \log (q) \\
    &= \log^2(q) \cdot(\frac{r}{2}+2).
\end{align*}

\item The fourth message, sent by the prover, is of length at most:
\begin{align*}
 & t \cdot (4\sqrt{d + \left( 2 \lceil \log(4\sqrt{B \ln(2/\eta)} + 2)\rceil \cdot \left( r + 2 \right) \right)}+ 1 ) \cdot \log(q) \\
 &\leq t \cdot \left( 4\sqrt{d + (\log(q)\cdot (r+2)/2)} + 1 \right) \cdot \log(q).
\end{align*}
\end{itemize}

\changed{The zero-knowledge property follows directly from the zero-knowledge of all sub-protocols, using composition as explained in Section~\ref{sec:prelims}. We note that for all sub-protocols it is possible to complete any share to an accepting input, by setting the input vector $X$ to the all zero vector $0^d$. Doing this allows us to fix the values of all other variables accordingly. Specifically, for $1\le k\le \tau\cdot r$ the simulator sets  $g_k=1$, and for $\tau\cdot r +1\le k\le r$ sets $g_k = 0$. For each $k\in[r]$, the simulator 
    sets the bits $\{v_{k,j}\}_{j \in [b]}$ to be the Binary representation of the value $2\alpha \sqrt{B}$, and sets the bits $\{u_{k,j}\}_{j \in [b]}$ to all be $0$. It further sets 
 $s_k$ to be $0$. Finally, the simulator sets the bits$\{v'_{j'}\}$ to all be $0$ and the bits $\{u'_{j'}\}$ to the Binary representation of $B$.
 We stress that for all required applications of the quadratic constraints sub-protocol, it is easy to generate satisfying assignment for the constraints.} 
\end{proof}

\section{Differentially Private Secret Sharing}
\label{app:dzk}

We first recall a notion of near-indistinguishability used in
Differential Privacy:
\begin{definition}[$(\varepsilon,\delta)$-closeness]
Two random variables $P$ and $Q$ are said to be $(\varepsilon,\delta)$-close, denoted by $P \edclose Q$ if for all events $S$, it holds that
$\Pr[P \in S] \leq e^{\varepsilon}\cdot \Pr[Q \in S]+\delta$, and similarly, $Pr[Q \in S] \leq e^{\varepsilon}\cdot \Pr[P \in S]+\delta$.
\end{definition}

One can relax the secrecy requirements in cryptography to differential secrecy. Here we adopt the suitable notion for Zero Knowledge from~\cite{Talwar22} with the mild change of requiring the simulator to be efficient.

\begin{definition}[Differential Zero Knowledge]
We say a protocol $\pi$ is $(\varepsilon, \delta)$-Differentially Zero Knowledge w.r.t.~$L$ if there is \changed{an efficiently samplable}
distribution $Q$ such that for all $x \in L$, the distribution $\pi(x)$ of the protocol’s transcript on input $x$
satisfies $\pi(x)\edclose Q$.
\end{definition}

We remark that even if the distribution $Q$ in the above definition is {\em not} efficiently samplable, we can get an efficient simulator by paying (at most) a factor of 2 in the privacy parameters: 

\begin{remark}[A computationally bounded simulator]
If $\pi$ is an efficient protocol which is $(\varepsilon, \delta)$-Differentially Zero Knowledge w.r.t.~$L$ but with an {\em inefficient} simulator distribution $Q$, one can get an efficient simulator
distribution $Q'$ such that for all $x \in L$, the distribution $\pi(x)$ of the protocol’s transcript on input $x$
satisfies $\pi(x)\edcloseTwo Q'$. A sample from $Q'$ is obtained by running $\pi$ with some $x_0\in L$ and outputting the transcript. We assume that it is computationally easy to find {\em some} input $x_0 \in L$ (or, alternatively, the simulator can be non-uniform). 
\end{remark}

\begin{figure*} 
  \begin{boxedminipage}{\textwidth}
    \small \medskip \noindent

    \underline{\textbf{Differential Zero Knowledge $L_2$-Bound Protocol}}\\

    \textbf{Common inputs:} Dimension $d \in \Nt$, claimed bound $B \in \Nt$, field size $q \in \Nt$, parameters $\eps,\delta \in [0,1]$.  \\

    \textbf{Prover Input:} A vector $X \in \GF[q]^d$.\\

     \textbf{Claim (to be verified):} $\sum_{i=1}^d X_i^2 \leq B$ (where summation is over the integers). \\

    \textbf{The Protocol:}

    \begin{enumerate}

    \item Set $\sigma = c_{\eps, \frac{\delta}{2}}\cdot \sqrt{B}$, $\Delta = d\sigma^2(1+\frac{2\sqrt{\log 8e/\delta}}{\sqrt{d}} + \frac{2\log 8e/\delta}{d})$. Sample $R \sim \mathcal{N}^{\Delta}(0, \sigma^2 \mathbb{I})$, and define secret shares $X_j^{(0)} = -\ceil{R_j}$, $X_j^{(1)} = X_j + \ceil{R_j}$.
    Send to the verifiers these secret shares. Verifier 0 (resp. Verifier 1) will check that each received secret shares satisfies $|X^{(0)}|_2 \leq \sqrt{B}+\sqrt{\Delta}+\sqrt{d}$ (resp. $|X_j^{(1)}| \leq \sqrt{B}+\sqrt{\Delta}+\sqrt{d}$) and rejects otherwise.

    \item Run the protocol of Section \ref{subsec:inequality} to verify that $\sum_i X_i^2 \in [0,B] \pmod{q}$.

    The prover's message in this sub-protocol is sent in message 1 of the $L_2$-bound protocol. Results in alleged shares of bits $\{v'_{j'},u'_{j'}\}$.

    \item \label{step:quadratic-check-dzk} Use the quadratic constraints protocol of Corollary \ref{corollary:quadratic-repeated}, setting the number of repetitions $t$ as in Theorem \ref{thm:dzk-norm-bound}, to verify the following quadratic constraints:
    \begin{enumerate}
        \item $\sum_{i=1}^d X_i^2 = \sum_{j'} v'_{j'} \cdot 2^{j'} \pmod{q}$.

        \item the secret-shared values $\{v'_{j'},u'_{j'}\}$ are all bits (i.e. in $\bitset$).

    \end{enumerate}

    The two messages of this sub-protocol (the verifier sends the first message) are sent in messages 2 and 3 of the $L_2$-bound protocol.

    \end{enumerate}

    If the verifiers in any of the sub-protocols executed above reject, then the verifiers in the $L_2$-bound protocol reject immediately. Otherwise, they accept.

  \end{boxedminipage}

  \caption{Differential Zero Knowledge $L_2$-Bound Protocol}
  \label{fig:L2-protocol-dzk}
\end{figure*}

\section{A Distributed Fiat-Shamir Transform for Removing Interaction}
\label{sec:fiat-shamir}

We use a distributed variant of the Fiat Shamir transform to obtain a non-interactive proof system. This distributed variant was proposed in~\cite{BonehBCGI19}, adapting the Fiat-Shamir transform~\cite{FiatS86}, a technique for eliminating interaction in standalone interactive protocols, to the setting of interactive proofs on distributed data.

\paragraph{Fiat-Shamir in the standalone setting.} The original scheme was shown to be sound in the random oracle model (ROM), when applied to constant-round public-coin protocols \cite{PointchevalS96}. The basic idea behind the transform is to replace each of the (standalone) verifier's random challenges with a challenge obtained by applying  cryptographic hash function $H$ to the preceding transcript (in the random oracle model $H$ is taken to be a truly random function). This eliminates the need for interaction between the prover and the verifier. In terms of zero-knowledge, if the original protocol satisfies honest-verifier zero knowledge, then the transformed protocol is zero-knowledge in the random oracle model. It has also been shown that concrete instantiations of the hash function guarantee zero knowledge under mild ``programmability'' conditions \cite{CanettiCHLRR18,CanettiLW18}.

\paragraph{Distributed Fiat-Shamir.} As mentioned by~\cite{BonehBCGI19}, in the distributed proof setting, even for public-coin protocols where the challenges $r_i$ are indeed random (and public), applying the Fiat-Shamir transform is not straightforward. This is because both the input and the communication transcript are
distributed and cannot be revealed to any single verifier. To get around this difficulty, they propose to let the prover generate each random challenge $r_i$ based on the joint view of the verifiers in previous rounds, by adding blinding values to the view of each verifier. This way, the validity of a challenge can be verified jointly by the verifiers, where each verifier checks the correlation to its private view (being given also the appropriate blinding value).

For an  $m$-round zero-knowledge proof protocol with two verifiers and public challenges $r_1,\ldots, r_m$, the distributed variant of the Fiat-Shamir transform suggested by~\cite{BonehBCGI19} proceeds as follows. The prover derives each random challenge
$r_i$ as the hash of the random challenges $r_{i,0}$ and $r_{i,1}$, where $r_{i,j}$ is obtained by hashing the view of the verifier $V_j$ up to this point. Concretely, let
$H \colon \set{0,1}^*\mapsto \set{0,1}^\kappa$
be a random hash function, where
$\kappa$ is a security parameter. For each round $i\in[m]$ and $j \in\set{0,1}$, the prover $P$ selects the blinding value  $\nu_{i,j}\leftarrow \set{0,1}^{\kappa}$ and sets $r_{i,j}=H(i,j,x_j,\nu_{i,j},\pi_{i,j})$ (to be the part of the challenge that is verifiable by $V_j$). Finally, the prover sets the $i$'th challenge $r_i$ to be $r_{i}=H(i,\bot, r_{i,0},r_{i,1})$.

Note that the blinding values $\nu_{i,j}$ ensure that the hashes
$r_{i,j}$ do not leak any information about $\pi_{i,j}$. Furthermore, if the prover sends each verifier $V_j$ the blinding value $\nu_{i,j}$, then
the verifiers can verify the validity of the challenges by replaying the process in which the prover constructed them.

Finally, we emphasize that the transformed non-interactive proof system maintains the zero-knowledge guarantee of the original interactive protocol. If that protocol was public coins, and the only communication between the verifiers was exchanging a single message at the end (see Definition \ref{def:dZKIP}), then the non-interactive protocol is dZK even against malicious verifiers (each verifier's behavior does not effect the messages sent by the prover or by the other verifier). For concrete instantiations of the protocol, this will be true so long as the hash family is programmable, see \cite{CanettiCHLRR18,CanettiLW18} (this is a very mild condition). Thus, our transformed non-interactive proofs all satisfy zero-knowledge also against a single malicious verifier (we always assume that at least one verifier is honest).

\paragraph{Proving the norm with Differential Zero Knowledge.}
We next describe a way for a
client to share a vector $X \in L = \{X \in \mathbb{Z}^d : \|X\|_2^2 \leq B\}$, while preserving differential secrecy. Here, the client will share each $X_i$ by adding (rounded) truncated Gaussian noise\footnote{Alternately, one can also use a truncation of the Discrete Gaussian Mechanism~\cite{canonne2020discrete}, which will give similar bounds.} of magnitude large enough to guarantee differential privacy.

For parameters $\sigma,\Delta$, let $\mathcal{N}^{\Delta}(\mu, \sigma^2 \mathbb{I})$ be the multi-dimensional Gaussian distribution with variance $\sigma^2 \mathbb{I}$, conditioned on the $\ell_2$ norm being at most $\sqrt{\Delta}$.
The following lemma is standard; we provide a proof for completeness. 
\begin{lemma}
  \label{lem:gaussian-dzk}
    Let $\eps,\delta \in (0,1)$, and let $x \in \mathbb{Z}^d$ satisfy $\|x\|_2^2 \leq B$. Let $\sigma^2 = c_{\eps,\frac{\delta}{2}}^2 B$, where $c_{\eps, \delta} =  \sqrt{2\ln \frac{1.25}{\delta}} / \eps$ and let $\Delta \geq d\sigma^2(1+\frac{2\sqrt{\log 8e/\delta}}{\sqrt{d}} + \frac{2\log 8e/\delta}{d})$.
  Let $P \sim \mathcal{N}^{\Delta}(0, \sigma^2 \mathbb{I})$ and $Q \sim x + \mathcal{N}^{\Delta}(0, \sigma^2 \mathbb{I})$.
  Then $P \edclose Q$. It follows that $\lceil P \rceil \edclose \lceil Q \rceil$.
\end{lemma}
\begin{proof}
  Let $P'$ and $Q'$ denote samples from un-truncated Gaussians $\mathcal{N}(0, \sigma^2 \mathbb{I})$ and $ x + \mathcal{N}(0, \sigma^2 \mathbb{I})$ respectively.  With this $\sigma$, the privacy analysis of the Gaussian Mechanism~\cite{DworkR14} implies that without truncation, the two random variables $P'$ and $Q'$ would be $(\eps,\delta/2)$-close.
  Let $r$ the probability that a sample from $P'$ lies in the $\ell_2$ ball; We choose $\Delta$ so that the probability of truncation $(1-r)$ is at most $\delta/8e$. Indeed recall that the squared $\ell_2$ norm of $P'$ (scaled by $\sigma^2$) is a sample from a chi-square distribution $\chi_d^2$. We will use the following tail bounds for $\chi_k^2$ random variables from Laurent and Massart \cite[Lemma 1 rephrased]{LaurentM2000}:
\begin{theorem}
  \label{thm:chi2_tails}
  Let $Z$ be a $\chi_k^2$ random variable. Then for any $\beta >0$,
  \begin{align*}
    \Pr[\frac{1}{k}Z \geq 1 + 2\sqrt{\beta/k} + 2\beta/k] &\leq \exp(-\beta).
  \end{align*}
  \end{theorem}
  Setting $\beta = \delta/8e$, we conclude that $\Delta = d\sigma^2(1+\frac{2\sqrt{\log 8e/\delta}}{\sqrt{d}} + \frac{2\log 8e/\delta}{d})$ suffices.

  Now for any event $E$, we write
  \begin{align*}
  \Pr_P[E] &\leq r^{-1}\Pr_{P'}[E]\\
  &\leq r^{-1} (e^\eps \cdot \Pr_{Q'}[E] + \delta/2)\\
  &\leq r^{-1} (e^{\eps}\cdot (r\Pr_Q[E] + (1-r)) + \delta/2)\\
  &\leq e^{\eps}\cdot \Pr_Q[E] + r^{-1}((1-r)e^{\eps}+\delta/2)\\
  &\leq e^{\eps}\cdot \Pr_Q[E] + (8/7)((\delta/8e)(e)+\delta/2)\\
  &\leq e^{\eps}\cdot \Pr_Q[E] + \delta.
\end{align*}
Finally, the standard post-processing property of differential privacy implies that applying co-ordinate-wise ceiling to $P$ and $Q$ preserves $(\eps,\delta)$-closeness.
\end{proof}
We note that the parameter $c_{\eps,\frac{\delta}{2}}$ may be replaced by a tighter numerical bound using the Analytical Gaussian Mechanism~\cite{BalleW18} without impacting the result.

The client (prover) will sample truncated gaussian noise and take the ceiling to get $R$. It will then secret-share $X$ as $-R$ and $X+R$. Note that since $\|X\|_2 \leq \sqrt{B}$ and $\|R\|_{2} \leq \sqrt{\Delta}+\sqrt{d}$, it follows that $\max(\|-R\|_2, \|X+R\|_{2}) \leq \Lambda \eqdef \sqrt{B}+\sqrt{\Delta}+\sqrt{d}$.
For any shares $R^{(1)}, R^{(2)}$ with $\|R^{(i)}\|_2 \leq \Lambda$, it is immediate the $\|R^{(1)}+R^{(2)}\|_2^2 \leq  4\Lambda^2$.
Thus for $q > 4 \Lambda^2$, there can be no wraparound, and it suffices to validate that the squared norm modulo $q$ is bounded.

Note that for an honest prover, the secret sharing scheme always succeeds, and satisfies $(\eps,\delta)$-differential zero knowledge. Further, for $q > 4\Lambda^2$, there is no wraparound and thus the protocol inherits the soundness of the ``range check modulo $q$'' protocol. We estimate this lower bound on $q$ in terms of the parameters:
\begin{align*}
  4&\Lambda^2  \leq 4 (\sqrt{B} + \sqrt{\Delta} + \sqrt{d})^2\\
  &= 4\left(\sqrt{B} + \sqrt{d} + \sigma\sqrt{d\cdot(1+\frac{2\sqrt{\log 8e/\delta}}{\sqrt{d}} + \frac{2\log 8e/\delta}{d})}\right)^2\\
  &= 4\left(\sqrt{B} + \sqrt{d} + \sqrt{dBc_{\eps,\frac{\delta}{2}}\cdot(1+\frac{2\sqrt{\log 8e/\delta}}{\sqrt{d}} + \frac{2\log 8e/\delta}{d})}\right)^2\\
\end{align*}
\dzknormbound*
\begin{proof}

  Completeness and soundness follow directly from the sub-protocols. For the quadratic constraints protocol that is run Step (\ref{step:quadratic-check-dzk}) there are:
  \begin{itemize}
      \item The number of constraints is $1 + 2\lceil \log(B + 1) \rceil$.
      \item The number of variables is $d + 2\lceil \log(B + 1) \rceil$.
  \end{itemize}
  Thus from~\cref{lemma:quadratic}, the soundness error is:
  \begin{align*}
\left( \frac{ 2\sqrt{d + 2\lceil \log(B + 1) \rceil}}{q-\sqrt{d + 2\lceil \log(B + 1) \rceil}} + \frac{1 + 2\lceil \log(B + 1) \rceil}{q} \right)^t.
  \end{align*}
  The claimed communication complexity follows from the complexity of the different sub-protocols, and the above bounds on the number of variables and constraints in the protocol run in Step (\ref{step:quadratic-check-dzk}).

  We next argue the zero knowledge property. The simulator for verifier $1$ samples $R \sim \mathcal{N}^{\Delta}(0, \sigma^2 \mathbb{I})$, and outputs $-\ceil{R}$. This gives perfect zero knowledge for this step.
  The simulator for verifier $2$ samples $R \sim \mathcal{N}^{\Delta}(0, \sigma^2 \mathbb{I})$, and outputs $\ceil{R}$. \cref{lem:gaussian-dzk} implies that this step satisfies $(\eps,\delta)$-zero knowledge.
  The zero knowledge property of steps (2) and (3) follows from the ZK property of the respective subprotocols.
\end{proof}

\fi

\section{Further Performance Evaluations}
\label{sec:appendix-evaluations}

We provide further performance evaluations for the PINE protocols (statistical and differential ZK). We remark that throughout the appendix, the performance evaluations use a finer-grained search for parameters than the one used for the evaluations in Tables \ref{table:intro-performance}, \ref{table:perf-eval-szk -64}, \ref{table:perf-eval-szk-128}, \ref{table:perf-eval-dzk-64} that appear in the body of the paper.

\paragraph{Runtime evaluation.} While our main focus in this work is on the improved communication complexity of PINE, we also evaluate the runtime overhead for the prover and the verifier. 
Our main finding is that, for a wide range of parameters, the runtimes of both the prover and the verifier are improved by 1-2 orders of magnitude compared to prior work. In more detail: first, we evaluate the runtime in terms of the number of multiplications of field elements performed by statistical and differential PINE, with varying soundness errors and differential privacy parameter (for differential ZK PINE). We follow prior work in focusing on the number of multiplications, as these are the most expensive operation and they dominate the runtime. For a field of size $2^{64}$, we evaluate the Prover's runtime in Table \ref{table:perf-eval-comp-cost-prover-multiplications-comb-errors--200-64} and the verifier's runtime in Table \ref{table:perf-eval-comp-cost-verifier-multiplications-comb-errors--200-64}. For a field of size $2^{128}$, evaluate the Prover's runtime in Table \ref{table:perf-eval-comp-cost-prover-multiplications-comb-errors--200-128} and the verifier's runtime in Table \ref{table:perf-eval-comp-cost-verifier-multiplications-comb-errors--200-128}. We take the prior work of \cite{BonehBCGI19,libprio} as our baseline. Our main finding is that both variants of PINE obtain a 1-2 order of magnitude improvement in the prover's and the verifier's runtime compared to the baseline. These improvements are mostly due to avoiding the need to secret share each bit of each field element. For the most part, Differential PINE has a slightly lower runtime overhead for both the prover and the verifier. 

We also evaluate the total number of field operations performed by PINE: counting {\em both additions and multiplications}. This diverges from prior work, which focused on the number of multiplications (since these are more expensive operations), but we find it prudent because the statistical PINE protocol is ``addition-heavy'': a major step in the statistical PINE protocol, which is repeated many times, is computing sums of the form $\sum Z_i X_i$, where the $X_i$'s are (secret-shared) field elements, but the $Z_i$'s are always in $\{-1,0,1\}$. Thus, this sum is computed using only addition and subtraction operations (with no multiplication of general field elements). These runtime overheads in terms of total field operations are in Tables 
\ref{table:perf-eval-comp-cost-prover-multiplications-and-additions-comb-errors--200-64} (prover runtime, field size $2^{64}$), Table \ref{table:perf-eval-comp-cost-verifier-multiplications-and-additions-comb-errors--200-64} (verifier runtime, field size $2^{64}$), Table \ref{table:perf-eval-comp-cost-prover-multiplications-and-additions-comb-errors--200-128}  (prover runtime, field size $2^{128}$) and Table \ref{table:perf-eval-comp-cost-verifier-multiplications-and-additions-comb-errors--200-128}  (verifier runtime, field size $2^{128}$).

\ifconf
\begin{table*}[!t]
\else
\begin{table}[!t]
\fi

\centering
\begin{tabular}{| c | c || c | c | c | c | c | c | c | c |}
    \hline
    $\rho$ & Protocol & $d=10^4$ & $d=10^5$ & $d=10^6$ & $d=10^7$\\ [0.5ex]
    \hline\hline
    $2^{-50}$ & Prior work & $ 5.66 * 10^{7} $ & $ 8.1 * 10^{8} $ & $ 5.43 * 10^{9} $ & $ 7.62 * 10^{10} $ \\
    & Stat. PINE & $ {1.23\%}^{} $ & $ {0.81\%}^{} $ & $ {0.80\%}^{} $ & $ {0.80\%}^{} $ \\
    & $10^{-1}$-DP PINE & $ {0.79\%}^{} $ & $ {0.78\%}^{} $ & $ {0.79\%}^{*} $ & $ {0.79\%}^{*} $ \\
    & $10^{-2}$-DP PINE & $ {0.79\%}^{*} $ & $ {0.78\%}^{*} $ & $ {0.79\%}^{*} $ & $ {0.79\%}^{**} $ \\
    & $10^{-3}$-DP PINE & $ {0.79\%}^{*} $ & $ {0.78\%}^{**} $ & $ {0.79\%}^{**} $ & $ {0.79\%}^{***} $ \\
    \hline
    $2^{-100}$ & Prior work & $ 8.48 * 10^{7} $ & $ 1.21 * 10^{9} $ & $ 8.14 * 10^{9} $ & $ 1.14 * 10^{11} $ \\
    & Stat. PINE & $ {1.82\%}^{} $ & $ {1.09\%}^{} $ & $ {1.05\%}^{} $ & $ {1.05\%}^{} $ \\
    & $10^{-1}$-DP PINE & $ {1.05\%}^{} $ & $ {1.04\%}^{} $ & $ {1.04\%}^{*} $ & $ {1.05\%}^{*} $ \\
    & $10^{-2}$-DP PINE & $ {1.05\%}^{*} $ & $ {1.04\%}^{*} $ & $ {1.04\%}^{*} $ & $ {1.05\%}^{**} $ \\
    & $10^{-3}$-DP PINE & $ {1.05\%}^{*} $ & $ {1.04\%}^{**} $ & $ {1.04\%}^{**} $ & $ {1.05\%}^{***} $ \\
    \hline
    $2^{-200}$ & Prior work & $ 1.41 * 10^{8} $ & $ 2.02 * 10^{9} $ & $ 1.36 * 10^{10} $ & $ 2.29 * 10^{11} $ \\
    & Stat. PINE & $ {4.08\%}^{} $ & $ {1.44\%}^{} $ & $ {1.28\%}^{} $ & $ {1.05\%}^{} $ \\
    & $10^{-1}$-DP PINE & $ {1.25\%}^{} $ & $ {1.24\%}^{} $ & $ {1.25\%}^{*} $ & $ {1.05\%}^{*} $ \\
    & $10^{-2}$-DP PINE & $ {1.25\%}^{*} $ & $ {1.24\%}^{*} $ & $ {1.25\%}^{*} $ & $ {1.05\%}^{**} $ \\
    & $10^{-3}$-DP PINE & $ {1.25\%}^{*} $ & $ {1.24\%}^{**} $ & $ {1.25\%}^{**} $ & $ {0.79\%}^{***} $ \\ [1ex]
    \hline

\end{tabular}

\caption{{\textbf{\color{blue}Prover's runtime}}. In number of multiplications of field elements. Runtime for prior work (the benchmark) is in absolute numbers. Runtime for our protocols is in percentages compared to the benchmark. Parameters: {\textbf{\color{blue} field size $q \approx 2^{64}$, zero-knowledge error $\delta = 2^{-200}$}}. We note for Differential ZK, larger field size than $2^{64}$ is needed for larger $d$ and smaller $\epsilon$. * means we use field size $q \approx 2^{72}$, ** means we use field size $q \approx 2^{80}$, *** means we use field size $q \approx 2^{88}$.}
\label{table:perf-eval-comp-cost-prover-multiplications-comb-errors--200-64}

\medskip

\centering
\centering
\begin{tabular}{| c | c || c | c | c | c | c | c | c | c |}
    \hline
    $\rho$ & Protocol & $d=10^4$ & $d=10^5$ & $d=10^6$ & $d=10^7$\\ [0.5ex]
    \hline\hline
    $2^{-50}$ & Prior work & $ 6.34 * 10^{6} $ & $ 9.06 * 10^{7} $ & $ 6.07 * 10^{8} $ & $ 8.52 * 10^{9} $ \\
    & Stat. PINE & $ {1.81\%}^{} $ & $ {1.38\%}^{} $ & $ {1.36\%}^{} $ & $ {1.37\%}^{} $ \\
    & $10^{-1}$-DP PINE & $ {1.36\%}^{} $ & $ {1.35\%}^{} $ & $ {1.35\%}^{*} $ & $ {1.37\%}^{*} $ \\
    & $10^{-2}$-DP PINE & $ {1.36\%}^{*} $ & $ {1.35\%}^{*} $ & $ {1.35\%}^{*} $ & $ {1.37\%}^{**} $ \\
    & $10^{-3}$-DP PINE & $ {1.36\%}^{*} $ & $ {1.35\%}^{**} $ & $ {1.35\%}^{**} $ & $ {1.37\%}^{***} $ \\
    \hline
    $2^{-100}$ & Prior work & $ 9.51 * 10^{6} $ & $ 1.36 * 10^{8} $ & $ 9.1 * 10^{8} $ & $ 1.28 * 10^{10} $ \\
    & Stat. PINE & $ {2.62\%}^{} $ & $ {1.86\%}^{} $ & $ {1.81\%}^{} $ & $ {1.82\%}^{} $ \\
    & $10^{-1}$-DP PINE & $ {1.81\%}^{} $ & $ {1.80\%}^{} $ & $ {1.80\%}^{*} $ & $ {1.82\%}^{*} $ \\
    & $10^{-2}$-DP PINE & $ {1.81\%}^{*} $ & $ {1.80\%}^{*} $ & $ {1.80\%}^{*} $ & $ {1.82\%}^{**} $ \\
    & $10^{-3}$-DP PINE & $ {1.81\%}^{*} $ & $ {1.80\%}^{**} $ & $ {1.80\%}^{**} $ & $ {1.82\%}^{***} $ \\
    \hline
    $2^{-200}$ & Prior work & $ 1.59 * 10^{7} $ & $ 2.26 * 10^{8} $ & $ 1.52 * 10^{9} $ & $ 2.56 * 10^{10} $ \\
    & Stat. PINE & $ {5.12\%}^{} $ & $ {2.36\%}^{} $ & $ {2.20\%}^{} $ & $ {1.83\%}^{} $ \\
    & $10^{-1}$-DP PINE & $ {2.17\%}^{} $ & $ {2.16\%}^{} $ & $ {2.17\%}^{*} $ & $ {1.82\%}^{*} $ \\
    & $10^{-2}$-DP PINE & $ {2.17\%}^{*} $ & $ {2.16\%}^{*} $ & $ {2.17\%}^{*} $ & $ {1.82\%}^{**} $ \\
    & $10^{-3}$-DP PINE & $ {2.17\%}^{*} $ & $ {2.16\%}^{**} $ & $ {2.17\%}^{**} $ & $ {1.37\%}^{***} $ \\ [1ex]
    \hline

\end{tabular}

\caption{{\textbf{\color{blue}Verifier's runtime}}. In number of multiplications of field elements. Runtime for prior work (the benchmark) is in absolute numbers. Runtime for our protocols is in percentages compared to the benchmark.  Parameters: {\textbf{\color{blue} field size $q \approx 2^{64}$, zero-knowledge error $\delta = 2^{-200}$}}. We note for Differential ZK, larger field size than $2^{64}$ is needed for larger $d$ and smaller $\epsilon$. * means we use field size $q \approx 2^{72}$, ** means we use field size $q \approx 2^{80}$, *** means we use field size $q \approx 2^{88}$.}
\label{table:perf-eval-comp-cost-verifier-multiplications-comb-errors--200-64}

\ifconf 
\end{table*} 
\else 
\end{table}
\fi

\ifconf
\begin{table*}[!t]
\else
\begin{table}[!t]
\fi

\centering
\begin{tabular}{| c | c || c | c | c | c | c | c | c | c |}
    \hline
    $\rho$ & Protocol & $d=10^4$ & $d=10^5$ & $d=10^6$ & $d=10^7$\\ [0.5ex]
    \hline\hline
    $2^{-50}$ & Prior work & $ 2.83 * 10^{7} $ & $ 4.05 * 10^{8} $ & $ 2.71 * 10^{9} $ & $ 3.81 * 10^{10} $ \\
    & Stat. PINE & $ {2.43\%}^{} $ & $ {1.62\%}^{} $ & $ {1.59\%}^{} $ & $ {1.59\%}^{} $ \\
    & $10^{-1}$-DP PINE & $ {1.59\%}^{} $ & $ {1.57\%}^{} $ & $ {1.58\%}^{} $ & $ {1.59\%}^{} $ \\
    & $10^{-2}$-DP PINE & $ {1.59\%}^{} $ & $ {1.57\%}^{} $ & $ {1.58\%}^{} $ & $ {1.59\%}^{} $ \\
    & $10^{-3}$-DP PINE & $ {1.59\%}^{} $ & $ {1.57\%}^{} $ & $ {1.58\%}^{} $ & $ {1.59\%}^{} $ \\
    \hline
    $2^{-100}$ & Prior work & $ 2.83 * 10^{7} $ & $ 4.05 * 10^{8} $ & $ 2.71 * 10^{9} $ & $ 3.81 * 10^{10} $ \\
    & Stat. PINE & $ {2.75\%}^{} $ & $ {1.65\%}^{} $ & $ {1.60\%}^{} $ & $ {1.59\%}^{} $ \\
    & $10^{-1}$-DP PINE & $ {1.59\%}^{} $ & $ {1.57\%}^{} $ & $ {1.58\%}^{} $ & $ {1.59\%}^{} $ \\
    & $10^{-2}$-DP PINE & $ {1.59\%}^{} $ & $ {1.57\%}^{} $ & $ {1.58\%}^{} $ & $ {1.59\%}^{} $ \\
    & $10^{-3}$-DP PINE & $ {1.59\%}^{} $ & $ {1.57\%}^{} $ & $ {1.58\%}^{} $ & $ {1.59\%}^{} $ \\
    \hline
    $2^{-200}$ & Prior work & $ 5.66 * 10^{7} $ & $ 8.1 * 10^{8} $ & $ 5.43 * 10^{9} $ & $ 7.62 * 10^{10} $ \\
    & Stat. PINE & $ {5.11\%}^{} $ & $ {1.80\%}^{} $ & $ {1.60\%}^{} $ & $ {1.58\%}^{} $ \\
    & $10^{-1}$-DP PINE & $ {1.57\%}^{} $ & $ {1.55\%}^{} $ & $ {1.57\%}^{} $ & $ {1.58\%}^{} $ \\
    & $10^{-2}$-DP PINE & $ {1.57\%}^{} $ & $ {1.55\%}^{} $ & $ {1.57\%}^{} $ & $ {1.58\%}^{} $ \\
    & $10^{-3}$-DP PINE & $ {1.57\%}^{} $ & $ {1.55\%}^{} $ & $ {1.57\%}^{} $ & $ {1.58\%}^{} $ \\ [1ex]
    \hline

\end{tabular}

\caption{{\textbf{\color{blue}Prover's runtime}}. In number of multiplications of field elements. Runtime for prior work (the benchmark) is in absolute numbers. Runtime for our protocols is in percentages compared to the benchmark. Parameters: {\textbf{\color{blue} field size $q \approx 2^{128}$, zero-knowledge error $\delta = 2^{-200}$}}. }
\label{table:perf-eval-comp-cost-prover-multiplications-comb-errors--200-128}

\centering
\begin{tabular}{| c | c || c | c | c | c | c | c | c | c |}
    \hline
    $\rho$ & Protocol & $d=10^4$ & $d=10^5$ & $d=10^6$ & $d=10^7$\\ [0.5ex]
    \hline\hline
    $2^{-50}$ & Prior work & $ 3.17 * 10^{6} $ & $ 4.53 * 10^{7} $ & $ 3.03 * 10^{8} $ & $ 4.26 * 10^{9} $ \\
    & Stat. PINE & $ {3.60\%}^{} $ & $ {2.76\%}^{} $ & $ {2.72\%}^{} $ & $ {2.74\%}^{} $ \\
    & $10^{-1}$-DP PINE & $ {2.72\%}^{} $ & $ {2.70\%}^{} $ & $ {2.71\%}^{} $ & $ {2.74\%}^{} $ \\
    & $10^{-2}$-DP PINE & $ {2.72\%}^{} $ & $ {2.70\%}^{} $ & $ {2.71\%}^{} $ & $ {2.74\%}^{} $ \\
    & $10^{-3}$-DP PINE & $ {2.72\%}^{} $ & $ {2.70\%}^{} $ & $ {2.71\%}^{} $ & $ {2.74\%}^{} $ \\
    \hline
    $2^{-100}$ & Prior work & $ 3.17 * 10^{6} $ & $ 4.53 * 10^{7} $ & $ 3.03 * 10^{8} $ & $ 4.26 * 10^{9} $ \\
    & Stat. PINE & $ {3.93\%}^{} $ & $ {2.78\%}^{} $ & $ {2.72\%}^{} $ & $ {2.74\%}^{} $ \\
    & $10^{-1}$-DP PINE & $ {2.72\%}^{} $ & $ {2.70\%}^{} $ & $ {2.71\%}^{} $ & $ {2.74\%}^{} $ \\
    & $10^{-2}$-DP PINE & $ {2.72\%}^{} $ & $ {2.70\%}^{} $ & $ {2.71\%}^{} $ & $ {2.74\%}^{} $ \\
    & $10^{-3}$-DP PINE & $ {2.72\%}^{} $ & $ {2.70\%}^{} $ & $ {2.71\%}^{} $ & $ {2.74\%}^{} $ \\
    \hline
    $2^{-200}$ & Prior work & $ 6.34 * 10^{6} $ & $ 9.06 * 10^{7} $ & $ 6.07 * 10^{8} $ & $ 8.52 * 10^{9} $ \\
    & Stat. PINE & $ {6.40\%}^{} $ & $ {2.96\%}^{} $ & $ {2.75\%}^{} $ & $ {2.74\%}^{} $ \\
    & $10^{-1}$-DP PINE & $ {2.72\%}^{} $ & $ {2.70\%}^{} $ & $ {2.71\%}^{} $ & $ {2.74\%}^{} $ \\
    & $10^{-2}$-DP PINE & $ {2.72\%}^{} $ & $ {2.70\%}^{} $ & $ {2.71\%}^{} $ & $ {2.74\%}^{} $ \\
    & $10^{-3}$-DP PINE & $ {2.72\%}^{} $ & $ {2.70\%}^{} $ & $ {2.71\%}^{} $ & $ {2.74\%}^{} $ \\ [1ex]
    \hline

\end{tabular}

\caption{{\textbf{\color{blue}Verifier's runtime}}. In number of multiplications of field elements. Runtime for prior work (the benchmark) is in absolute numbers. Runtime for our protocols is in percentages compared to the benchmark. Parameters: {\textbf{\color{blue} field size $q \approx 2^{128}$, zero-knowledge error $\delta = 2^{-200}$}}. }
\label{table:perf-eval-comp-cost-verifier-multiplications-comb-errors--200-128}

\ifconf 
\end{table*} 
\else 
\end{table}
\fi

\ifconf
\begin{table*}[!t]
\else
\begin{table}[!t]
\fi

\centering
\begin{tabular}{| c | c || c | c | c | c | c | c | c | c |}
    \hline
    $\rho$ & Protocol & $d=10^4$ & $d=10^5$ & $d=10^6$ & $d=10^7$\\ [0.5ex]
    \hline\hline
    $2^{-50}$ & Prior work & $ 1.4 * 10^{8} $ & $ 2.05 * 10^{9} $ & $ 1.39 * 10^{10} $ & $ 1.99 * 10^{11} $ \\
    & Stat. PINE & $ {1.60\%}^{} $ & $ {1.08\%}^{} $ & $ {1.21\%}^{} $ & $ {1.08\%}^{} $ \\
    & $10^{-1}$-DP PINE & $ {0.75\%}^{} $ & $ {0.75\%}^{} $ & $ {0.76\%}^{*} $ & $ {0.77\%}^{*} $ \\
    & $10^{-2}$-DP PINE & $ {0.75\%}^{*} $ & $ {0.75\%}^{*} $ & $ {0.76\%}^{*} $ & $ {0.77\%}^{**} $ \\
    & $10^{-3}$-DP PINE & $ {0.75\%}^{*} $ & $ {0.75\%}^{**} $ & $ {0.76\%}^{**} $ & $ {0.77\%}^{***} $ \\
    \hline
    $2^{-100}$ & Prior work & $ 2.1 * 10^{8} $ & $ 3.08 * 10^{9} $ & $ 2.09 * 10^{10} $ & $ 2.98 * 10^{11} $ \\
    & Stat. PINE & $ {2.26\%}^{} $ & $ {1.41\%}^{} $ & $ {1.55\%}^{} $ & $ {1.39\%}^{} $ \\
    & $10^{-1}$-DP PINE & $ {0.99\%}^{} $ & $ {0.99\%}^{} $ & $ {1.00\%}^{*} $ & $ {1.02\%}^{*} $ \\
    & $10^{-2}$-DP PINE & $ {0.99\%}^{*} $ & $ {0.99\%}^{*} $ & $ {1.00\%}^{*} $ & $ {1.02\%}^{**} $ \\
    & $10^{-3}$-DP PINE & $ {0.99\%}^{*} $ & $ {0.99\%}^{**} $ & $ {1.00\%}^{**} $ & $ {1.02\%}^{***} $ \\
    \hline
    $2^{-200}$ & Prior work & $ 3.5 * 10^{8} $ & $ 5.14 * 10^{9} $ & $ 3.48 * 10^{10} $ & $ 5.96 * 10^{11} $ \\
    & Stat. PINE & $ {4.54\%}^{} $ & $ {1.79\%}^{} $ & $ {1.84\%}^{} $ & $ {1.37\%}^{} $ \\
    & $10^{-1}$-DP PINE & $ {1.18\%}^{} $ & $ {1.18\%}^{} $ & $ {1.20\%}^{*} $ & $ {1.01\%}^{*} $ \\
    & $10^{-2}$-DP PINE & $ {1.18\%}^{*} $ & $ {1.18\%}^{*} $ & $ {1.20\%}^{*} $ & $ {1.01\%}^{**} $ \\
    & $10^{-3}$-DP PINE & $ {1.18\%}^{*} $ & $ {1.18\%}^{**} $ & $ {1.20\%}^{**} $ & $ {0.76\%}^{***} $ \\ [1ex]
    \hline

\end{tabular}

\caption{{\textbf{\color{blue} Prover's runtime: number of multiplications and additions of field elements}}. Runtime for prior work (the benchmark) is in absolute numbers. Runtime for our protocols is in percentages compared to the benchmark. Parameters: {\textbf{\color{blue} field size $q \approx 2^{64}$, zero-knowledge error $\delta = 2^{-200}$}}. For Differential ZK, larger field size is needed for larger $d$ and smaller $\epsilon$. * means we use field size $q \approx 2^{72}$, ** means we use field size $q \approx 2^{80}$, *** means we use field size $q \approx 2^{88}$.}
\label{table:perf-eval-comp-cost-prover-multiplications-and-additions-comb-errors--200-64}

\centering
\begin{tabular}{| c | c || c | c | c | c | c | c | c | c |}
    \hline
    $\rho$ & Protocol & $d=10^4$ & $d=10^5$ & $d=10^6$ & $d=10^7$\\ [0.5ex]
    \hline\hline
    $2^{-50}$ & Prior work & $ 1.48 * 10^{7} $ & $ 2.18 * 10^{8} $ & $ 1.48 * 10^{9} $ & $ 2.13 * 10^{10} $ \\
    & Stat. PINE & $ {5.92\%}^{} $ & $ {4.18\%}^{} $ & $ {5.51\%}^{} $ & $ {4.25\%}^{} $ \\
    & $10^{-1}$-DP PINE & $ {1.30\%}^{} $ & $ {1.31\%}^{} $ & $ {1.32\%}^{*} $ & $ {1.34\%}^{*} $ \\
    & $10^{-2}$-DP PINE & $ {1.30\%}^{*} $ & $ {1.31\%}^{*} $ & $ {1.32\%}^{*} $ & $ {1.34\%}^{**} $ \\
    & $10^{-3}$-DP PINE & $ {1.30\%}^{*} $ & $ {1.31\%}^{**} $ & $ {1.32\%}^{**} $ & $ {1.34\%}^{***} $ \\
    \hline
    $2^{-100}$ & Prior work & $ 2.22 * 10^{7} $ & $ 3.27 * 10^{8} $ & $ 2.22 * 10^{9} $ & $ 3.19 * 10^{10} $ \\
    & Stat. PINE & $ {7.59\%}^{} $ & $ {5.24\%}^{} $ & $ {6.85\%}^{} $ & $ {5.32\%}^{} $ \\
    & $10^{-1}$-DP PINE & $ {1.74\%}^{} $ & $ {1.74\%}^{} $ & $ {1.75\%}^{*} $ & $ {1.78\%}^{*} $ \\
    & $10^{-2}$-DP PINE & $ {1.74\%}^{*} $ & $ {1.74\%}^{*} $ & $ {1.75\%}^{*} $ & $ {1.78\%}^{**} $ \\
    & $10^{-3}$-DP PINE & $ {1.74\%}^{*} $ & $ {1.74\%}^{**} $ & $ {1.75\%}^{**} $ & $ {1.78\%}^{***} $ \\
    \hline
    $2^{-200}$ & Prior work & $ 3.7 * 10^{7} $ & $ 5.46 * 10^{8} $ & $ 3.7 * 10^{9} $ & $ 6.39 * 10^{10} $ \\
    & Stat. PINE & $ {10.72\%}^{} $ & $ {6.22\%}^{} $ & $ {7.94\%}^{} $ & $ {5.15\%}^{} $ \\
    & $10^{-1}$-DP PINE & $ {2.08\%}^{} $ & $ {2.09\%}^{} $ & $ {2.11\%}^{*} $ & $ {1.78\%}^{*} $ \\
    & $10^{-2}$-DP PINE & $ {2.08\%}^{*} $ & $ {2.09\%}^{*} $ & $ {2.11\%}^{*} $ & $ {1.78\%}^{**} $ \\
    & $10^{-3}$-DP PINE & $ {2.08\%}^{*} $ & $ {2.09\%}^{**} $ & $ {2.11\%}^{**} $ & $ {1.34\%}^{***} $ \\ [1ex]
    \hline

\end{tabular}

\caption{{\textbf{\color{blue} Verifier's runtime: number of multiplications and additions of field elements}}. Runtime for prior work (the benchmark) is in absolute numbers. Runtime for our protocols is in percentages compared to the benchmark. Parameters: {\textbf{\color{blue} field size $q \approx 2^{64}$, zero-knowledge error $\delta = 2^{-200}$}}. For Differential ZK, larger field size is needed for larger $d$ and smaller $\epsilon$. * means we use field size $q \approx 2^{72}$, ** means we use field size $q \approx 2^{80}$, *** means we use field size $q \approx 2^{88}$.}
\label{table:perf-eval-comp-cost-verifier-multiplications-and-additions-comb-errors--200-64}

\ifconf 
\end{table*} 
\else 
\end{table}
\fi

\ifconf
\begin{table*}[!t]
\else
\begin{table}[!t]
\fi

\centering
\begin{tabular}{| c | c || c | c | c | c | c | c | c | c |}
    \hline
    $\rho$ & Protocol & $d=10^4$ & $d=10^5$ & $d=10^6$ & $d=10^7$\\ [0.5ex]
    \hline\hline
    $2^{-50}$ & Prior work & $ 7.0 * 10^{7} $ & $ 1.03 * 10^{9} $ & $ 6.95 * 10^{9} $ & $ 9.93 * 10^{10} $ \\
    & Stat. PINE & $ {3.16\%}^{} $ & $ {2.14\%}^{} $ & $ {2.40\%}^{} $ & $ {2.15\%}^{} $ \\
    & $10^{-1}$-DP PINE & $ {1.49\%}^{} $ & $ {1.49\%}^{} $ & $ {1.52\%}^{} $ & $ {1.53\%}^{} $ \\
    & $10^{-2}$-DP PINE & $ {1.49\%}^{} $ & $ {1.49\%}^{} $ & $ {1.52\%}^{} $ & $ {1.53\%}^{} $ \\
    & $10^{-3}$-DP PINE & $ {1.49\%}^{} $ & $ {1.49\%}^{} $ & $ {1.52\%}^{} $ & $ {1.53\%}^{} $ \\
    \hline
    $2^{-100}$ & Prior work & $ 7.0 * 10^{7} $ & $ 1.03 * 10^{9} $ & $ 6.95 * 10^{9} $ & $ 9.93 * 10^{10} $ \\
    & Stat. PINE & $ {4.22\%}^{} $ & $ {2.67\%}^{} $ & $ {3.15\%}^{} $ & $ {2.67\%}^{} $ \\
    & $10^{-1}$-DP PINE & $ {1.49\%}^{} $ & $ {1.49\%}^{} $ & $ {1.52\%}^{} $ & $ {1.53\%}^{} $ \\
    & $10^{-2}$-DP PINE & $ {1.49\%}^{} $ & $ {1.49\%}^{} $ & $ {1.52\%}^{} $ & $ {1.53\%}^{} $ \\
    & $10^{-3}$-DP PINE & $ {1.49\%}^{} $ & $ {1.49\%}^{} $ & $ {1.52\%}^{} $ & $ {1.53\%}^{} $ \\
    \hline
    $2^{-200}$ & Prior work & $ 1.4 * 10^{8} $ & $ 2.05 * 10^{9} $ & $ 1.39 * 10^{10} $ & $ 1.99 * 10^{11} $ \\
    & Stat. PINE & $ {6.45\%}^{} $ & $ {2.76\%}^{} $ & $ {3.08\%}^{} $ & $ {2.61\%}^{} $ \\
    & $10^{-1}$-DP PINE & $ {1.48\%}^{} $ & $ {1.48\%}^{} $ & $ {1.50\%}^{} $ & $ {1.52\%}^{} $ \\
    & $10^{-2}$-DP PINE & $ {1.48\%}^{} $ & $ {1.48\%}^{} $ & $ {1.50\%}^{} $ & $ {1.52\%}^{} $ \\
    & $10^{-3}$-DP PINE & $ {1.48\%}^{} $ & $ {1.48\%}^{} $ & $ {1.50\%}^{} $ & $ {1.52\%}^{} $ \\ [1ex]
    \hline

\end{tabular}

\caption{{\textbf{\color{blue} Prover's runtime: number of multiplications and additions of field elements}}. Runtime for prior work (the benchmark) is in absolute numbers. Runtime for our protocols is in percentages compared to the benchmark. Parameters: {\textbf{\color{blue} field size $q \approx 2^{128}$, zero-knowledge error $\delta = 2^{-200}$}}.}
\label{table:perf-eval-comp-cost-prover-multiplications-and-additions-comb-errors--200-128}

\centering
\begin{tabular}{| c | c || c | c | c | c | c | c | c | c |}
    \hline
    $\rho$ & Protocol & $d=10^4$ & $d=10^5$ & $d=10^6$ & $d=10^7$\\ [0.5ex]
    \hline\hline
    $2^{-50}$ & Prior work & $ 7.4 * 10^{6} $ & $ 1.09 * 10^{8} $ & $ 7.41 * 10^{8} $ & $ 1.06 * 10^{10} $ \\
    & Stat. PINE & $ {11.68\%}^{} $ & $ {8.26\%}^{} $ & $ {10.88\%}^{} $ & $ {8.41\%}^{} $ \\
    & $10^{-1}$-DP PINE & $ {2.60\%}^{} $ & $ {2.61\%}^{} $ & $ {2.63\%}^{} $ & $ {2.68\%}^{} $ \\
    & $10^{-2}$-DP PINE & $ {2.60\%}^{} $ & $ {2.61\%}^{} $ & $ {2.63\%}^{} $ & $ {2.68\%}^{} $ \\
    & $10^{-3}$-DP PINE & $ {2.60\%}^{} $ & $ {2.61\%}^{} $ & $ {2.63\%}^{} $ & $ {2.68\%}^{} $ \\
    \hline
    $2^{-100}$ & Prior work & $ 7.4 * 10^{6} $ & $ 1.09 * 10^{8} $ & $ 7.41 * 10^{8} $ & $ 1.06 * 10^{10} $ \\
    & Stat. PINE & $ {19.02\%}^{} $ & $ {13.04\%}^{} $ & $ {17.90\%}^{} $ & $ {13.29\%}^{} $ \\
    & $10^{-1}$-DP PINE & $ {2.60\%}^{} $ & $ {2.61\%}^{} $ & $ {2.63\%}^{} $ & $ {2.68\%}^{} $ \\
    & $10^{-2}$-DP PINE & $ {2.60\%}^{} $ & $ {2.61\%}^{} $ & $ {2.63\%}^{} $ & $ {2.68\%}^{} $ \\
    & $10^{-3}$-DP PINE & $ {2.60\%}^{} $ & $ {2.61\%}^{} $ & $ {2.63\%}^{} $ & $ {2.68\%}^{} $ \\
    \hline
    $2^{-200}$ & Prior work & $ 1.48 * 10^{7} $ & $ 2.18 * 10^{8} $ & $ 1.48 * 10^{9} $ & $ 2.13 * 10^{10} $ \\
    & Stat. PINE & $ {20.67\%}^{} $ & $ {12.70\%}^{} $ & $ {17.18\%}^{} $ & $ {12.77\%}^{} $ \\
    & $10^{-1}$-DP PINE & $ {2.60\%}^{} $ & $ {2.61\%}^{} $ & $ {2.63\%}^{} $ & $ {2.68\%}^{} $ \\
    & $10^{-2}$-DP PINE & $ {2.60\%}^{} $ & $ {2.61\%}^{} $ & $ {2.63\%}^{} $ & $ {2.68\%}^{} $ \\
    & $10^{-3}$-DP PINE & $ {2.60\%}^{} $ & $ {2.61\%}^{} $ & $ {2.63\%}^{} $ & $ {2.68\%}^{} $ \\ [1ex]
    \hline

\end{tabular}

\caption{{\textbf{\color{blue} Verifier's runtime: number of multiplications and additions of field elements}}. Runtime for prior work (the benchmark) is in absolute numbers. Runtime for our protocols is in percentages compared to the benchmark. Parameters: {\textbf{\color{blue} field size $q \approx 2^{128}$, zero-knowledge error $\delta = 2^{-200}$}}.}
\label{table:perf-eval-comp-cost-verifier-multiplications-and-additions-comb-errors--200-128}

\ifconf 
\end{table*} 
\else 
\end{table}
\fi

\paragraph{Further evaluation of communication overhead.} We provide further evaluation of PINE's communication overhead with different combinations of soundness error and statistical zero-knowledge error (the completeness error equals the statistical zero-knowledge error throughout). These evaluations are for statistical PINE, and show the overhead over the naive baseline (no robustness).
Tables \ref{table:perf-eval-comb-errors-64} and \ref{table:perf-eval-comb-errors-128} provide further information beyond what was presented in Tables 1-3 in the introduction.

\ifconf
\begin{table*}[!t]
\else
\begin{table}[!t]
\fi

\centering

\begin{tabular}{| c | c || c | c | c | c |}
    \hline
    $\rho$ & $\delta$ & $d=10^4$ & $d=10^5$ & $d=10^6$ & $d=10^7$ \\ [0.5ex]
    \hline\hline
    $2^{-50}$ & $2^{-50}$ & $17.87\%$ & $2.77\%$ & $0.45\%$ & $0.13\%$\\
    & $2^{-100}$ & $18.96\%$ & $2.90\%$ & $0.46\%$ & $0.13\%$\\
    & $2^{-200}$ & $20.26\%$ & $3.01\%$ & $0.47\%$ & $0.13\%$\\
    \hline
    $2^{-100}$ & $2^{-50}$ & $35.55\%$ & $5.52\%$ & $0.89\%$ & $0.26\%$\\
    & $2^{-100}$ & $36.85\%$ & $5.65\%$ & $0.91\%$ & $0.26\%$\\
    & $2^{-200}$ & $38.19\%$ & $5.79\%$ & $0.92\%$ & $0.26\%$\\
    \hline
    $2^{-200}$ & $2^{-50}$ & $79.71\%$ & $11.91\%$ & $1.87\%$ & $0.52\%$\\
    & $2^{-100}$ & $81.26\%$ & $12.06\%$ & $1.89\%$ & $0.52\%$\\
    & $2^{-200}$ & $84.17\%$ & $12.35\%$ & $1.92\%$ & $0.52\%$\\ [1ex]
    \hline

\end{tabular}
\caption{Statistical PINE Communication: overhead compared to non-robust sharing (see Table \ref{table:intro-performance}), different combinations of soundness error $\rho$ and zero-knowledge error $\delta$. {\textbf{\color{blue} Field size $q \approx 2^{64}$}}.}
\label{table:perf-eval-comb-errors-64}

\begin{tabular}{| c | c || c | c | c | c |}
    \hline
    $\rho$ & $\delta$ & $d=10^4$ & $d=10^5$ & $d=10^6$ & $d=10^7$ \\ [0.5ex]
    \hline\hline
    $2^{-50}$ & $2^{-50}$ & $17.66\%$ & $2.75\%$ & $0.45\%$ & $0.13\%$\\
    & $2^{-100}$ & $18.96\%$ & $2.88\%$ & $0.46\%$ & $0.13\%$\\
    & $2^{-200}$ & $20.03\%$ & $2.99\%$ & $0.47\%$ & $0.13\%$\\
    \hline
    $2^{-100}$ & $2^{-50}$ & $28.69\%$ & $3.85\%$ & $0.56\%$ & $0.14\%$\\
    & $2^{-100}$ & $29.97\%$ & $3.98\%$ & $0.57\%$ & $0.14\%$\\
    & $2^{-200}$ & $31.27\%$ & $4.11\%$ & $0.58\%$ & $0.14\%$\\
    \hline
    $2^{-200}$ & $2^{-50}$ & $61.27\%$ & $8.09\%$ & $1.15\%$ & $0.28\%$\\
    & $2^{-100}$ & $62.78\%$ & $8.25\%$ & $1.17\%$ & $0.28\%$\\
    & $2^{-200}$ & $64.97\%$ & $8.46\%$ & $1.19\%$ & $0.28\%$\\ [1ex]
    \hline

\end{tabular}
\caption{Statistical PINE Communication: overhead compared to non-robust sharing (see Table \ref{table:intro-performance}), different combinations of soundness error $\rho$ and zero-knowledge error $\delta$. {\textbf{\color{blue} Field size $q \approx 2^{128}$}}.}
\label{table:perf-eval-comb-errors-128}

\ifconf 
\end{table*} 
\else 
\end{table}
\fi

\paragraph{Verbose evaluation tables.} To give further details on some of the parameter choices in PINE, we also provide more detailed tables that show the communication overhead for different choices of soundness and zero-knowledge errors (similarly to Tables \ref{table:perf-eval-comb-errors-64} and \ref{table:perf-eval-comb-errors-128}), but also detail the internal choices of parameters such as $t$, the number of repetitions of the quadratic constraints protocol, $r$, the number of repetitions of the wraparound protocol, and $\tau$, the fraction of those repetitions that should succeed for the verifier to accept. These results are detailed in Table \ref{table:perf-eval-comb-errors-64-verbose} (for field size $2^{64}$) and Table \ref{table:perf-eval-comb-errors-128-verbose} (for field size $2^{128}$).

\ifconf
\begin{table*}[!t]
\else
\begin{table}[!t]
\fi

\centering

\begin{tabular}{|c | c | c || c | c | c | c |}
    \hline
    $\rho$ & $\delta$ & $d$ & $t$ & $r$ & $\tau$ & Overhead \\ [0.5ex]
    \hline\hline
    $2^{-50}$ & $2^{-50}$ & $10^4$ & $1$ & $51$ & $1.0000$ & $17.87\%$ \\
    & & $10^5$ & $1$ & $51$ & $1.0000$ & $2.77\%$ \\
    & & $10^6$ & $1$ & $51$ & $1.0000$ & $0.45\%$ \\
    & & $10^7$ & $1$ & $51$ & $1.0000$ & $0.13\%$ \\
    \hline
    & $2^{-100}$ & $10^4$ & $1$ & $56$ & $0.9821$ & $18.96\%$ \\
    & & $10^5$ & $1$ & $57$ & $0.9825$ & $2.90\%$ \\
    & & $10^6$ & $1$ & $57$ & $0.9825$ & $0.46\%$ \\
    & & $10^7$ & $1$ & $57$ & $0.9825$ & $0.13\%$ \\
    \hline
    & $2^{-200}$ & $10^4$ & $1$ & $62$ & $0.9677$ & $20.26\%$ \\
    & & $10^5$ & $1$ & $62$ & $0.9677$ & $3.01\%$ \\
    & & $10^6$ & $1$ & $62$ & $0.9677$ & $0.47\%$ \\
    & & $10^7$ & $1$ & $62$ & $0.9677$ & $0.13\%$ \\
    \hline
    $2^{-100}$ & $2^{-50}$ & $10^4$ & $2$ & $101$ & $1.0000$ & $35.55\%$ \\
    & & $10^5$ & $2$ & $101$ & $1.0000$ & $5.52\%$ \\
    & & $10^6$ & $2$ & $101$ & $1.0000$ & $0.89\%$ \\
    & & $10^7$ & $2$ & $101$ & $1.0000$ & $0.26\%$ \\
    \hline
    & $2^{-100}$ & $10^4$ & $2$ & $107$ & $0.9907$ & $36.85\%$ \\
    & & $10^5$ & $2$ & $107$ & $0.9907$ & $5.65\%$ \\
    & & $10^6$ & $2$ & $107$ & $0.9907$ & $0.91\%$ \\
    & & $10^7$ & $2$ & $108$ & $0.9907$ & $0.26\%$ \\
    \hline
    & $2^{-200}$ & $10^4$ & $2$ & $113$ & $0.9823$ & $38.19\%$ \\
    & & $10^5$ & $2$ & $113$ & $0.9823$ & $5.79\%$ \\
    & & $10^6$ & $2$ & $113$ & $0.9823$ & $0.92\%$ \\
    & & $10^7$ & $2$ & $113$ & $0.9823$ & $0.26\%$ \\
    \hline
    $2^{-200}$ & $2^{-50}$ & $10^4$ & $4$ & $201$ & $1.0000$ & $79.71\%$ \\
    & & $10^5$ & $4$ & $201$ & $1.0000$ & $11.91\%$ \\
    & & $10^6$ & $4$ & $201$ & $1.0000$ & $1.87\%$ \\
    & & $10^7$ & $4$ & $201$ & $1.0000$ & $0.52\%$ \\
    \hline
    & $2^{-100}$ & $10^4$ & $4$ & $208$ & $0.9952$ & $81.26\%$ \\
    & & $10^5$ & $4$ & $208$ & $0.9952$ & $12.06\%$ \\
    & & $10^6$ & $4$ & $208$ & $0.9952$ & $1.89\%$ \\
    & & $10^7$ & $4$ & $208$ & $0.9952$ & $0.52\%$ \\
    \hline
    & $2^{-200}$ & $10^4$ & $4$ & $215$ & $0.9907$ & $84.17\%$ \\
    & & $10^5$ & $4$ & $215$ & $0.9907$ & $12.35\%$ \\
    & & $10^6$ & $4$ & $215$ & $0.9907$ & $1.92\%$ \\
    & & $10^7$ & $4$ & $215$ & $0.9907$ & $0.52\%$ \\
    \hline

\end{tabular}
\caption{Statistical PINE Communication (verbose): overhead compared to non-robust sharing (see Table \ref{table:intro-performance}), different combinations of soundness error $\rho$ and zero-knowledge error $\delta$. Parameters: {\textbf{\color{blue} field size $q \approx 2^{64}$}}, $\alpha = 7.99996948$, $\eta = 2^{-91.33}$.}
\label{table:perf-eval-comb-errors-64-verbose}

\ifconf 
\end{table*} 
\else 
\end{table}
\fi

\ifconf
\begin{table*}[!t]
\else
\begin{table}[!t]
\fi

\centering

\begin{tabular}{|c | c | c || c | c | c | c |}
    \hline
    $\rho$ & $\delta$ & $d$ & $t$ & $r$ & $\tau$ & Overhead \\ [0.5ex]
    \hline\hline
    $2^{-50}$ & $2^{-50}$ & $10^4$ & $1$ & $50$ & $1.0000$ & $17.66\%$ \\
    & & $10^5$ & $1$ & $50$ & $1.0000$ & $2.75\%$ \\
    & & $10^6$ & $1$ & $50$ & $1.0000$ & $0.45\%$ \\
    & & $10^7$ & $1$ & $50$ & $1.0000$ & $0.13\%$ \\
    \hline
    & $2^{-100}$ & $10^4$ & $1$ & $56$ & $0.9821$ & $18.96\%$ \\
    & & $10^5$ & $1$ & $56$ & $0.9821$ & $2.88\%$ \\
    & & $10^6$ & $1$ & $56$ & $0.9821$ & $0.46\%$ \\
    & & $10^7$ & $1$ & $56$ & $0.9821$ & $0.13\%$ \\
    \hline
    & $2^{-200}$ & $10^4$ & $1$ & $61$ & $0.9672$ & $20.03\%$ \\
    & & $10^5$ & $1$ & $61$ & $0.9672$ & $2.99\%$ \\
    & & $10^6$ & $1$ & $61$ & $0.9672$ & $0.47\%$ \\
    & & $10^7$ & $1$ & $61$ & $0.9672$ & $0.13\%$ \\
    \hline
    $2^{-100}$ & $2^{-50}$ & $10^4$ & $1$ & $101$ & $1.0000$ & $28.69\%$ \\
    & & $10^5$ & $1$ & $101$ & $1.0000$ & $3.85\%$ \\
    & & $10^6$ & $1$ & $101$ & $1.0000$ & $0.56\%$ \\
    & & $10^7$ & $1$ & $101$ & $1.0000$ & $0.14\%$ \\
    \hline
    & $2^{-100}$ & $10^4$ & $1$ & $107$ & $0.9907$ & $29.97\%$ \\
    & & $10^5$ & $1$ & $107$ & $0.9907$ & $3.98\%$ \\
    & & $10^6$ & $1$ & $107$ & $0.9907$ & $0.57\%$ \\
    & & $10^7$ & $1$ & $107$ & $0.9907$ & $0.14\%$ \\
    \hline
    & $2^{-200}$ & $10^4$ & $1$ & $113$ & $0.9823$ & $31.27\%$ \\
    & & $10^5$ & $1$ & $113$ & $0.9823$ & $4.11\%$ \\
    & & $10^6$ & $1$ & $113$ & $0.9823$ & $0.58\%$ \\
    & & $10^7$ & $1$ & $113$ & $0.9823$ & $0.14\%$ \\
    \hline
    $2^{-200}$ & $2^{-50}$ & $10^4$ & $2$ & $201$ & $1.0000$ & $61.27\%$ \\
    & & $10^5$ & $2$ & $201$ & $1.0000$ & $8.09\%$ \\
    & & $10^6$ & $2$ & $201$ & $1.0000$ & $1.15\%$ \\
    & & $10^7$ & $2$ & $201$ & $1.0000$ & $0.28\%$ \\
    \hline
    & $2^{-100}$ & $10^4$ & $2$ & $208$ & $0.9952$ & $62.78\%$ \\
    & & $10^5$ & $2$ & $208$ & $0.9952$ & $8.25\%$ \\
    & & $10^6$ & $2$ & $208$ & $0.9952$ & $1.17\%$ \\
    & & $10^7$ & $2$ & $208$ & $0.9952$ & $0.28\%$ \\
    \hline
    & $2^{-200}$ & $10^4$ & $2$ & $215$ & $0.9907$ & $64.97\%$ \\
    & & $10^5$ & $2$ & $215$ & $0.9907$ & $8.46\%$ \\
    & & $10^6$ & $2$ & $215$ & $0.9907$ & $1.19\%$ \\
    & & $10^7$ & $2$ & $215$ & $0.9907$ & $0.28\%$ \\
    \hline

\end{tabular}
\caption{Statistical PINE Communication (verbose): overhead compared to non-robust sharing (see Table \ref{table:intro-performance}), different combinations of soundness error $\rho$ and zero-knowledge error $\delta$. Parameters: {\textbf{\color{blue} field size $q \approx 2^{128}$}}, $\alpha = 7.99996948$, $\eta = 2^{-91.33}$.}
\label{table:perf-eval-comb-errors-128-verbose}

\ifconf 
\end{table*} 
\else 
\end{table}
\fi

\fi
\end{document}